\documentclass[11pt]{article}

\usepackage{graphicx}
\usepackage{amsmath,amsfonts,amssymb}
\usepackage{amsthm}
\usepackage{parskip}
\usepackage{geometry}
\usepackage{authblk}
\usepackage[numbers,square]{natbib}
\usepackage{bm,cancel}
\usepackage{wrapfig, framed}
\usepackage{overpic}
\usepackage{comment}
\usepackage{mathtools, mathrsfs}
\usepackage{scalerel}
\usepackage{relsize}
\usepackage{graphbox}
\usepackage{stmaryrd}
\usepackage[T1]{fontenc}
\usepackage{enumitem}
\usepackage{tikz-cd}
\usepackage{tikz}
\usepackage{xfrac}
\usepackage{chngcntr}
\usepackage{hyperref}

\hypersetup{hidelinks}

\newenvironment{keywords}
  {\par\noindent\textbf{Keywords.}\enspace\ignorespaces}
  {\par}

\newcommand{\ifrac}[2]{\scalebox{0.85}{$ \dfrac{#1}{#2}$}}

\DeclareFontEncoding{LS1}{}{}
\DeclareFontSubstitution{LS1}{stix}{m}{n}
\newcommand*\kay{%
  \text{%
  \fontencoding{LS1}%
  \fontfamily{stixscr}%
  \fontseries{\textmathversion}%
  \fontshape{n}%
  \selectfont\symbol{"6B}}}
\makeatletter
  \newcommand*\textmathversion{\csname textmv@\math@version\endcsname}
  \newcommand*\textmv@normal{m}
  \newcommand*\textmv@bold{b}
\makeatother

\newcounter{descriptcount}
\newlist{enumdescript}{description}{1}
\setlist[enumdescript,1]{%
  before={\setcounter{descriptcount}{0}%
          \renewcommand*\thedescriptcount{\arabic{descriptcount}}},
        font={\bfseries\stepcounter{descriptcount}Aim \thedescriptcount:}
}

\theoremstyle{plain}
\newtheorem{theorem}{Theorem}[section]
\newtheorem*{theorem*}{Theorem}
\newtheorem*{proposition*}{Proposition}
\newtheorem{proposition}[theorem]{Proposition}
\newtheorem{corollary}[theorem]{Corollary}
\newtheorem{lemma}[theorem]{Lemma}

\theoremstyle{definition}
\newtheorem{definition}[theorem]{Definition}

\theoremstyle{remark}
\newtheorem{remark}[theorem]{Remark}

\newtheorem{exampleinner}{Example}[section]
\counterwithin{exampleinner}{section}
\newenvironment{example}
  {\begin{list}{}{\setlength{\leftmargin}{2em} \setlength{\rightmargin}{0em}}
   \item\relax\begin{exampleinner}}
  {\end{exampleinner}\end{list}}

\newcommand{\TightBracketB}[1]{\raisebox{0.00ex}{\scalebox{1}{${#1}$}}}%
\newcommand*{\eq}[1]{\TightBracketB[#1\TightBracketB]}

\newcommand{\R}{\mathbb{R}}
\newcommand{\ssubset}{\subset\joinrel \! \!\subset}
\newcommand{\B}{\mathcal{B}}

\newcommand{\ST}{\Lambda}
\newcommand{\HT}{\mathrm{HT}}

\newcommand{\C}{\mathcal{C} (\ST_i)}
\newcommand{\I}{\mathcal{I}(\ST)}
\newcommand{\h}{\mathcal{H}_D^2(\ST_i)}
\renewcommand{\hbar}{h}

\renewcommand{\j}{\text{\reflectbox{$\iota$}}}

\makeatletter
\renewcommand\maketitle{%
  {\raggedright
    \begin{center}
      {\vspace{-5em}\huge\@title}\\[2ex]
      {\Large\@author}\\[2ex]
      {\small Correspondence: \texttt{kooper.delacy@research.uwa.edu.au}}
    \end{center}}%
}
\makeatother

\title{A Shape-Variational Characterisation of Multi-region Relaxed Magnetohydrodynamic Equilibria}
\author[1]{K. de Lacy}
\author[1]{L. Noakes}
\author[1]{D. Pfefferl\'e}
\affil[1]{\small The University of Western Australia, 35 Stirling Highway, Crawley WA 6009, Australia}
\date{}

\begin{document}
	
\renewcommand{\cite}{\citep}

\maketitle

\newsavebox\mybox
\savebox{\mybox}{$\stretchrel*{|}{\rule[-.6ex]{0ex}{2.5ex}}$}
\def\myabs#1{\usebox\mybox#1\usebox\mybox}

\vspace{2em}
\begin{abstract}
\noindent
Let \(\ST\subset\mathbb{R}^3\) be a region admitting a partition into \(n\) compact, connected subregions \(\ST_1,\dots,\ST_n\), each with smooth boundary such that the interfaces between regions are closed and embedded in the interior \(\text{int}(\ST)\). Consider a vector field \(B\) on \(\ST\) where \(B|_{\ST_i}\) is smooth, divergence free, and tangent to \(\partial \ST_i\) for all \(i\). For each such partition we show that the multi-region relaxed magnetohydrodynamics (MRxMHD) equilibrium equations are necessary and sufficient conditions for the magnetic field, together with the placement of the interfaces, to be a stationary point of the magnetic energy under appropriate constraints. We prescribe the pressure, loop-corrected helicity, and magnetic flux of \(B\) through all smooth surfaces in \(\ST_i\) whose boundary lies on \(\partial \ST_i\).

\noindent
We isolate a gauge condition that is imposed implicitly in the existing toroidal treatments and show that the variation of the loop-corrected helicity cancels the corresponding boundary contribution without restricting the potentials. On toroidal domains this helicity agrees with the familiar loop-corrected expressions. For a single fixed region we give a condition on the Beltrami parameter of a region that is sufficient for a critical point of the magnetic energy to be the strict global minimiser of the energy among all admissible fields of that region. This is a condition for Beltrami fields on fixed domains that is applicable to individual regions \(\ST_i\) given a fixed partition and metric.
\end{abstract}


\section{Introduction}\label{section:Introduction}

Magnetohydrodynamics (MHD) equilibrium equations describe the steady-state behaviour of a globally neutral plasma under the influence of a magnetic field. Formally, MHD is derived from the Vlasov--Maxwell equations by assuming quasi-neutrality and taking low-order velocity moments to produce a single-fluid model coupled to the magnetic field. This model neglects effects like collisionless damping, viscosity, and high-frequency waves \cite{Freidberg_2014}. The MHD equilibrium equations characterise steady-state MHD solutions and take the form
\begin{align*}
  (\nabla \times B) \times B = \nabla p \, , \quad \nabla \cdot B = 0 \, , \quad N \cdot B \big|_{\partial f(\ST)} = 0 \, ,
\end{align*}
where the \emph{magnetic field} \(B\) is a smooth vector field on a compact, oriented domain \(f(\ST)\subset\mathbb{R}^3\) with smooth boundary \(\partial f(\ST)\). The \emph{pressure} \(p : f( \ST) \to \R\) is a smooth function that is constant on each connected component of the boundary \(\partial f(\ST)\).

We intend to discuss a variational problem where the domain varies. To encode the moving domain we hold \(\ST\) fixed and encode the variation in a diffeomorphism \(f\) of \(\ST\). The physical, or laboratory, configuration is then the image \(f(\ST)\) carrying the Euclidean metric, and it is there that we work throughout this section. Every differential operator, inner product, norm and outward unit normal \(N\) is therefore the usual Euclidean one on the subset \(f(\ST)\) of \(\mathbb{R}^3\). To discuss past results we work with \(f(\ST)\) and pull back to \(\ST\) in Section \ref{section:Problem Construction}.

The MHD equilibrium equations constitute a nonlinear boundary‐value problem. MHD equilibria are used in the early stages of designing and optimising magnetic field configurations for plasma confinement devices \cite{chen1984introduction, hazeltine2003plasma}. By exploring domain geometries one can identify promising reactor shapes and parameters before undertaking more realistic but computationally expensive simulations. Both iterative solvers and variational methods are in common use for computing or approximating MHD equilibria (assuming such solutions exist) with variational formulations typically exhibiting better convergence properties \cite{loizu2016verification, REIMAN1986157, hirshman2011siesta, helander2013ideal, dudt2020desc, hirshman1986three, chodura19813d, hayashi1995evolution}. Restricting to axisymmetric fields that are invariant under continuous rotations about the toroidal axis yields the 2D elliptic Grad--Shafranov equation, for which many numerical solvers exist \cite{PATAKI201328, LEE201572, zakharov1999theory, takeda1991computation}.

This work presents a variational formulation of the MRxMHD equilibrium equations on relatively general domains \(\ST\). Specifically, we show that a variational method developed for nested toroidal domains \cite{dennis2014multi}---which can approximate MHD equilibria under certain assumptions \cite{10.1063/1.4795739}---can be extended to more general domains \(\ST\). To properly discuss the motivating ideas we introduce the necessary notation and provide a brief overview of the relevant literature. 

\subsection{Construction}\label{sub:overview}

Consider a partition of \(\ST\) into \(n<\infty\) compact, connected subregions \(\ST_i\) (\(i=1,\dots,n\)), each with smooth boundary, so that
\[
    \bigcup_{i=1}^n \ST_i = \ST \, ,
    \qquad \text{and} \qquad
    \mathrm{int}(\ST_i)\cap \mathrm{int}(\ST_j)=\emptyset
    \quad (i\neq j) \, .
\]
We further assume that no point of \(\ST\) lies in more than two of the subregions. Consequently, whenever \(\ST_i \cap \ST_j \neq \emptyset\) with \(i \neq j\), this intersection is shared by the two regions \(\ST_i\) and \(\ST_j\) alone. These intersections are called \textit{interfaces} and we assume that each is a smooth, closed, connected surface embedded in \(\mathrm{int}(\ST)\). Since \(n\) is finite, there are finitely many pairwise distinct interfaces, which we label \(I_q\) for \(q = 1, \dots, m\).

Let \(\I\) be the set of embeddings \(f : \ST \to f(\ST)\) for which there exists a smooth isotopy \(f_t : \ST \to f_t(\ST)\), \(t \in [0,1]\), through orientation preserving diffeomorphisms with \(f_t|_{\partial \ST} = \mathrm{id}\), \(f_1 = \mathrm{id}\) and \(f_0 = f\). The family \(\{f_t(\ST_i)\}_{i=1}^n\) is then a smooth deformation of the partition \(\{\ST_i\}_{i=1}^n\) into \(\{f(\ST_i)\}_{i=1}^n\), the subregions having disjoint interiors and covering \(f(\ST)\) for every \(t\).

For a given \(I_q\) there exists, by definition, an \(i\) and \(j\) where \(i > j\) such that \(I_q = \ST_i \cap \ST_j\), and hence \(f(I_q) = f(\ST_i) \cap f(\ST_j)\). For fields \(V_i\) defined on \(f(\ST_i)\) and \(V_j\) defined on \(f(\ST_j)\), we define the jump across \(f(I_q)\) by
\[
    \llbracket V \rrbracket \big|_{f(I_q)}
    = \llbracket V \rrbracket \big|_{f(\ST_i) \cap f(\ST_j)}
    \coloneqq V_i \big|_{f(I_q)} - V_j \big|_{f(I_q)} \, .
\]

Since \(f|_{\partial \ST} = \mathrm{id}\) we have \(f(\ST) = \ST\) as subsets of \(\mathbb{R}^3\), so that the map \(f\) controls the placement of the interfaces \(f(I_q)\). Each \(f(\ST_i)\) is a compact, connected, oriented region of \(\mathbb{R}^3\) with smooth boundary, carrying the Euclidean metric with volume form \(\varpi\). We write
\[
    \myabs{f(\ST_i)} \coloneqq \int_{f(\ST_i)} \varpi \, ,
\]
for the Euclidean volume of \(f(\ST_i)\).


From this point onward, we permit \(B\) to be discontinuous across partition boundaries, requiring only that its restriction to each \(\mathrm{int}( f(\ST_i) )\) is smooth. We define a smooth field \(B_i \) on \(f(\ST_i)\) that matches \(B\) on \(\mathrm{int}(f(\ST_i))\). Similarly, we allow the pressure \(p\) to be discontinuous across partition boundaries, and define a smooth \(p_i \in C^{\infty} ( f(\ST_i) )\) that matches \(p\) on \(\mathrm{int}(f(\ST_i))\). On each domain \(f(\ST_i)\), \(p_i\) is identified with its image, namely a constant. Both \(B\) and \(p\) may therefore jump across the interfaces \(f(I_q)\).

The computations of Sections \ref{section:Preparatory Results} and \ref{section:Main Proofs} are carried out in the language of exterior calculus, which is better adapted to the change of variables induced by \(f\). The change together with pulling back to a fixed reference domain is given once in Section \ref{section:Problem Construction}. Until then, we work with vector fields on the regions \(f(\ST_i)\), as in the references cited below.


In each region a magnetic field \(B_i\) is a smooth divergence free vector field tangent to the boundary \(\partial f(\ST_i)\). Namely,
\[
    \nabla\cdot B_i = 0 \ \text{ in } f(\ST_i) \, , \qquad B_i \cdot N = 0 \ \text{ on } \partial f(\ST_i) \, .
\]
We define the set of such \(B_i \) as \(C( f(\ST_i) )\). Interfaces are described as \emph{ideal barriers}, meaning that no field line from any \(B_i\) crosses them.

The multi-region relaxed magnetohydrodynamics (MRxMHD) equilibrium equations are
\begin{align}\label{prob:MRxMHD}
    \nabla\times B_i = \mu_i B_i \ \text{ for } B_i \in C ( f(\ST_i) ) \, ,
    \qquad
    \Big[ \! \!\Big[ \ifrac12 \|B \|^2 + p  \Big] \! \!\Big] \Big|_{f(I_q)} = 0 \, ,
\end{align}
where \(p_i\) maps \(f(\ST_i)\) to a constant. The first of these conditions is a Beltrami equation for \(B_i\) in each \(f(\ST_i)\). Physicists refer to \(B_i\) as a linear force-free field. The quantity \( ( \sfrac12 )\|B_i \|^2  + p\) is referred to as the \emph{total pressure}, so the latter condition in the MRxMHD equations enforces continuity of the total pressure. This is the form given in \citet{10.1063/1.4765691}. The constant \(\mu_i\) arises as a Lagrange multiplier and we refer to it as such. 

For smooth vector fields \(V_1, V_2 \) on \(f(\ST)\) the \(L^2\) inner product and norm are given by:
\begin{align*}
\begin{gathered}
  \langle V_1, V_2 \rangle_{L^2} \coloneqq \int_{f(\ST)} \langle V_1 ,V_2 \rangle \,  \varpi = \int_{f(\ST)} V_1 \cdot V_2 \, \varpi \, , \qquad
  \| V_1 \|_{L^2}^2 \coloneqq \langle V_1, V_1 \rangle_{L^2} \, ,
\end{gathered}
\end{align*}
where \(\langle V_1, V_2 \rangle\) denotes the Euclidean pointwise inner product. For a field defined only on a subregion the integral is taken over that subregion.

The canonical \emph{total energy} in units with unit vacuum permeability \((\mu_0 = 1)\) is defined as
\begin{align}\label{eq:original_energy}
    W ( f, B_1, \dots, B_n, p_1, \dots, p_n) \coloneqq \sum_{i=1}^n \int_{f(\ST_i)}  \Big( \frac{1}{2} \big\| B_{i}\big\|^2 +  \frac{p_i}{\gamma - 1} \Big) \varpi \, ,
\end{align}
where \(\gamma > 1\) is the ratio of specific heats. Here \(f\) enters through the regions of integration, and hence through the domains of the \(B_i\) and \(p_i\). This is the total energy appearing in \citet{10.1063/1.4765691}.


For a fixed \(i\), \(f(\ST_i)\) is a compact, oriented subset of \(\R^3\) with a smooth boundary, so any \(B_i \in C( f(\ST_i) )\) is the curl of a smooth vector field \(A_i\) on \(f(\ST_i)\) (namely \(B_i = \nabla \times A_i\)). This is shown by Cantarella and Parsley (\citeyear{CANTARELLA20101127}, Proposition C.4). As we use this result throughout the paper, an exposition of this proof can be found in the appendix, Lemma \ref{Lemma:Potential}. The field \(A_i\) is often called a vector potential, potential, or primitive of \(B_i\).


Consider a subset \(G_i \) of smooth vector fields on \(f(\ST_i)\) such that for every \(B_i \in C( f(\ST_i) )\) there exists at least one \(A_i \in G_i\) such that \(B_i = \nabla \times A_i \). The \emph{helicity} of \(B_i \) on \(f(\ST_i)\) with respect to \(G_i\) is defined by
\begin{align*}
  H_{G_i}(B_i , f(\ST_i) ) \coloneqq \int_{f(\ST_i)} A_i \cdot B_i \, \varpi \, ,
\end{align*}
where \(A_i \in G_i\) and \(B_i = \nabla \times A_i \). Without a restriction on the admissible potentials, the helicity need not be well-defined. Precisely, if the flux of \( B_i \) through some smooth surface in \(f(\ST_i)\) whose boundary lies in \(\partial f(\ST_i) \) is non-trivial, then different choices of potential \(A\) lead to arbitrarily different values of the helicity (see Subsection \ref{sub:Well-Posed}). One approach, not adopted here but common in the literature, is to use a Biot--Savart type operator to select a gauge in which helicity becomes a well-defined function of the magnetic field and the underlying geometry \cite{parsley2001taylor, Arnold_Topological_Methods, soperimetric_problems_helicity, CANTARELLA20101127, MacTaggart_2023}.

A well-defined helicity requires an appropriate choice of \(G_i\) within each \(f(\ST_i)\). These additional choices can make a variational problem with a helicity constraint somewhat cumbersome. And we are interested in the corresponding conditions on the gauge that will result in reasonable Euler--Lagrange equations. A related construction is the conventional \emph{relative helicity}
\[
    \int_{ f(\ST_i) } ( A_i + A_{\text{ref}} ) \cdot ( B_i - B_{\text{ref}} ) \, \varpi \, ,
\]
taken with respect to a reference field \(B_{\text{ref}} = \nabla \times A_{\text{ref}}\) with the same normal component on the boundary. On multiply connected domains, gauge independence under all choices of potentials additionally requires matching cut fluxes or the corresponding loop corrections \cite{MacTaggart_2023}. For the boundary-tangent fields considered here, we instead use the \emph{loop-corrected helicity}. It can be written using a smooth curl-free field \(A_{\text{ref}}\) whose loop integrals match the free loop integrals of \(A_i\), as specified in Definition \ref{def:a_ref}. This gives
\[
    \mathscr{H}_{A_i} ( B_i, f(\ST_i) ) \coloneqq \int_{ f(\ST_i) } ( A_i - A_{\text{ref}} ) \cdot  B_i \, \varpi \, .
\]
We will drop the explicit \(A_i\) dependence after proving that this quantity is gauge invariant.


The physical motivation behind a variational formulation of MRxMHD is \emph{Taylor relaxation} \cite{taylor1974relaxation}. In a plasma with small but non-zero resistivity the magnetic energy is dissipated faster than magnetic helicity. On timescales that are long compared with the energy decay time, but short compared with the helicity decay time a confined plasma is expected to settle into a state of minimum magnetic energy at fixed helicity and fixed flux. This is called a \emph{Taylor state}, which Woltjer identified as a force-free Beltrami field \cite{doi:10.1073/pnas.44.6.489}. MRxMHD is motivated by this relaxation occurring independently within each subregion \(f(\ST_i)\), while the interfaces \(f(I_q)\) are treated as ideal barriers that the field cannot cross (\(B_i \cdot N = 0\) on \(\partial f(\ST_i)\)), their placement being controlled by \(f\).

Accordingly, one seeks minimisers of the total energy \(W\) subject to constraints on the helicity and the flux in each region, together with a thermodynamic constraint that closes the system in the pressure. We take as our starting point, up to the geometry, the formulation of Hudson, Dewar, Dennis, Hole, McGann, von Nessi and Lazerson \citeyearpar{10.1063/1.4765691}, in which the pressure is constant on each region and obeys an adiabatic law,
\begin{align}\label{prob:MRxMHD_adiabatic} 
\begin{gathered}
    \underset{\substack{f \in \I \\ A_i \in G_i \\ p_i \in C^\infty ( f(\ST_i) ) } }{\min}
    \; \sum_{i=1}^n \int_{f(\ST_i)}  \Big( \frac{1}{2} \big\| B_{i}\big\|^2 +  \frac{p_i}{\gamma - 1} \Big) \varpi \, , \\
    H_{G_i} ( B_i , f(\ST_i) )  = h_i  \, , \qquad
    \int_{ f(T_{i,j}) } B_i \cdot dS = \psi_{i,j} \, , \qquad
    p_i \myabs{ f(\ST_i) }^\gamma = v_i \, .
\end{gathered}
\end{align}
Note that Hudson et al. provide a formulation directly in terms of variations of the interface rather than through \(f\). These formulations are equivalent (Appendix \ref{section:energy}). 

Here the constants \(v_i\), \(h_i\) and \(\psi_{i,j}\) are prescribed for \(i = 1, \dots, n\) and \(j = 1 , \dots, \ell_i\) where \(\ell_i \coloneqq \dim H_1(\ST_i)\) is the genus of \(\partial \ST_i\), equivalently of \(\partial f(\ST_i)\). This problem assumes that the flux of \(B_i\) is prescribed through each surface \( f(T_{i,j})\) with \(\partial f( T_{i,j} ) \subseteq \partial f(\ST_i)\). The admissible surfaces \(T_{i,j}\) are described in Subsection \ref{sub:flux}, and are made explicit for toroidal regions in Subsection \ref{sub:back}. Additionally, it is assumed that each \(G_i\) is a set of potentials on \(f(\ST_i)\) that renders the helicity a well-defined map of \(( B_i , f(\ST_i) )\). Since the elements of \(G_i\) are fields on \(f(\ST_i)\), the admissible set of potentials moves with \(f\); the transcription of \(G_i\) onto the fixed reference domain is given in Subsection \ref{sub:setup}.


Different choices of \(G_i\) give genuinely different variational problems, since \(G_i\) fixes the gauge and with it the constraint surface. One of the main results of this paper isolates a condition, equation \eqref{eq:gauge_condition}, under which the critical points of \eqref{prob:MRxMHD_adiabatic} are precisely the solutions of the MRxMHD equations \eqref{prob:MRxMHD} with a non-degeneracy assumption. The condition is sufficient to remove a residual term in the first variation; if it fails, that term may modify the Euler--Lagrange equations.

The helicity and flux constraints appearing in \eqref{prob:MRxMHD_adiabatic} are topological in nature. Moffatt showed that in a particular gauge the helicity measures the linking of field lines, namely, for a field supported on a collection of flux tubes it is the sum of the Gauss linking numbers of those tubes, weighted by the product of their fluxes \cite{moffatt1969degree}. This interpretation was extended to fields whose field lines might not be closed \cite{Arnold_Topological_Methods,berger1984topological}.

Additionally, we will show that the flux of \(B_i\) through a surface \(T\) with \(\partial T \subseteq \partial f(\ST_i)\) depends only on the second relative homology class of \(T\) and the flux constraints in \eqref{prob:MRxMHD_adiabatic}. The application of homology to magnetic fields on such regions is not new \cite{blank1957notes}.

We now switch from a minimisation problem that is motivated by Taylor relaxation to a problem to determine critical points (c.\ p.). The reason for this is that our goal is to find a variational problem that aligns exactly with the MRxMHD problem, and looking for critical points is the correct setup for this.

Problem \eqref{prob:MRxMHD_adiabatic} is not the only way that the MRxMHD problem has been posed in the literature. Dennis, Hudson, Dewar and Hole \citeyearpar{10.1063/1.4795739} prescribe the plasma entropy of each region instead of the adiabatic constraint in \eqref{prob:MRxMHD_adiabatic}. Another setup was given by Dewar, Yoshida, Bhattacharjee and Hudson \citeyearpar{Dewar_Yoshida_Bhattacharjee_Hudson_2015}. These two variational problems are posed on toroidal domains.

In Appendix \ref{section:energy} we record the known fact that these formulations have related critical points. Proposition \ref{prop:energy_comparison} shows that the solution set of \eqref{prob:MRxMHD_adiabatic} with data \(v_i\) is the same as replacing the objective with 
\[
\sum_{i=1}^n \Big( \ifrac{1}{2} \| B_i \|_{L^2}^2 - p_i \myabs{f(\ST_i)} \Big) \, ,
\]
using \(p_i = v_i \myabs{ f(\ST_i) }^{-\gamma}\), and similarly for the other two formulations. Passing between them amounts to a relabelling of the data at each critical configuration. At a critical point of these problems the pressure is automatically constant on each region, so mathematically nothing is lost by prescribing constant pressure within each \(f(\ST_i)\).


We take this a step further and use a gauge invariant constraint given by the loop-corrected helicity \(\mathscr  H_{A_i}\) we show that this matches the usual helicity constraint in a particular gauge (see Section \ref{section:Preparatory Results}),
\begin{align}\label{prob:MRxMHD_fixed}
\begin{gathered}
    \underset{\substack{f \in \I \\ A_i  } }{ \text{c. p.} }
    \; \sum_{i=1}^n \Big( \ifrac{1}{2} \| B_i \|_{L^2}^2 - p_i \myabs{f(\ST_i)} \Big) \, , \\
    \mathscr  H_{A_i} ( B_i , f(\ST_i) ) = h_i  \, , \qquad
    \int_{f(T_{i,j})} B_i \cdot dS = \psi_{i,j} \, ,
\end{gathered}
\end{align}
where \(A_i\) is any smooth vector field and \(B_i = \nabla \times A_i\) such that \(B_i \in C( f( \ST_i ))\).

The use of MRxMHD equilibria for magnetic confinement is discussed by Hudson and Kraus \citeyearpar{hudson2017three}. The authors state that in a stepped-pressure equilibrium the change in pressure is supported entirely on the interfaces, each of which is a flux surface. The pressure difference \(\llbracket p \rrbracket\) is opposed by a jump in the magnetic field strength, as given by the MRxMHD jump condition. Confinement in such a model is therefore carried by finitely many surfaces rather than by a continuum of nested flux surfaces. As the number of interfaces increases the stepped profile can approximate a smooth one if a corresponding ideal MHD equilibrium exists \cite{10.1063/1.4795739}. If MRxMHD equilibria approximate MHD equilibria, it is a natural framework in which to study magnetic confinement.

\subsection{Background: Toroidal Geometry}\label{sub:back}

\setlength{\FrameSep}{3pt}

\begin{wrapfigure}{r}{0.3\linewidth}
    \vspace{-3.0em}
  \begin{framed}\raggedleft
  \begin{overpic}[width=\linewidth, clip]{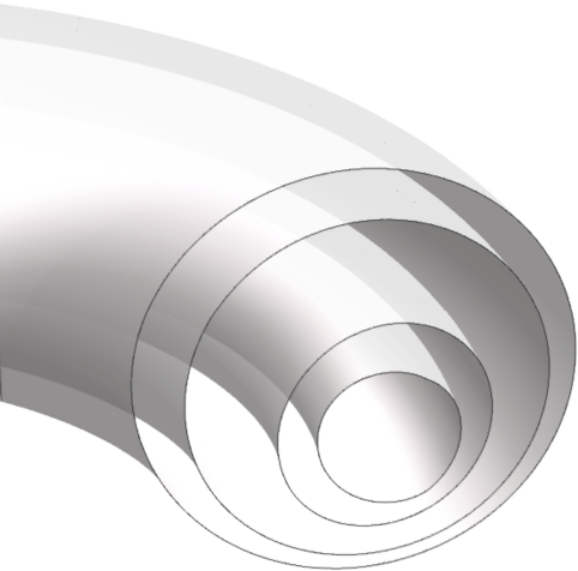}
    \put(57,23){\(f(\ST_1)\)}
    \put(57,36.5){\(f(\ST_2)\)}
    \put(57,49.75){\(f(\ST_3)\)}
    \put(57,63){\(f(\ST_4)\)}
  \end{overpic}
  \caption{A cross-section and partition of a solid torus embedded in \(\R^3\) into 4 domains.}\label{fig:torus}
  \end{framed}
  \vspace{-1em}
\end{wrapfigure}

In this subsection we restrict ourselves to a special geometric setting and review previous work.


Consider a smooth pressure \(p\) that is non-degenerate (\(\nabla p\neq0\)) on \(\partial f(\ST)\). As an MHD solution satisfies \(B \cdot \nabla p = 0\), then \(B\) is tangential to pressure level sets satisfying \(\nabla p\neq0\).

On level sets of \(p\) in \(f(\ST)\) with \(\nabla p \neq 0\) everywhere, we have \(\nabla p = (\nabla \times B) \times B\), so \(B\) and \(\nabla \times B\) are both non-zero and \(B\) is nowhere parallel to \(\nabla \times B\). The Poincar\'e--Hopf theorem implies that any compact manifold admitting a nowhere-zero continuous tangent vector field has Euler characteristic zero. Every compact, connected, oriented, closed smooth surface with Euler characteristic zero is diffeomorphic to a 2-torus \cite{cantarella2002vector}. Consequently, if the pressure is non-degenerate then each level set is a disjoint union of 2-tori.

In this subsection we consider regions \(f(\ST_i)\) for \(i = 1, \dots, n\) that are solid and hollow tori embedded in \(\R^3\), as in Figure \ref{fig:torus} with \(n=4\). This implies that interfaces are given by \(f(I_q) = \partial \cup_{i=1}^q f(\ST_i) = \partial f(\ST_{q}) \cap \partial f(\ST_{q+1})  \) for \(q=1,2, \dots, n-1\), which are pairwise-disjoint toroidal surfaces.

Consider Problem \eqref{prob:MRxMHD_adiabatic} for fixed sets \(G_i\) that produce well-defined helicities. The objective of this problem is to determine constrained minimisers of the magnetic energy \((B_1, \dots, B_n, p_1, \dots, p_n, f)\).

For Problem \eqref{prob:MRxMHD_adiabatic} with a toroidal geometry and nested toroidal interfaces, \(f(T_{i,1})\) is a poloidal cut of \(f(\ST_i)\) for all \(i = 1,\dots,n\) (naming \(f(T_{i,1})\) as in \citealt{Gauge_freedom}). Namely, \(f(T_{i,1})\) is a disc for \(i=1\) and an annulus for \(i>1\). For \(i = 2,\dots,n\), \(f(T_{i,2})\) is a toroidal annulus (or ribbon) spanning the two boundary tori of the hollow torus \(f(\ST_i)\). No such surface exists for \(i=1\), whose boundary is connected. Accordingly, \(\psi_{i,1}\) is the toroidal flux and \(\psi_{i,2}\) the poloidal flux.


The MRxMHD equilibrium equations \eqref{prob:MRxMHD} are derived from \eqref{prob:MRxMHD_adiabatic} in a fixed gauge by a Lagrange multiplier argument \cite{10.1063/1.4765691}. For a fixed geometry and varying the potential \(A_i\) the resulting stationary condition is the Beltrami equation \(\nabla \times B_i = \mu_i B_i\) with \(\mu_i\) the Lagrange multiplier of the helicity constraint in \(f(\ST_i)\). This was also done in isolation by Woltjer \citeyearpar{doi:10.1073/pnas.44.6.489}. Varying the interfaces, that is varying \(f\), gives the continuity condition on the total pressure.

Similar to what we have done in Problem \eqref{prob:MRxMHD_fixed}, some authors replace the helicity constraint in \eqref{prob:MRxMHD_adiabatic} by a gauge-independent constraint given by the relative helicity or a loop-corrected helicity \cite{spies2001relaxed,kaiser2004relaxed}. This loop-corrected helicity is defined by
\begin{equation}\label{eq:relhel-vector-intro}
\int_{f(\ST_i)}A_i \cdot B_i\, \varpi
-\sum_{j=1}^{\ell_i}\psi_{i,j}\oint_{f(\sigma_{i,j})} A_i\cdot d\ell \, ,
\end{equation}
where \(d\ell\) is the Euclidean length element. The loop-corrected helicity is the helicity corrected by the circulation of $A_i$ around each of a finite collection of loops $f(\sigma_{i,j})$ on $\partial f(\ST_i)$, weighted by the flux $\psi_{i,j}$. This is equal to \(\mathscr{H}_{A_i}(B_i, f(\ST_i))\) of Subsection \ref{sub:overview} (see Example \ref{ex:toroidal_helicity}), and is the zero-reference-field case of the general relative helicity with loop corrections \cite{MacTaggart_2023}. In the toroidal setting of this subsection these loops are a toroidal loop \(f(\sigma_{i,1})\) for all \(i\) and a poloidal loop \(f(\sigma_{i,2})\) for \(i>1\), dual to \(T_{i,1}\) and \(T_{i,2}\) respectively. The relation between these loops, the corresponding orientation, and the fluxes is provided explicitly on a hollow torus in Appendix \ref{append:example}, and the general relation is in Subsection \ref{sub:coho_homo}.



The Stepped Pressure Equilibrium Code (SPEC) is a numerical solver for the MRxMHD variational problem \cite{10.1063/1.4765691}. Convergence studies suggest that its solutions converge to the MRxMHD equilibrium \cite{loizu2016verification}. SPEC can compute MRxMHD equilibria that incorporate cross-helicity conditions, field-aligned flow and angular-momentum constraints \cite{qu2020stepped}, and has been used to generate and analyse stellarator configurations \cite{loizu2016verification}. In Subsection \ref{subsec:spec} we show that the gauge used in SPEC satisfies the conditions that we obtain (condition \eqref{eq:gauge_condition}) that ensure the corresponding variational problem has an equivalent solution set to MRxMHD.


Numerical studies using SPEC provide evidence that solutions to the MRxMHD variational problem can approximate MHD equilibria. However, because the existence of MHD equilibria is not presently known in general, there is no a priori guarantee that one can choose constants \(\psi_{i,j}, h_i, p_i\) for the MRxMHD problem which correspond to an MHD equilibrium \cite{10.1063/1.4795739}. More generally, passing to the infinite-interface limit in the MRxMHD variational problem then yields Euler--Lagrange conditions that coincide with the ideal MHD equilibrium equations \cite{10.1063/1.4795739}.

In this subsection the relevant classes were exhibited by hand. The surfaces \(f(T_{i,1})\) and \(f(T_{i,2})\) are the poloidal and toroidal cross-sections, and the loops \(f(\sigma_{i,j})\) on \(\partial f(\ST_i)\) are the corresponding toroidal and poloidal curves. On the general regions considered in this paper no such explicit description is available so we work with the homology and cohomology of \(\ST_i\) directly (see Subsection \ref{sub:coho_homo}).

\vspace{2em}

\subsection{Background: Existence of MRxMHD Equilibria}

A longstanding conjecture on the existence of MHD equilibria was proposed by Grad \citeyearpar{grad_conjecture}. Grad conjectured that, under suitable regularity assumptions, any MHD equilibrium in \(\R^3\) must satisfy one of two alternatives: either the pressure is constant, or the magnetic field is invariant under a nontrivial one-parameter group of Euclidean isometries.

An early study on the existence of MHD equilibria with stepped pressure, namely MRxMHD equilibria, was conducted by \citet{bruno_laurence_1996} who showed the existence of solutions on domains that are a perturbation away from axisymmetry. For context, the MRxMHD solutions are region-wise MHD solutions, each with constant pressure. Hence, this result from Bruno and Laurence does not answer Grad's conjecture. In Section \ref{section:Preparatory Results} we show an equivalence between the boundary conditions they provide and the jump condition given in the MRxMHD equilibrium equations when \(\llbracket p \rrbracket \neq 0\).

More recently, Enciso, Luque, and Peralta-Salas demonstrate the existence of solutions to \eqref{prob:MRxMHD} on nondegenerate, analytic, toroidal domains using a Cauchy--Kovalevskaya theorem. They also show that these results are generic \cite{enciso2023mhd}.

The existence of energy minimisers for \(n=1\) on any given domain \(\ST\), which may not be toroidal, has been studied. When the magnetic field is flux-free, existence is guaranteed for any helicity \cite{exist_gerner}. Similarly,  \citet{laurence1991woltjer} provide the corresponding Euler--Lagrange equations and an existence proof. It was also shown by Parsley that if there exist magnetic fields with a given flux, helicity, and gauge then there exists a corresponding energy minimiser \citep[Theorem 6.1]{parsley2001taylor}. Similarly, a study by  \citet{yoshida2012helical} showed existence of solutions to the Beltrami equations with helicity and flux constraints. We provide a weak existence result with interior regularity, using the direct method and a bootstrapping technique, in Appendix~\ref{append:exist}.

\begin{figure} 
  \centering
  \begin{overpic}[width=0.5\linewidth]{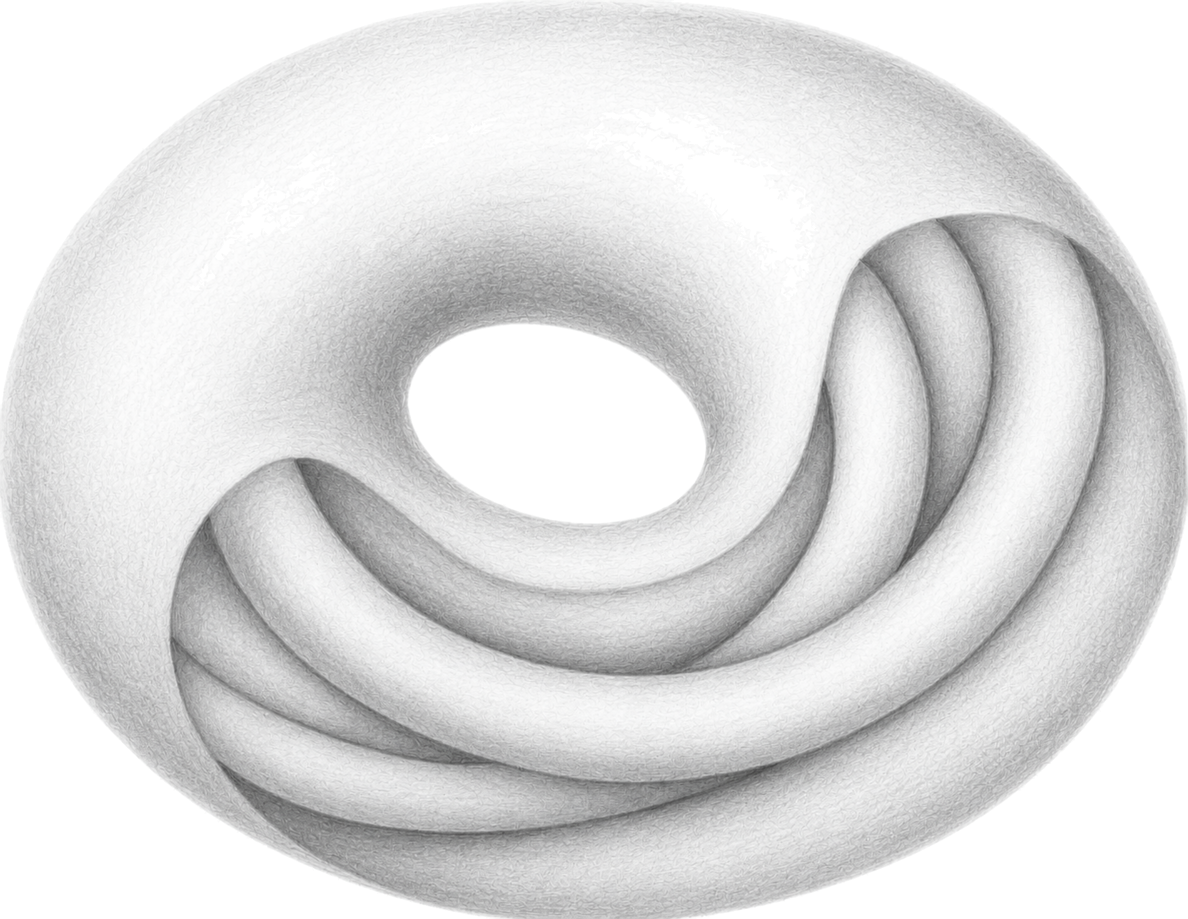}
  \end{overpic}
  \caption{A volume with 3 toroidal boundaries. In particular, the boundary of this region is a union of 3 pairwise distinct tori. The outer boundary contains 2 inner tori that are linked.}\label{fig:topology} 
\end{figure}

\subsection{Our Contribution}\label{sub:contribution}

The results reviewed above are stated on nested toroidal regions, concern minimisers, and are derived in a fixed gauge. We return to the context of Problem \eqref{prob:MRxMHD_fixed} with critical points, partitions \(\ST_i\) of a domain \(\ST\), and with a gauge invariant helicity.

Our principal result, Theorem \ref{prop:vanishing}, is that the MRxMHD equations \eqref{prob:MRxMHD} are necessary and sufficient conditions for a configuration to be a critical point of \eqref{prob:MRxMHD_fixed}. In particular, we replace the nested toroidal interfaces assumption by the substantially weaker assumptions on \(\ST_i\) and the interfaces. The variational characterisation of MRxMHD equilibria is therefore available on a much broader class of geometries than previously considered (see Figure \ref{fig:topology}). This provides a framework to utilise MRxMHD for the design of stellarators with pressure level sets that have more complex topology than previously considered.

The variation of a loop-corrected helicity constraint supplies exactly the term required to cancel the boundary contributions that arise when the interfaces are varied, and the MRxMHD jump condition follows. If instead the constraint \(H_{G_i}(B_i, f(\ST_i)) = h_i\) is imposed in a prescribed gauge (assuming that this constraint is not abnormal, that is, that its derivative is non-degenerate) a residual term may remain in the first variation, and the jump condition is recovered whenever that residual term vanishes. Common gauges, such as the Amperian gauge defined in Subsection \ref{subsection:Relative Helicity}, or the gauge used in SPEC do satisfy this condition (verified in Subsection \ref{subsec:spec}).

Dewar, Yoshida, Bhattacharjee and Hudson \citeyearpar{Dewar_Yoshida_Bhattacharjee_Hudson_2015} impose on the interfaces the holonomic constraint on the variation of the potential. The authors show that the line integrals of $A$ around loops on an interface then have zero variation as the interface is displaced. That constraint is stated as the interface is displaced, that is, under variation of $f$. This restriction satisfies our corresponding gauge condition, in the toroidal case. What we add is to isolate the condition on an arbitrary partition, where the relevant loops are the generators of an Alexander basis rather than a distinguished poloidal and toroidal pair. The variation of the loop-corrected helicity cancels the corresponding boundary contribution, so no restriction on the admissible potentials is then required.

Combined with Proposition \ref{prop:energy_comparison}, Proposition \ref{thr:iff_simple} shows that the critical points of \eqref{prob:MRxMHD_adiabatic}, and of the entropy-constrained formulations of \citet{10.1063/1.4795739} and \citet{Dewar_Yoshida_Bhattacharjee_Hudson_2015}, with positive pressures, are the solutions of the MRxMHD equilibrium equations \eqref{prob:MRxMHD}.

When $n=1$ and \(f = \text{id}\) there are no interfaces and the geometry does not vary. In this case we give a sufficient condition, namely, a bound on $|\mu_1|$ in terms of the spectrum of the curl operator on the flux-free fields of $\ST$, under which a magnetic field at fixed geometry is a minimiser of the magnetic energy among all admissible fields on $\ST$. As this result is at fixed geometry it applies only to Beltrami fields, and by the decoupling of Subsection \ref{sec:main_theorem} it applies to each region of a general partition separately.

\subsection{Paper Layout}

In Section \ref{section:Problem Construction} we introduce the definitions required to state our principal theorems, pull the problem back onto the fixed reference domain \(\ST\), and convert the vector field formulation above into the language of exterior calculus. Section \ref{section:Main Results} then states and discusses these theorems. Our formulation is expressed in terms of a loop-corrected helicity which, under a suitable gauge choice, reduces to the conventional helicity.

Section \ref{section:Preparatory Results} provides useful results and a reformulation of our variational problem constraining loop-corrected helicity. Section \ref{section:Main Proofs} presents proofs for the main results in Section \ref{section:Main Results}. Section \ref{section:Characterisation and Local uniqueness} gives a sufficient condition, at fixed geometry and \(n=1\), for a critical point of the magnetic energy in a region to be a minimiser of the energy over the admissible fields of that region.

Appendix \ref{section:Homology cohomology construction} provides definitions related to cohomology and homology. Appendix \ref{section:energy} compares the variational formulations of the MRxMHD problem appearing in the literature and shows that they share the critical points of \eqref{prob:MRxMHD_fixed}. Appendix \ref{section:Alexander Basis} provides a known construction of an Alexander basis in the language of exterior calculus. Appendix \ref{append:exist} provides an existence proof for Beltrami fields. Appendix \ref{append:vector_notation} provides a dictionary to convert results between exterior calculus and vector calculus.

\section{Notation}\label{section:Problem Construction}

\subsection{Exterior Calculus}\label{sub:setup}

The definitions of Section \ref{section:Introduction} were given in vector calculus so that they could be compared directly with the literature. The remainder of the paper is written in exterior calculus, which is better adapted to the change of variables induced by \(f \in \mathcal{I}(\ST)\), and to the boundary and cohomological conditions that appear below. We transcribe once, here, and work in the transcribed language thereafter unless we are comparing with results in the literature. The translation back into vector notation is collected in Appendix \ref{append:vector_notation} so that the reader need not hold both notations in mind at any point.

Problems \eqref{prob:MRxMHD_adiabatic} and \eqref{prob:MRxMHD_fixed} may equally be posed on the reference domain \(\ST\). Pulling back the Euclidean metric on \(f(\ST_i)\) to \(\ST_i\) pulls back the volume form \(f^* \varpi\). We set
\[
    \tilde p_i \coloneqq f^{*} p_i \coloneqq p_i \circ f \, , \qquad
    f^{*} V \coloneqq (df)^{-1} \circ V \circ f \, , \qquad
    f^{*} G_i \coloneqq \{ f^{*} A : A \in G_i \}  \, ,
\]
where \(V\) is a smooth vector field on \(f(\ST_i)\). The pressures \(p_i\) of \eqref{prob:MRxMHD_fixed} are constants, so \(\tilde p_i\) and \(p_i\) agree as real numbers and the tilde records only the region on which the function is defined. Under this correspondence the objective and the flux agree because \(f\) is an isometry, the adiabatic constraint is unchanged, and the helicity agrees by the invariance
\[
    H_{f^{*} G_i} (  f^{*} B_i , \ST_i ) = H_{G_i} ( B_i , f(\ST_i) ) \, ,
\]
shown by \citet{CANTARELLA20101127}. The same is true of the loop-corrected helicity (Subsection \ref{sub:Well-Posed}). The reference domain formulation therefore replaces a moving domain by a fixed domain carrying a varying metric, and changes nothing else. 


From here onward we work with a pulled back version of problem \eqref{prob:MRxMHD_fixed}, knowing that for an appropriate choice of the constants \(p_1, \dots, p_n\) its critical points contain the solutions of \eqref{prob:MRxMHD_adiabatic}. Problem \eqref{prob:MRxMHD_fixed} carries no multiplier for the pressure and no positivity requirement, so the \(p_i\) may be arbitrary real constants, whereas \eqref{prob:MRxMHD_adiabatic} and the entropy-constrained formulations produce only pressures of a prescribed sign (Appendix \ref{section:energy}). Every statement we make about \eqref{prob:MRxMHD_fixed} with positive pressures therefore applies verbatim to the critical points of \eqref{prob:MRxMHD_adiabatic}.

Let \(\Omega^k(\ST)\) denote the space of smooth \(k\)-forms on \(\ST\). We write \(d\colon\Omega^k(\ST)\to \Omega^{k+1}(\ST)\) for the \emph{exterior derivative}, \(\star\colon\Omega^k(\ST)\to\Omega^{3-k}(\ST)\) for the \emph{Hodge star} operator, \(\wedge\) for the wedge product, and \(i_{V}\colon \Omega^k(\ST) \to \Omega^{k-1}(\ST)\) for interior multiplication by a smooth vector field \(V\). The \emph{co-derivative} applied to a \(k\)-form on an \(q\)-dimensional Riemannian manifold is \(\delta \coloneqq (-1)^{q(k+1)+1} \star d \star\) \cite{lee}. On \(\ST_i\) this is \((-1)^k \star d \star\), and \(\delta\) is also used on the surface \(\partial \ST_i\), where \(q=2\). We reserve the symbol \(\mathcal{J}\) for inclusion maps, specifying the corresponding domain and codomain only when the context does not make them clear, and in integrals we often omit the explicit insertion of \(\mathcal{J}^*\) whenever the integration domain alone suffices to indicate the pullback.

Given an inclusion \(\mathcal{J}\colon \partial \ST_i \hookrightarrow \ST_i\), the \emph{tangential trace} of \(\omega \in \Omega^k(\ST_i)\) on \(\partial \ST_i\) is \(t \omega \coloneqq \mathcal{J}^* \omega\). Letting \(N\) be a unit normal field to \(\partial \ST_i\), the \emph{normal trace} is the \((k-1)\)-form \(n \omega \coloneqq \mathcal{J}^*(i_{N} \omega)\).

Recall that \(\ST\) carries the metric pulled back by \(f\) from the Euclidean metric on \(f(\ST)\), so that \(f\) is an orientation preserving isometry onto its image and \(f^*\varpi\) is the Riemannian volume form of \(\ST\). Since \(f\) is an isometry, pullback by \(f\) commutes with \(d\) and preserves pointwise inner products and integrals. The magnetic field \(B\) on \(f(\ST)\) and the volume form \(f^*\varpi\) therefore yield a smooth 2-form on \(\ST\),
\begin{align*}
   \beta \coloneqq f^*\big(i_{ B}  \varpi \big) \, ,
\end{align*}
and conversely \(B = \big( (f^{-1})^{*} ( \star \beta ) \big)^{\sharp}\), the operations on the right being the Euclidean ones on \(f(\ST)\). Thus \(\beta\) is the Euclidean 2-form of \(B\) pulled back to \(\ST\) and every identity below may be read on either side of the isometry. The correspondence between the two languages is recorded in Lemma \ref{lem:vector-dictionary} of Appendix \ref{append:vector_notation}, together with a table of translations.

The MHD equilibrium equations are equivalent to
\begin{align} \label{eq:mhe_original}
  \delta \beta  \wedge \star \beta  &= \star d \tilde p \, , \quad
  d \beta  = 0 \, , \quad \mathcal{J}^* \beta  = 0 \, ,
\end{align}
by Lemma \ref{lem:vector-dictionary} items (3) and (4), together with the final assertion of that lemma. Here the boundary condition on \(B\) is replaced by \(\mathcal{J}^*\beta = 0\), where \(\mathcal{J}\colon \partial\ST \hookrightarrow \ST\) is the inclusion map and \(\mathcal{J}^*\) its pullback.

For \(\omega_1, \omega_2 \in \Omega^k (\ST)\) the \(L^2\) inner product and norm are given by:
\begin{align*}
\begin{gathered}
  \langle \omega_1, \omega_2 \rangle_{L^2} \coloneqq \int_{\ST} \langle \omega_1 ,\omega_2 \rangle \, f^* \varpi = \int_{\ST} \omega_1 \wedge \star \omega_2 \, , \qquad
  \| \omega_1 \|_{L^2}^2 \coloneqq \langle \omega_1, \omega_1 \rangle_{L^2} \, ,
\end{gathered}
\end{align*}
where \(\langle \omega_1, \omega_2 \rangle\) denotes the pointwise inner product induced by the metric on \(\ST\) pulled back from \(f(\ST)\). By Lemma \ref{lem:vector-dictionary}(2) we have \(\|\beta\| = \|B\| \circ f\) pointwise and hence \(\|\beta\|_{L^2} = \|B\|_{L^2}\), so the magnetic energy is unchanged by the transcription. We therefore retain the symbols \(\|\beta\|\) and \(\|\beta\|_{L^2}\) without ambiguity.

As \(\ST\) is a compact, oriented subset of \(\R^3\) with a smooth boundary, a closed 2-form \(\beta  \in  \Omega^2(\ST)\) that is zero when pulled back by inclusion to the boundary is automatically exact. Namely, there exists an \(a \in  \Omega^1(\ST) \) with \(\beta  = da\), as shown by Cantarella and Parsley (\citeyear{CANTARELLA20101127}, Proposition C.4) see also Lemma \ref{Lemma:Potential}. Consider any closed \( s \in \Omega^1(\ST) \) and fix a potential \(a\) of \(\beta\). Since \( ds = 0 \), we have \( \beta  = d(a + s) \), so \( a + s \) is likewise a valid primitive for \( \beta  \). Conversely, the difference of any two primitives of \(\beta\) is a closed 1-form. Hence, \(\{a+s \in \Omega^1(\ST) : \; da = \beta  ,\; ds=0\} \) is the set of every potential for \(\beta \).

Consider a set \(G \subseteq \Omega^1(\ST)\) of potentials with the following property: for every closed \(2\)-form \(\beta \) whose pullback to \(\partial \ST\) vanishes, there exists at least one potential \(a \in G\) of \(\beta\). This mirrors the corresponding symbol in the vector field case.

The \emph{helicity} of \(\beta \) on \(\ST\) with respect to \(G\) is 
\begin{align*}
  H_G(\beta ,\ST) \coloneqq \int_{\ST} a \wedge \beta  \, ,
\end{align*}
where \(a \in G\) is a potential of \(\beta \). By Lemma \ref{lem:vector-dictionary}(5) and (7), this helicity matches that of the vector field case for the vector field \(A \coloneqq  \big( (f^{-1})^{*} a \big)^{\sharp}\) on \(f(\ST)\). Here \(^\sharp\) is the canonical musical isomorphism taking a smooth 1-form to a smooth vector field.

A replacement for helicity that does not require restricting potentials to \(G\) is the loop-corrected helicity. Taking \(a_{\text{ref}}\) to be closed with the free loop integrals equal to those of \(a\), we define
\begin{align} \label{eq:relhel-forms}
    \mathscr{H}_a ( \beta_i,  \ST_i  ) \coloneqq \int_{ \ST_i } ( a - a_{\text{ref}} ) \wedge \beta_i  \, .
\end{align}
In Subsection \ref{subsection:Relative Helicity} we give the equivalent expression using explicit loop corrections. As with the helicity, this mirrors the corresponding vector field definition.

For \(n>1\), a partition \((\ST_1,\dots,\ST_n)\) specifies the locations where \(B\) may be discontinuous. In each region we transcribe the field \(B_i\) on \(f(\ST_i)\) by setting
\[
  \beta_i \coloneqq f^* \bigl(i_{ B_i}\varpi\bigr) 
  \in \Omega^2(\ST_i) \, ,
\]
for \(i=1,\dots,n\). If \(n=1\) then \(B\) is a smooth vector field on \(f(\ST)\) and \(\beta_1\) reduces to the definition of \(\beta\).

We define the following sets of \(k\)-forms in \(\ST_i\):
\begin{align*}
   \Omega^k_D( \ST_i ) &\coloneqq \{  \omega  \in   \Omega^k( \ST_i )  : \; t\omega  = 0 \} \, , \\
   \Omega^k_N( \ST_i ) &\coloneqq \{  \omega  \in   \Omega^k( \ST_i )  : \; n \omega  = 0 \} \, , \\
  \mathcal{C}( \ST_i ) &\coloneqq  \{  \omega  \in   \Omega^2( \ST_i )  : \; t \omega  = 0 \, , \; d \omega  = 0 \} \, .
\end{align*}
Elements of \(\Omega^k_D( \ST_i )\) and \( \Omega^k_N( \ST_i )\) are called Dirichlet and Neumann \(k\)-forms respectively. By Lemma \ref{lem:vector-dictionary}(3) and (4), \(\beta_i \in \mathcal{C}(\ST_i)\) if and only if \(B_i \in C( f(\ST_i) )\), so that \(\mathcal{C}(\ST_i)\) is the transcription of the admissible set of Subsection \ref{sub:overview}.

For a given partition \((\ST_1, \dots,\ST_n)\) and piecewise-constant pressure \(\tilde p : \ST \to \R\) equal to \(\tilde p_i\) in \(\mathrm{int}(\ST_i)\), a solution to the MRxMHD equilibrium equations is given by \(\beta _i\in \mathcal{C}(\ST_i) \), \(i = 1, \dots, n\), and \(f \in \I\) solves
\begin{align}
\label{prob:MRxMHD_forms}
  d \star \beta _i = \mu_i \beta _i \, , \qquad
  \Big\llbracket \ifrac12 \| \beta  \|^2 + \tilde p \Big\rrbracket \Big|_{I_q} = 0 \, ,
\end{align}
for each \(i\), \(q\) and constants \(\mu_i \in \R\). In equations \eqref{prob:MRxMHD_forms} the contribution from \(f\) lies in the Hodge star and the norm. By Lemma \ref{lem:vector-dictionary} item (6) the first equation holds if and only if \(\nabla \times B_i = \mu_i B_i\) on \(f(\ST_i)\). We therefore continue to call equation \eqref{prob:MRxMHD_forms} a Beltrami equation for \(B_i\). By Lemma \ref{lem:vector-dictionary}(2) the second equation in \eqref{prob:MRxMHD_forms} is the continuity of the total pressure \( (\sfrac12) \|B\|^2 + p\) across each interface \(f(I_q)\). Hence \eqref{prob:MRxMHD_forms} is the transcription of \eqref{prob:MRxMHD}.

Given a partition \((\ST_1,\dots,\ST_n)\), the magnetic energy of a configuration \(f \in \I\), \(\beta_i \in \mathcal{C}(\ST_i)\) is
\begin{align*}
    E(f, \beta_{1}, \dots, \beta_{n})  =  \sum_{i=1}^n \Big( \ifrac{1}{2} \big\| \beta_{i}\big\|_{L^2}^2 -  \tilde p_i \myabs{ f(\ST_i) } \Big)
    \, ,
\end{align*}
where \(\myabs{f(\ST_i)}\) remains the Euclidean volume of \(f(\ST_i)\). Note that \(f\) enters both terms through the pullback metric in the Hodge star and through the volume form in \(\myabs{f(\ST_i)}\). Since each \(\tilde p_i\) is constant, the term \(- \tilde p_i \myabs{f(\ST_i)}\) is unaffected by variations of \(\beta_i\).


The functional \(E\) is not the classical total energy \eqref{eq:original_energy} as it includes pressure terms \( - \tilde p_i \myabs{f(\ST_i)}\) and is the objective function of \eqref{prob:MRxMHD_fixed}. By Proposition \ref{prop:energy_comparison}, however, the two functionals \(W\) and \(E\) have the same critical points, so nothing is lost in passing from \eqref{eq:original_energy} to \(E\). We therefore refer to \(E\) as the magnetic energy throughout, with the understanding that it is a surrogate objective sharing the critical points of \eqref{eq:original_energy}.


\subsection{De Rham Cohomology and Singular Homology}\label{sub:coho_homo}

Let \((\ST_1, \dots, \ST_n)\) be a partition of $\ST$ with corresponding interfaces $I_1,\dots,I_m$ (see Section \ref{section:Introduction}). For each $i$ and each integer $k\ge 0$ we write
\[
H^k_{\mathrm{dR}}(\ST_i)\quad\text{and}\quad H_k(\ST_i)
\]
for the $k$th de Rham cohomology group and the $k$th singular homology group of $\ST_i$, respectively, both taken with real coefficients. The homology and cohomology groups of the boundary are denoted analogously for $\partial\ST_i$. The corresponding groups for \(\ST_i\) relative to the boundary \(\partial \ST_i\) are written
\[
H^k_{\mathrm{dR}}(\ST_i,\partial\ST_i)\quad\text{and}\quad H_k(\ST_i,\partial\ST_i) \, .
\]
For the constructions of singular homology and de Rham cohomology we refer to Appendix \ref{sub:construct}. Within the current section we utilise de Rham's theorem and Lefschetz duality, but prove most other results. The outcomes of this section are known and presented here for reader's convenience.

For each subdomain $\ST_i$ and each integer $k \in \mathbb{Z}$ there is a bilinear pairing
\begin{align}\label{eq:pairing}
    \big\langle \eq{\omega}, \eq{z} \big\rangle = \int_{z}\omega \, ,
\qquad \eq \omega \in H^k_{\mathrm{dR}}(\ST_i) \, ,\ \ \eq z \in H_k(\ST_i) \, .
\end{align}
The pairing is well-defined by Stokes' theorem. Namely, it depends only on the cohomology classes \(\eq \omega\) and \(\eq z\), not on the particular representatives \(\omega\) and \(z\). To see this we replace $\omega$ by $\omega+d\nu$ with $\nu$ a smooth $(k-1)$-form, then
\[
\int_z(\omega+d\nu)-\int_z\omega=\int_z d\nu=\int_{\partial z}\nu=0,
\]
since $z$ is a cycle ($\partial z=0$). Likewise, if $z$ is replaced by $z+\partial y$ for some $(k+1)$-chain $y$, then
\[
\int_{z+\partial y}\omega-\int_z\omega=\int_{\partial y}\omega=\int_y d\omega=0,
\]
because $\omega$ is closed. The same Stokes-type argument shows the pairing descends to the relative groups. Namely, there are well-defined pairings \(\langle \eq{\omega}, \eq{z} \rangle\) by integration for \(H^k_{\mathrm{dR}}(\ST_i, \partial \ST_i) \times H_k(\ST_i)\) and \(H^k_{\mathrm{dR}}(\ST_i, \partial \ST_i) \times H_k(\ST_i, \partial \ST_i)\).

We write the 3-sphere as \(\mathbb S^3\). Let \(\overline{\ST}_i^c\) denote the closure of \( (\R^3\cup\{\infty\} ) \setminus \ST_i \cong   \mathbb{S}^3 \setminus \ST_i \). Since \(\mathbb S^3\) is compact and \(\overline{\ST}_i^c\) is closed in \(\mathbb S^3\), it follows that \(\overline{\ST}_i^c\) is compact. Moreover, the first homology group \(H_1(\overline{\ST}_i^c)\) is isomorphic to the first homology of the closure of \(\R^3\setminus\ST_i\), see Appendix \ref{sub:complement} for details. For the homological constructions, we extend every \(f \in \I\) to an ambient homeomorphism of \(\R^3 \cup \{\infty\}\) by the identity, which is consistent because \(f|_{\partial \ST} = \mathrm{id}\), so that \(f ( \overline{\ST}_i^c ) = \overline{f(\ST_i)}^{\,c}\) and the complementary constructions below are carried onto the laboratory domain by \(f\) exactly as those on \(\ST_i\) are.
\begin{lemma}\label{Lemma:finite}
    The homologies \(H_k (\ST_i)\) and \(H_k (\overline{\ST}_i^c)\) are finite dimensional real vector spaces for all \(k,i\).
\end{lemma}
\begin{proof}
    By construction the homologies are vector spaces. To show that these are also finite dimensional, we follow a proof by Bott and Tu (\citeyear{bott2013differential}, Proposition 5.3.1) see also Appendix \ref{sub:finite}.
\end{proof}

Given Lemma \ref{Lemma:finite}, let the dimension of \(H_1 (\ST_i)\) be denoted \(\ell_i < \infty\) and the dimension of \(H_1 (\overline{\ST}_i^c)\) be denoted \(\kay_i < \infty\). If \(\ell_i, \kay_i > 0 \) then select any bases
\begin{align*}
    \big\{  \eq{s_{i,j} } \big\}_{j=1}^{\ell_i} \ \text{ for } \ H_1(\ST_i) \, , \qquad \text{and} \qquad \big\{  \eq{ t_{i,j} } \big\}_{j=1}^{\kay_i} \ \text{ for } \ H_1 (\overline{\ST}^c_i ) \, .
\end{align*}
If \(\ell_i = 0\) or \(\kay_i = 0\) then the corresponding basis is empty.

Assuming that \(\ell_i, \kay_i > 0 \) and bases for \( H_1(\ST_i) \) and \(H_1 (\overline{\ST}^c_i )\) are given, we may define dual bases for the corresponding de Rham cohomologies because the de Rham theorem tells us that equation \eqref{eq:pairing} is a perfect pairing \cite{lee}. Namely, there are unique basis elements \(\eq{s_{i,j}^*} \in H^1_{\mathrm{dR}}(\ST_i)\) and \(\eq{ t_{i,j}^*} \in  H^1_{\mathrm{dR}}(\overline{\ST}^c_i)\) given by 
\begin{align*}
        \big\langle \eq{ s_{i,k}^* }, \eq{ s_{i,j} } \big\rangle &= 
        \delta_{jk}
        \, , \ \text{ for } \ k,j = 1,2, \dots, \ell_i \, ,
        \\
        \big\langle \eq{ t_{i,k}^* }, \eq{ t_{i,j} } \big\rangle &= \delta_{jk} \, , \ \text{ for } \ k,j = 1,2, \dots, \kay_i  \, .
\end{align*}
Lefschetz duality combined with de Rham duality provides an isomorphism over compact, oriented 3-manifolds \cite{bott2013differential}
\begin{align*}
    H^1_{\mathrm{dR}}(\ST_i) \cong H_2(\ST_i, \partial \ST_i)
\end{align*}
given by
\begin{align*}
     [\omega] \mapsto [z] \qquad \text{for} \qquad  \int_{\ST_i} \omega \wedge \omega_0
    =   \int_{z}   \omega_0 \, , \quad \text{for all } [ \omega_0 ] \in H^2_{\mathrm{dR}}(\ST_i, \partial \ST_i) \, ,
\end{align*}
where \([\omega] \in H^1_{\mathrm{dR}}(\ST_i)\) and \([z] \in H_2(\ST_i, \partial \ST_i) \). Similarly, for \(H^1_{\mathrm{dR}}(\overline{\ST}^c_i) \cong H_2(\overline{\ST}^c_i, \partial \overline{\ST}^c_i)\). 

With this isomorphism we may utilise the bases for \(H^1_{\mathrm{dR}}(\ST_i)\) and \(H^1_{\mathrm{dR}}(\overline{\ST}^c_i) \) to generate bases for \(H_2(\ST_i, \partial \ST_i)\) and \(H_2(\overline{\ST}^c_i, \partial \overline{\ST}^c_i)\) respectively. These basis elements are denoted \([T_{i,j}]  \in H_2(\ST_i, \partial \ST_i)\), \([S_{i,j}]\in H_2(\overline{\ST}^c_i, \partial \overline{\ST}^c_i)\), and they are given uniquely by the relations,
\begin{align*}
    \int_{\ST_i} s_{i,j}^* \wedge \omega_1  &=   \int_{T_{i,j}}  \omega_1 \, , \ \text{ for } \ j = 1,2, \dots, \ell_i \, , 
    \\
    \int_{\overline{\ST}^c_i} t_{i,j}^* \wedge \omega_2  &=   \int_{S_{i,j}}  \omega_2\, , \ \text{ for } \ j = 1,2, \dots, \kay_i \, , 
\end{align*}
for all \(\omega_1 \in H^2_{\mathrm{dR}}(\ST_i, \partial \ST_i)\) and \(\omega_2 \in H^2_{\mathrm{dR}}(\overline{\ST}^c_i, \partial \overline{\ST}^c_i)\). 

Utilising de Rham's theorem, the corresponding algebraic duals are computed via the perfect pairing in equation \eqref{eq:pairing}, acting on \( H^2_{\mathrm{dR}}({\ST}_i, \partial {\ST}_i) \times H_2( {\ST}_i, \partial \ST_i)\) or \( H^2_{\mathrm{dR}}(\overline{\ST}^c_i, \partial {\ST}_i) \times H_2( \overline{\ST}^c_i, \partial \ST_i)\). Hence, we define the corresponding basis elements \([T_{i,j}^*] \in H^2_{\mathrm{dR}}( \ST_i, \partial \ST_i) \) and \([S_{i,j}^*]\in H^2_{\mathrm{dR}}( \overline{\ST}^c_i, \partial \ST_i) \) via the relationships
\begin{align*}
        \big\langle [T_{i,k}^*], [T_{i,j}] \big\rangle &= 
        \delta_{jk}
        \, , \ \text{ for } \ k,j = 1,2, \dots, \ell_i \, ,
        \\
        \big\langle [S_{i,k}^*], [S_{i,j}] \big\rangle &= \delta_{jk} \, , \ \text{ for } \ k,j = 1,2, \dots, \kay_i  \, .
\end{align*}
We note that combining these algebraic duals with the Lefschetz duality isomorphism, gives the usual non-degenerate Lefschetz pairing,
\begin{align*}
    \int_{\ST_i} s_{i,j}^* \wedge T_{i,k}^* & =   \int_{T_{i,j}} T_{i,k}^* = \delta_{jk}\, , \ \text{ for } \ j,k = 1,2, \dots, \ell_i \, , 
    \\ 
    \int_{\overline{\ST}^c_i} t_{i,j}^* \wedge S_{i,k}^*  &=   \int_{S_{i,j}}  S_{i,k}^*  = \delta_{jk} \, , \ \text{ for } \ j,k = 1,2, \dots, \kay_i \, .
\end{align*}
We are interested in one final collection of variables, again given through an isomorphism. This isomorphism we provide in the following Lemma.
\begin{lemma}\label{Lemma:isomorphism_inclusion}
    Consider the inclusion maps \(\mathcal{J} : \partial \ST_i \hookrightarrow \ST_i\) and \(\mathcal{J}^c : \partial \ST_i \hookrightarrow \overline{\ST}^c_i\). Then \((\mathcal{J}_*, \mathcal{J}^c_*) : H_1(\partial \ST_i) \to H_1(\ST_i) \oplus H_1(\overline{\ST}^c_i)\) is an isomorphism.
\end{lemma}
\begin{proof} 
    The sequence is applied to open thickenings of \(\ST_i\) and \(\overline{\ST}^c_i\) in \(\mathbb{S}^3\) that deformation retract onto them, as in Appendix \ref{append:general_formulation}, so that the two sets are open and cover \(\mathbb{S}^3\), and their intersection retracts onto \(\partial \ST_i\). By a Mayer--Vietoris sequence for the decomposition \( \mathbb S^3 = \ST_i \cup \overline{\ST}^c_i\), we obtain for \(k>0\) the exact sequence
    \begin{align*}
        H_{k+1} (\mathbb{S}^3) \to H_k (\partial \ST_i) \xrightarrow{ \, (\mathcal{J}_* , - \mathcal{J}^c_* ) \, } H_k (\ST_i) \oplus H_k (\overline{\ST}_i^c) \to H_{k} (\mathbb{S}^3) \, .
    \end{align*}
    We have that \( H_{k} (\mathbb{S}^3) \cong \{0\}\) for \(k=1,2\), hence the map \((\mathcal{J}_* , -\mathcal{J}^c_* )\) is an isomorphism 
    \begin{align}\label{eq:split_boundary}
        H_1 (\partial \ST_i) \cong H_1 (\ST_i) \oplus H_1 (\overline{\ST}_i^c) \, ,
    \end{align}
    as is the map \((\mathcal{J}_* , \mathcal{J}^c_* )\). 
\end{proof}
As \((\mathcal{J}_*, \mathcal{J}^c_*) : H_1(\partial \ST_i) \to H_1(\ST_i) \oplus H_1(\overline{\ST}^c_i)\) is an isomorphism, it has an inverse \((\mathcal{J}_*, \mathcal{J}^c_*)^{-1}\) giving a unique element \([\sigma_{i,j}] \coloneqq (\mathcal{J}_*, \mathcal{J}^c_*)^{-1} ( [s_{i,j}] , 0) \in H_1 (\partial \ST_i)\) for each \(i,j\). Similarly, we define a unique element \([\tau_{i,j}] \coloneqq (\mathcal{J}_*, \mathcal{J}^c_*)^{-1} ( 0, [t_{i,j}] ) \in H_1 (\partial \ST_i)\) for each \(i,j\).

\begin{lemma}\label{Lemma:form_basis}
Given non-empty bases \( \{[s_{i,j}]\}_{j=1}^{\ell_i}\), \(\{[t_{i,j}]\}_{j=1}^{\kay_i}\), of \(H_1(\ST_i)\) and \(H_1(\overline{\ST}^c_i)\) respectively, then
\[
\big\{[\sigma_{i,1}], \dots , [\sigma_{i,\ell_i}],[\tau_{i,1}], \dots, [\tau_{i,\kay_i}]\big\}
\]
is a basis of \(H_1(\partial\ST_i)\).
\end{lemma}
\begin{proof}
The collection \(\{([s_{i,1}],0), \dots, (0,[t_{i,1}]), \dots\}\) is a basis of the direct sum \(H_1(\ST_i)\oplus H_1(\overline{\ST}_i^c)\). Applying the isomorphism \((\mathcal J_*,\mathcal J^c_*)^{-1}\) carries that basis to the set \(\{[\sigma_{i,1}]\mathbin{,}\dots \mathbin{,} [\tau_{i,1}] \mathbin{,} \dots\}\) in \(H_1(\partial\ST_i)\). An isomorphism sends bases to bases, therefore giving a basis of \(H_1(\partial\ST_i)\).
\end{proof}

We can consider maps induced on de Rham cohomologies via the inclusions \(\mathcal J\) and \(\mathcal J^c\). In this case there is a Mayer--Vietoris sequence for de Rham classes
\begin{align*}
    \{0\} \cong H^{1}_{\mathrm{dR}} ( \mathbb{S}^3 ) \to H^{1}_{\mathrm{dR}} ( \ST_i ) \oplus H^{1}_{\mathrm{dR}} ( \overline{\ST}^c_i ) \xrightarrow{\mathcal{J}^* -(\mathcal{J}^c)^*} H^{1}_{\mathrm{dR}} ( \partial \ST_i ) \to H^{2}_{\mathrm{dR}} (  \mathbb{S}^3 ) \cong \{0\} \, , 
\end{align*} 
so the map \(\mathcal{J}^*- (\mathcal{J}^c)^*\) is an isomorphism, acting as
\begin{align*}
    (\mathcal{J}^*- (\mathcal{J}^c)^*) ([s],[t]) = \mathcal{J}^*[s]- (\mathcal{J}^c)^*[t] \qquad \text{for} \qquad ([s], [t]) \in H^{1}_{\mathrm{dR}} ( \ST_i ) \oplus H^{1}_{\mathrm{dR}} ( \overline{\ST}^c_i ) \, .
\end{align*}
Hence \(\mathcal{J}^* + (\mathcal{J}^c)^* \) is an isomorphism. We define the unique elements \([\sigma_{i,j}^*] \coloneqq (\mathcal{J}^* + (\mathcal{J}^c)^*) ( [s_{i,j}^*] , 0) \in H^1_{\mathrm{dR}} (\partial \ST_i)\) and \([\tau_{i,j}^*] \coloneqq (\mathcal{J}^* + (\mathcal{J}^c)^*)( 0, [t_{i,j}^*] ) \in H^1_{\mathrm{dR}} (\partial \ST_i)\). Note that \([\sigma_{i,j}^*], [\tau_{i,j}^*]\) are clearly not defined as the duals of \([\sigma_{i,j}], [\tau_{i,j}]\) respectively. But these are, in fact, the unique duals in \(H^1_{\mathrm{dR}}(\partial \ST_i)\), as
\begin{align*}
\big\langle [\sigma^*_{i,k}],[\sigma_{i,j}] \big\rangle
&=\int_{\sigma_{i,j}} \sigma^*_{i,k}
=\int_{\sigma_{i,j}} \mathcal J^* (s^*_{i,k}) +  (\mathcal{J}^c)^* (0)
= \int_{\sigma_{i,j}} \mathcal J^* (s^*_{i,k})
\\
&= \int_{\mathcal J_* \sigma_{i,j}} s^*_{i,k}
= \int_{s_{i,j}} s^*_{i,k}
= \delta_{jk} \, , 
\end{align*}
similarly \(\delta_{jk} = \langle  [\tau_{i,k}^*]  , [\tau_{i,j}] \rangle\). Additionally, the pairings between the remaining elements,
\begin{align*}
\big\langle [\tau^*_{i,k}],[\sigma_{i,j}] \big\rangle
= \int_{\sigma_{i,j}} \tau^*_{i,k}
= \int_{\sigma_{i,j}}   \mathcal J^* (0) +  (\mathcal{J}^c)^* (t^*_{i,k})
= \int_{\mathcal{J}^c_* \sigma_{i,j}}   t^*_{i,k}
= 0\, , 
\end{align*}
similarly \(\langle [\sigma^*_{i,k}],[\tau_{i,j}]\rangle=0\).

Note that \(\sigma_{i,j}\) play the roles of curves with an orientation in Subsection \ref{sub:back}, their images \(f(\sigma_{i,j})\) showing as integration domains in the loop-corrected helicity \eqref{eq:relhel-vector-intro} for example. The chains are identified with these oriented curves, in the toroidal geometry, under appropriate isomorphisms.

\begin{theorem}\label{thr:Alexander}
    For \(\ell_i > 0\) consider any basis \( \{ [s_{i,1}], \dots , [s_{i,\ell_i}]\} \subset H_1 ( \ST_i) \). Then there exists a choice of basis \(\{[t_{i,1}], \dots , [t_{i,\ell_i}]\}\) for \(H_1 (\overline{\ST}^c_i )\) with dimension \(\kay_i = \ell_i\) so that the following additional conditions hold:
    \begin{enumerate}[label={\arabic*. }  ]
        \item For all \( j, k \in \{ 1, \dots, \ell_i \} \) and \(i = 1,2, \dots, n\): 
        \begin{align*}
        \begin{gathered}
             \int_{\partial \ST_i} \tau_{i,j}^* \wedge \sigma_{i,k}^* = \delta_{jk}  \,, \qquad
            \int_{\partial \ST_i} \sigma_{i,j}^* \wedge \sigma_{i,k}^* = 
             0 \,,    \qquad 
             \int_{\partial \ST_i} \tau_{i,j}^* \wedge \tau_{i,k}^*
            = 0 \, .
        \end{gathered}
        \end{align*}   
        \item The basis elements \( [S_{i,j}] \in H_2( \overline{\ST}^c_i, \partial \overline{\ST}^c_i) \) and \( [T_{i,j}] \in H_2( \ST_i, \partial \ST_i) \) satisfy \( \partial [S_{i,j}] =  [\sigma_{i,j}] \) and \( \partial [T_{i,j}] =  [\tau_{i,j}] \) for \( j = 1, \dots, \ell_i \). 
    \end{enumerate}
\end{theorem}

\begin{proof}
    An exposition of the original proof by Cantarella and Parsley \citeyearpar[Theorems B.2 and B.3]{CANTARELLA20101127} is adapted into differential forms and presented in Appendix \ref{section:Alexander Basis}. 
\end{proof}

For each domain \(\ST_i\) with \(\ell_i > 0\), any choice of basis for \(H_1(\ST_i)\), together with a basis for \(H_1(\overline{\ST}_i^{\,c})\) as in Theorem \ref{thr:Alexander}, induces a basis for \(H_1(\partial \ST_i)\) via Lemma \ref{Lemma:form_basis}. We refer to such a basis of \(H_1(\partial \ST_i)\) as an \emph{Alexander basis}. If \(\ell_i = 0\), then \(H_1(\partial \ST_i) \cong \{0\}\) (see Appendix \ref{section:Alexander Basis}), and we define the Alexander basis to be empty. Furthermore, we assume from here onward that the basis elements \([t_{i,1}], \dots , [t_{i,\ell_i}]\) are uniquely determined by \([s_{i,1}], \dots , [s_{i,\ell_i}]\) according to the construction in Appendix \ref{section:Alexander Basis}. Theorem \ref{thr:Alexander} gives \(\kay_i = \ell_i\), and we write \(\ell_i\) throughout the remainder of the paper.

The constructions in this section are encoded by the following commutative diagram, which records the relevant homology classes associated to the reference basis element \([s_{i,j}] \in H_1(\ST_i)\). Taking duals with respect to the corresponding basis yields the associated de Rham cohomology basis elements:
\[
    \begin{tikzcd}[column sep=5em]
        {[T_{i,j}]\in H_2(\ST_i,\partial\ST_i)} \arrow[swap]{d}{\partial}
        & {[s_{i,j}]\in H_1(\ST_i)} \arrow[swap]{l}{\text{LD}^*}  
        & {[\sigma_{i,j}]\in H_1(\partial\ST_i)} \arrow[swap]{l}{\mathcal{J}_*} \\
        {[\tau_{i,j}]\in H_1(\partial\ST_i)} \arrow{r}{\mathcal{J}^c_*}
        & {[t_{i,j}]\in H_1(\overline{\ST}_i^c)} \arrow{r}{\text{LD}^*} 
        & {[S_{i,j}]\in H_2(\overline{\ST}^c_i,\partial\overline{\ST}^c_i)} \arrow[swap]{u}{\partial}
    \end{tikzcd}
\]
Here \(\mathrm{LD}^*\) denotes Lefschetz duality followed by passage to the dual basis. The induced map \([s_{i,j}] \mapsto [t_{i,j}]\), together with this dualisation, realises \emph{Alexander duality}. Moreover, the composition of all maps around the diagram is the identity.

We define the \emph{genus} of \(\partial \ST_i\) to be the dimension \(\ell_i\) of \(H_1(\ST_i)\) (equivalently, of \(H_1(\overline{\ST}^{\,c}_i)\)). This equals the total genus of the boundary \(\partial \ST_i\). In MHD, as discussed in Subsection \ref{sub:back}, we typically work with toroidal level sets and boundaries. In this setting, the genus is obtained by counting the number of connected toroidal boundary components, since each torus has genus \(1\). For example, the domain shown in Figure \ref{fig:topology}, whose boundary consists of three tori, has total genus \(3\) \cite{cantarella2002vector}.

As mentioned in the introduction Figure \ref{fig:topology} is an example of a case where our generalisation for the MRxMHD variational formulation can be applied. The domain boundary is a union of tori, but the domain is not a hollow or solid torus. By the argument in Subsection \ref{sub:back}, an MHD solution with flux surfaces that foliate the domain will not satisfy \(\| \nabla p \| \neq 0\) everywhere. However, it is plausible that MRxMHD is able to approximate such MHD equilibria if they exist for the corresponding helicity, pressure, and flux values chosen.

\subsection{Hodge--Morrey--Friedrichs Decomposition} \label{sub:Hodge}

The space $\C$ decomposes via a Hodge--Morrey--Friedrichs (HMF) decomposition \cite{schwarz-1995, Gauge_freedom}. Firstly, we define some notation for function spaces. We reserve \(\mathcal{H}^k(\ST_i)\) for harmonic \(k\)-forms on \(\ST_i\), and a subscript \(D\) or \(N\) indicates the Dirichlet or Neumann boundary conditions, as in \( \Omega^k (\ST_i)\). More explicitly, we have
\begin{align*}
    \mathcal{H}^k_D (\ST_i)  &= \{   \lambda \in  \Omega^k (\ST_i) \, : \; d \lambda =  \delta \lambda =  t \lambda = 0    \} \, , \\ \mathcal{H}^k_N (\ST_i) & = \{   \lambda \in  \Omega^k (\ST_i) \, : \; d \lambda =  \delta \lambda =  n \lambda = 0    \} \, . 
\end{align*}
Throughout this subsection \(\star\), \(\delta\) and the harmonic spaces \(\mathcal{H}^k_D (\ST_i)\), \(\mathcal{H}^k_N (\ST_i)\) are taken with respect to the metric pulled back by \(f\) and therefore depend on \(f\). In particular the harmonic representatives \(T_{i,j}^*\) selected below move with the metric.

Consider \(\beta _i \in \C\), then we have a HMF decomposition (see Appendix \ref{append:Hodge}):
\begin{align*}
    \star \beta _i = \delta (\star \eta_i) + d \phi_i + dz_i + \star \lambda_i \, , 
\end{align*}
where \(\star \eta_i \in  \Omega_N^2 (\ST_i) \), \(\lambda_i \in \mathcal{H}_D^2(\ST_i)\), and \(\phi_i,z_i \in  \Omega^0 (\ST_i)\) for \(\phi_i|_{\partial \ST_i} = 0\), and \(z_i\) is a harmonic function \(\Delta z_i = 0\) where \(\Delta \) is the Laplacian. Via an application of a Hodge star isomorphism:
\begin{align*}
    \beta _i = d \eta_i + \star d\phi_i + \star dz_i + \lambda_i \, .
\end{align*}
Using \(d \beta _i = d \star d \phi_i = 0\) and \(\phi_i|_{\partial \ST_i} = 0\) with Stokes' theorem:
\begin{align*}
    \| d \phi_i \|_{L^2}^2 = \int_{\ST_i}d \phi_i \wedge \star d \phi_i = -  \int_{\ST_i} \phi_i \wedge d \star d \phi_i + \int_{\partial \ST_i} \phi_i \wedge \star d \phi_i = 0 \, .
\end{align*}
With smoothness of \(\phi_i\), we have that \(d \phi_i = 0\). Similarly for the \(dz_i\) component, utilising that \(t \beta _i = t ( \star dz_i ) =0 \):
\begin{align*}
    \| d z_i \|_{L^2}^2 = \int_{\ST_i}d z_i \wedge \star d z_i =  -  \int_{\ST_i} z_i \wedge d \star d z_i + \int_{\partial \ST_i} t(z_i) \wedge t( \star d z_i) = 0 \, . 
\end{align*}
By smoothness of \(z_i\) we have \(dz_i = 0\). Therefore we are left with a decomposition \(\beta _i = d \eta_i + \lambda_i\). By the Hodge isomorphism and Poincaré–Lefschetz duality
\begin{align*}
    \h \cong H^2_{\mathrm{dR}} ( \ST_i, \partial \ST_i) \cong H^1_{\mathrm{dR}} ( \ST_i)  \, ,
\end{align*}
for coefficients in \(\R\). So \(\h\) is a \(\ell_i\)-dimensional vector space. The set \(\{ [T_{i,j}^*] \}_{j=1}^{\ell_i}\) (\(\ell_i > 0\)) forms a basis for \( H^2_{\mathrm{dR}} (\ST_i, \partial \ST_i )\). The Hodge theorem says that there is a unique harmonic representative for each class \([T_{i,j}^*]\). So, without loss of generality, we select \(T_{i,j}^* \in \h \). As \(T_{i,j}^*\) are a basis for \(\h\), for constants \(\psi_{i,j}\) we have
\begin{align}
    \beta _i = d \eta_i + \lambda_i = d \eta_i + \sum_{j=1}^{\ell_i} \psi_{i,j} T_{i,j}^* \, . 
\end{align}
In the case of \(\ell_i = 0\):
\begin{align*}
    \beta _i = d \eta_i  \, . 
\end{align*}
Examples of these decompositions in the case of solid and hollow tori are given by Pfefferlé, Noakes, and Perrella \citeyearpar{Gauge_freedom}.

\subsection{Flux} \label{sub:flux}

Let \(\beta _i\in\C\) be a closed Dirichlet \(2\)-form on \(\ST_i\), and let \([T]\in H_2(\ST_i,\partial\ST_i)\) be a relative homology class. The \emph{flux} of \(\beta _i\) through \(T\) is by definition the pairing \(\langle[\beta _i],[T]\rangle\). Applying Stokes' theorem to this pairing gives
\begin{align*}
\big\langle [\beta _i],[T]\big\rangle
&= \int_T \beta _i
= \int_{\partial T}\mathcal{J}^*\eta_i + \sum_{j=1}^{\ell_i}\psi_{i,j}\int_T T_{i,j}^*
= \sum_{j=1}^{\ell_i}\psi_{i,j}\int_T T_{i,j}^* =\sum_{j=1}^{\ell_i}\psi_{i,j}\big\langle [T_{i,j}^*],[T]\big\rangle \, ,
\end{align*}
since \(\eta_i\) vanishes on \(\partial\ST_i\). In particular, for the relative cycles \(T_{i,1},\dots,T_{i,\ell_i}\) we obtain
\[
\big\langle [\beta _i],[T_{i,j}]\big\rangle=\psi_{i,j}\qquad(1\le j\le\ell_i) \, ,
\]
so that specifying the values \(\psi_{i,1},\dots,\psi_{i,\ell_i}\) is equivalent to fixing the harmonic component of \(\beta_i\) and therefore determines the flux of \(\beta _i\) through every class \([T]\). This is the set of constraints we impose when formulating the variational problem with fixed flux. The flux constraints are invariant under the push-forward action of any \(f\in\I\), since
\[
\big\langle[\beta _i],[T_{i,k}]\big\rangle
=\big\langle[(f^{-1})^* \beta _i],[f(T_{i,k})]\big\rangle \, ,
\]
which by Lemma \ref{lem:vector-dictionary}(4) is the flux 
\[
    \int_{f(T_{i,k})} B_i \cdot dS
\]
 of \(B_i\) through the corresponding surface in \(f(\ST_i)\).
When \(\ell_i=0\) there is no harmonic component and every closed Dirichlet \(2\)-form is exact, so \(\beta_i=d\eta_i\) with \(\eta_i\in\Omega^1_D(\ST_i)\). Stokes' theorem then shows that in this case the flux through any relative class vanishes,
\[
\langle[\beta _i],[T]\rangle=\int_T d\eta_i=\int_{\partial T}\mathcal{J}^*\eta_i=0 \, .
\]
Therefore \(\beta_i\) is an element of the flux-free subspace of closed Dirichlet \(2\)-forms when \(\ell_i=0\).

\subsection{Connected \texorpdfstring{\(\ST_i\)}{Domains}}

We assume that each partition element \(\ST_i\) is connected. This hypothesis removes a benign but inconvenient source of non-uniqueness, which we now explain.

If \(\ST_i\) is a connected, compact domain with smooth boundary, then a Beltrami field on \(\ST_i\) is determined uniquely by the Beltrami parameter \(\mu_i\) (that is not in the corresponding flux-free spectrum of the curl operator) together with the flux constraints \(\psi_{i,1},\dots,\psi_{i,\ell_i}\) \cite{yoshida2012helical}. Replacing the prescription of \(\mu_i\) by a constraint on the loop-corrected helicity still yields existence of a solution to the Beltrami equations.

By contrast, if \(\ST_i\) has several connected components then a helicity condition imposed only on the whole of \(\ST_i\) fixes merely the sum of the helicities of the components. The total constraint does not prescribe the individual component helicities and may leave additional freedom. Requiring each \(\ST_i\) to be connected removes this ambiguity. This is a restriction on the partition rather than on \(\ST\). A disconnected \(\ST_i\) may be refined into its connected components, provided the refined partition still satisfies the assumption of Subsection \ref{sub:overview} that no point of \(\ST\) lies in more than two subregions. The refinement replaces the single constant \(h_i\) by one constant for each component, which is exactly the freedom removed above.

\section{Main Results}\label{section:Main Results}

Consider the following objects as given:
\begin{enumerate}[label={\arabic*. }  ]
  \item A finite collection of compact, oriented, connected subdomains \(\ST_i\) \((i=1,\dots,n)\) that partition \(\ST\). The corresponding interfaces are smooth, closed, connected embedded surfaces in \(\text{int}( \ST) \) and each point in the domain belongs to at most 2 subdomains.
  \item An Alexander basis of \(H_1(\partial\ST_i)\) for each \(i\).
  \item For each \(i\), if \(\ell_i=0\), then \(\tilde p_i\) and \(h_i\) are given. If \(\ell_i>0\), then \(\tilde p_i\), \(h_i\) and \(\psi_{i,1},\dots,\psi_{i,\ell_i}\) are given. 
\end{enumerate}
We remind the reader that we are fixing \(\ST_i\) and allowing the metric to change through \(f \in \I\). This effectively determines the locations of interfaces within \(f(\ST)\). 

Consider the variational problem (equivalent to \eqref{prob:MRxMHD_fixed}):
\begin{align} \label{eq:ext}
    \begin{gathered} 
    \underset{\substack{f \in \I \\ a_i \in \Omega^1 ( \ST_i)} }{ \text{c. p.} } \left( \sum_{i=1}^n  \frac{1}{2} \big\|\beta _i\big\|_{L^2}^2 -  \tilde p_i \myabs{ f(\ST_i) } \right) \, , \\
    \mathscr{H}_{a_i}( \beta_i,  \ST_i) = h_i\, , \qquad
    \int_{T_{i,j} } \beta _i = \psi_{i,j}\, ,
    \end{gathered}
\end{align}
such that \(d a_i = \beta_i \in \mathcal{C}( \ST_i)\) for all \(i\). The constraints hold for all \(i = 1, \dots, n\) and relevant \(j\). If the helicity and flux constants for some \(i\) vanish (\(h_i = \psi_{i,1} = \dots = \psi_{i,\ell_i}=0\)) then \(\beta_i = 0\) is admissible and is the unique minimiser of the corresponding magnetic term \( ( \sfrac12 )\|\beta_i\|_{L^2}^2 \). The pressure plays no part in this statement as the term \(-\tilde p_i \myabs{f(\ST_i)}\) does not depend on \(\beta_i\), although it does vary with \(f\). Configurations in which \(\beta_i\) vanishes on one or more regions are treated in Theorem \ref{prop:vanishing}.

The following proposition details a connection between solutions to the variational Problem \eqref{eq:ext} and solutions to the MRxMHD equations \eqref{prob:MRxMHD_forms} (or equations \eqref{prob:MRxMHD} in vector notation). 

\begin{proposition}\label{thr:iff_simple}
    Let \(f \in \I\) and \(a_i \in \Omega^1( \ST_i)\) such that \(da_i = \beta_i \in \mathcal{C}(\ST_i)\), \(i = 1, \dots, n\), satisfy the helicity and flux constraints of \eqref{eq:ext}, and suppose that \(\beta_i \neq 0\) for \emph{every} \(i = 1, \dots, n\). Then \((f, a_1, \dots, a_n)\) is a critical point of \eqref{eq:ext} if and only if it satisfies the MRxMHD equations \eqref{prob:MRxMHD_forms}, equivalently, by Lemma \ref{lem:vector-dictionary}, if and only if the corresponding fields \(B_i\) on \(f(\ST_i)\) satisfy \eqref{prob:MRxMHD}.
\end{proposition}


The assumption that \(\beta_i \neq 0\) within Proposition \ref{thr:iff_simple} is imposed in every region separately. The following theorem extends this to the case where any \(\beta_i\) may be zero.

\begin{theorem}\label{prop:vanishing}
    Let \(f \in \I\) and \(a_i \in \Omega^1( \ST_i)\) such that \(da_i = \beta_i \in \mathcal{C}(\ST_i)\), \(i = 1, \dots, n\), satisfy the helicity and flux constraints of \eqref{eq:ext}. Then \((f, a_1, \dots, a_n)\) is a critical point of \eqref{eq:ext} if and only if it satisfies the MRxMHD equations \eqref{prob:MRxMHD_forms}, equivalently, by Lemma \ref{lem:vector-dictionary}, if and only if the corresponding fields \(B_i\) on \(f(\ST_i)\) satisfy \eqref{prob:MRxMHD}. 
\end{theorem}



Theorem \ref{prop:vanishing} implies that the MRxMHD equations \eqref{prob:MRxMHD_forms} are necessary for a minimiser to Problem \eqref{eq:ext}. This result is known in the case of nested toroidal surfaces, as discussed in Subsection \ref{sub:back}.

Proposition \ref{thr:iff_simple} is a statement about the variational formulation, and not a new physical model or a numerical method. The MRxMHD equations \eqref{prob:MRxMHD} are unchanged, the energy and the constraints are those of the literature, and no algorithm is proposed. What is new is the class of domains on which the variational characterisation holds, and the strength of the characterisation itself. Most earlier treatments pose the problem on a solid torus partitioned by nested toroidal interfaces and obtain \eqref{prob:MRxMHD} as necessary conditions for a minimiser \cite{10.1063/1.4765691, 10.1063/1.4795739, Dewar_Yoshida_Bhattacharjee_Hudson_2015}. Proposition \ref{thr:iff_simple} replaces the toroidal hypothesis by weaker assumptions on each element \(\ST_i\) of the partition and gives the MRxMHD equations as necessary and sufficient conditions for a critical point.

The gauge condition observed within the proof of Proposition \ref{thr:iff_simple} is likewise a clarification of the existing derivations rather than a correction to them. To state it, define for \(a \in \Omega^1(\ST_i)\) the flux-weighted period functional
\begin{align} \label{eq:kappa}
    \kappa_i(a) \coloneqq \sum_{j=1}^{\ell_i} \psi_{i,j} \int_{\sigma_{i,j}} \mathcal{J}^* a \, ,
\end{align}
with \(\kappa_i(a) = 0\) when \(\ell_i = 0\). The Amperian gauge \(\mathscr{A}(\ST_i)\) satisfies \(\kappa_i(a) = 0\) for all \(a \in \mathscr{A}(\ST_i)\) automatically.

On a technical note, we record what is meant by a critical point. A \emph{path} through a configuration \(P_0 = (f, a_1, \dots, a_n)\) is a family \(P_t = (f_t, a_{1,t}, \dots, a_{n,t})\), \(t \in (-\varepsilon, \varepsilon)\) for some \(\varepsilon > 0\), with \(P_t|_{t=0} = P_0\), \(f_t \in \I\), \(a_{i,t} \in \Omega^1(\ST_i)\), and \((t,x) \mapsto f_t(x)\) and \((t,x) \mapsto a_{i,t}|_x\) smooth. We write \(f' \coloneqq \partial_t f_t|_{t=0}\) and \(a_i' \coloneqq \partial_t a_{i,t}|_{t=0}\). A path is \emph{admissible} if \(P_t\) satisfies the constraints of \eqref{eq:ext} for every \(t\). A configuration satisfying the constraints is a \emph{critical point} of \eqref{eq:ext} if the first variation of the objective vanishes at \(t=0\) along every admissible path through it.

This is the only notion of differentiation used below, and it requires no further hypotheses. The objective is an integral over the compact \(\ST_i\) of an integrand that is smooth in \((t,x)\), so \(t \mapsto E(P_t)\) is a smooth real function and may be differentiated under the integral sign. In particular no topology on \(\Omega^1(\ST_i)\) or on \(\I\) is needed. At a fixed \(f\), both the objective and the loop-corrected helicity are inhomogeneous quadratic polynomials in \(a_i\), so their increments are exact algebraic identities (Lemma \ref{lem:quad}), and Lemma \ref{lem:tangent} identifies precisely which \(a_i'\) arise from admissible paths.


\begin{remark}
    A result establishing the existence of solutions to the Beltrami equation was proven by Yoshida and Dewar. These authors use the Fredholm alternative and a technique they refer to as \textit{helicity matching} \citep[section 2.3]{yoshida2012helical}. Therefore, there is always a solution to Problem \eqref{eq:ext} when \(n=1\). In Appendix \ref{append:exist}, we prove an existence result that is weaker than this known result using the direct method from calculus of variations.
\end{remark}

\begin{theorem} \label{thr:one_gauge}
Take the setup of Proposition \ref{thr:iff_simple} and replace the loop-corrected helicity constraint in \eqref{eq:ext} by the helicity constraint \(H_{G_i}(\beta_i, \ST_i) = h_i\), where each gauge \(G_i \subseteq \Omega^1(\ST_i)\) contains at least one primitive of every \(\beta_i \in \mathcal{C}(\ST_i)\), and any two primitives of the same field in \(G_i\) differ by an exact \(1\)-form. Thus \(G_i\) selects a unique primitive modulo \(d\Omega^0(\ST_i)\). Consider a path
\[
    t \longmapsto \big( f_t, a_{1,t}, \dots, a_{n,t} \big) , \qquad t \in [-\varepsilon,\varepsilon] ,
\]
with \(f_t \in \mathcal{I}(\Lambda)\), \(a_{i,t} \in G_i\), \(da_{i,t} = \beta_{i,t} \in \mathcal{C}(\ST_i)\), differentiable at \(t = 0\) and satisfying the flux constraint for every \(t\). Write \(W_i \subseteq \Omega^1(\ST_i)\) for the set of derivatives \(\partial_t a_{i,t}|_{t=0}\) arising from such paths with fixed \(f_t = f\). Assume that for each \(i\) these variations are exhaustive, namely
\[
    \big\{ \, d a_i' :\; a_i' \in W_i  \big\}
    = \big\{ \, d\eta_i' :\; \eta_i' \in \Omega^1_D(\ST_i)  \big\} \, ,
\]
and that this holds at the level of families, meaning that whenever \(\{ \eta_{i,t} \}_{t \in (-\varepsilon,\varepsilon)} \subseteq \Omega^1_D(\ST_i)\) is smooth in \((t,x)\) and 
\[
    d \eta_{i,0} + \sum_j \psi_{i,j} T^*_{i,j} = \beta_i\, ,
\]
 there is a path \(a_{i,t}\) as above, with fixed \(f_t = f\), such that 
\[
    d a_{i,t} = d \eta_{i,t} + \sum_j \psi_{i,j} T^*_{i,j}
\]
 for every \(t\). Assume also that for each \(i\), along every such path,
\begin{align}\label{eq:gauge_condition}
    \frac{d}{dt}\bigg|_{t=0} \kappa_i(a_{i,t})  =  0 \, .
\end{align}
Then \((f, a_1, \dots, a_n)\) is a critical point of the resulting variational problem if and only if it satisfies the MRxMHD equations \eqref{prob:MRxMHD_forms}.
\end{theorem}

The inclusion of \(W_i\) in Theorem \ref{thr:one_gauge} is automatic. Namely, if \(a_{i,t}\) is a path as in Theorem \ref{thr:one_gauge} then \(\beta_i' = d a_i'\) is closed, Dirichlet, and flux-free. Hence, \(\beta_i'\) is exact with a primitive in \(\Omega^1_D(\ST_i)\) by Subsections \ref{sub:Hodge} and \ref{sub:flux}. What is assumed is the reverse inclusion, that every admissible variation of the field is realised by a variation of the potential within the gauge, together with its family-level counterpart. Both hold for the Amperian gauge: choose an Amperian primitive \(\Gamma_i\) of the fixed harmonic component and lift a smooth family \(\eta_{i,t}\) by \(a_{i,t}=\eta_{i,t}+\Gamma_i\). A fixed exact \(1\)-form may be added to match the initial representative. The gauge used by SPEC is treated in Subsection \ref{subsec:spec}.

In the toroidal setting Dewar et al. impose condition \eqref{eq:gauge_condition} by requiring the gauge function to be single-valued as the interface is displaced \citep[appendix C]{Dewar_Yoshida_Bhattacharjee_Hudson_2015}, so the condition is present in earlier work as a restriction on the admissible variations rather than as an identified condition on the gauge. Its omission from the statements of those problems has not affected their conclusions, and Subsection \ref{subsec:spec} verifies that the gauge used by SPEC satisfies Theorem \ref{thr:one_gauge}. What is added here is the identification of the condition on regions \(\ST_i\) where the relevant loops are the generators of an Alexander basis rather than a distinguished poloidal and toroidal pair. Additionally, the variation of the loop-corrected helicity cancels the corresponding boundary contribution, so no restriction on the potentials is then required.

\begin{remark} \label{remark:sharpness}
Should \eqref{eq:gauge_condition} fail along a deformation of the interfaces, \(a_{i,t}\) being a representative in \(G_i\) of the Beltrami field \(\beta_{i,t}\) at each \(t\), then an additional term proportional to \(\mu_i\kappa_i(a_i')\) may survive in the first variation of the magnetic energy. Unless these contributions vanish or cancel, the interface equation differs from the jump condition of \eqref{prob:MRxMHD_forms}.
\end{remark}

\section{Preparatory Results}\label{section:Preparatory Results}


\begin{lemma}\label{Lemma:compute}
    For any closed 1-forms \(s, \tilde{s} \in \Omega^{1} (\partial \ST_i)\) and an Alexander basis \(\{[\sigma_{i,j}],[\tau_{i,j}]\}_{j=1}^{\ell_i}\) of \(H_1(\partial \ST_i)\) with \(\ell_i > 0\) we have a decomposition:
    \begin{align}\label{eq:compute_2}
    \int_{\partial \ST_i} \tilde{s} \wedge s = \sum_{j=1}^{\ell_i}  \left(\int_{\sigma_{i,j}} \! s \int_{\tau_{i,j}}  \! \tilde{s} -\int_{\sigma_{i,j}}  \! \tilde{s}    \int_{\tau_{i,j}}  \! s   \right) 
     \, .
\end{align}
If \(\ell_i=0\), then the left-hand side of equation \eqref{eq:compute_2} is zero.
\end{lemma}


\begin{proof}
	Assume that \(\ell_i > 0\). Any 1-form \( s \in \Omega^{1} (\partial \ST_i)\) admits a Hodge decomposition into an exact form, co-exact form, and a harmonic 1-form representing a cohomology class \cite{schwarz-1995}:
\begin{align*}
    s = \delta \eta + d \phi + \gamma  \, , 
\end{align*}
where \( \eta \in  \Omega^2 (\partial \ST_i) \), \(\phi \in  \Omega^0 (\partial  \ST_i)\), and a harmonic \(\gamma \in  \mathcal{H}^1(\partial \ST_i)\). Utilising the closure of \(s\) (\(0 = ds = d \delta  \eta\)) and Stokes' theorem,
\begin{align*}
    \| \delta  \eta \|^2_{L^2} = \int_{\partial \ST_i} \delta  \eta \wedge \star \delta  \eta
    = \int_{\partial \ST_i} 
    \delta  \eta \wedge d \star \eta
    =  \int_{\partial \ST_i} 
    d\delta  \eta \wedge  \star \eta
-
\int_{\partial^2 \ST_i} \delta  \eta \wedge  \star \eta = 0 \, . 
\end{align*}
Note that the Hodge star operators on \(\partial \ST_i\) acting on 1-forms satisfy \(\star \star = -1\). Therefore, the co-exact component in the Hodge-Morrey decomposition of a closed 1-form disappears. Taking the Hodge decomposition for two closed \(1\)-forms \(s, \tilde{s} \in \Omega^{1} (\partial \ST_i)\):
\begin{align*}
    s = d \phi + \gamma  \, , & &  \tilde{s} = d \tilde{\phi} + \tilde{\gamma} \, ,
\end{align*}
where \(\tilde{\phi} \in \Omega^0(\partial\Lambda_i)\) and \(\tilde{\gamma} \in  \mathcal{H}^1 (\partial \ST_i)\).

Moving all terms in equation \eqref{eq:compute_2} to one side, we want to see if the result is zero. We begin this check by replacing the closed 1-forms with their Hodge-Morrey decomposition:
\begin{align*}
\begin{gathered}
    \int_{\partial \ST_i} \tilde{s} \wedge s - \sum_{j=1}^{\ell_i}  \left( \int_{\sigma_{i,j}} s \int_{\tau_{i,j}} \tilde{s} - \int_{\sigma_{i,j}} \tilde{s}    \int_{\tau_{i,j}} s \right) 
    \\ = \\   \int_{\partial \ST_i} \! (d \tilde{\phi} + \tilde{\gamma}) \wedge (d \phi + \gamma) -  \sum_{j=1}^{\ell_i}  \left( \int_{\sigma_{i,j}} \! \! (d \phi + \gamma) \int_{\tau_{i,j}} \! \! (d \tilde{\phi} + \tilde{\gamma}) - \int_{\sigma_{i,j}} \! \! (d \tilde{\phi} + \tilde{\gamma})    \int_{\tau_{i,j}} \! \! (d \phi + \gamma) \! \right)  \! . 
\end{gathered}
\end{align*}
Knowing that \(\partial \tau_{i,j} = \partial \sigma_{i,j} = 0\) we have, by Stokes' theorem:
\begin{align*}
    \int_{\tau_{i,j}} d \phi  = \int_{\partial \tau_{i,j}} \phi = 0 \, ,
\end{align*}
similarly for the integrals over \(\sigma_{i,j}\). Additionally:
\begin{align*}
    \int_{\partial \ST_i} d \tilde{\phi}  \wedge s
    =
    - \int_{\partial \ST_i}  \tilde{\phi}  \wedge ds
    +
    \int_{\partial^2 \ST_i}  \tilde{\phi}  \wedge s = 0 \, .
\end{align*}
Similarly for the term involving \(\tilde s \wedge d \phi  \). This leaves us with the relation:
\begin{equation}\label{eq:compute}
	\begin{aligned}
		& \int_{\partial \ST_i} \tilde{s} \wedge s - \sum_{j=1}^{\ell_i}  \left( \int_{\sigma_{i,j}} \! s \int_{\tau_{i,j}}  \!  \tilde{s} - \int_{\sigma_{i,j}} \!  \tilde{s}    \int_{\tau_{i,j}} \!  s \right) 
		\\=   & \int_{\partial \ST_i} \tilde{\gamma} \wedge \gamma -  \sum_{j=1}^{\ell_i}  \left( \int_{\sigma_{i,j}} \!  \gamma \int_{\tau_{i,j}}  \! \tilde{\gamma} - \int_{\sigma_{i,j}}  \! \tilde{\gamma}    \int_{\tau_{i,j}}  \! \gamma \right)  \, .
	\end{aligned}
\end{equation}
Replacing \(\gamma\) and \(\tilde{\gamma}\) with their classes in \(H^1_{\mathrm{dR}}(\partial \ST_i)\) given by \([\gamma]\) and \([\tilde{\gamma}]\) respectively, these integrals remain well-defined. Decomposing \([\gamma] \) and \([\tilde{\gamma}]\) in the dual to the Alexander basis we have that 
\begin{align*}
    [\gamma] = \sum_{j=1}^{\ell_i} c_{j}^1  [\sigma_{i,j}^*]+ c_{j}^2 [\tau_{i,j}^*] \, , & &
    [\tilde{\gamma}] = \sum_{j=1}^{\ell_i} \tilde{c}_{j}^1 [\sigma_{i,j}^*]+ \tilde{c}_{j}^2 [\tau_{i,j}^*] \, ,
\end{align*}
for constants \( c_{j}^1,  c_{j}^2, \tilde{c}_{j}^1\) and \(\tilde{c}_{j}^2\). Computing our wedge product between the gamma's:
\begin{align*}
    \int_{\partial \ST_i} \tilde{\gamma} \wedge \gamma  & = \scalebox{1.2}[1.4]{$\displaystyle\int$}_{ \! \! \! \! \! \partial \ST_i} \!  \Bigg(  \sum_{j=1}^{\ell_i} \tilde{c}_{j}^1 \sigma_{i,j}^*+ \tilde{c}_{j}^2 \tau_{i,j}^* \Bigg) \wedge \Bigg(   \sum_{j=1}^{\ell_i} c_{j}^1  \sigma_{i,j}^*+ c_{j}^2 \tau_{i,j}^*   \Bigg) \, .
\end{align*}
Recall from Theorem \ref{thr:Alexander} that: 
\begin{align*}
    \begin{gathered}
         \int_{\partial \ST_i} \tau_{i,j}^* \wedge \sigma_{i,k}^* = \delta_{jk}  \,, \qquad
        \int_{\partial \ST_i} \sigma_{i,j}^* \wedge \sigma_{i,k}^* = 
         0 \,,    \qquad 
         \int_{\partial \ST_i} \tau_{i,j}^* \wedge \tau_{i,k}^*
        = 0 \, ,
    \end{gathered}
\end{align*}   
so
\begin{align*}
\int_{\partial \ST_i} \tilde{\gamma} \wedge \gamma & 
  =  \sum_{j=1}^{\ell_i} 
  \tilde{c}_{j}^2  c_{j}^1  - \tilde{c}_{j}^1  c_{j}^2 
  \, . 
\end{align*}
And:
\begin{align*}
     \sum_{j=1}^{\ell_i}  \left( \int_{\sigma_{i,j}}  \!  \gamma \int_{\tau_{i,j}}  \!  \tilde{\gamma} - \int_{\sigma_{i,j}}  \!  \tilde{\gamma}    \int_{\tau_{i,j}}  \!  \gamma \right) 
     = \sum_{j=1}^{\ell_i}  \left( c^1_{j}  \tilde{c}^2_{j} - \tilde{c}^1_{j}  c^2_{j}  \right) 
     & = \int_{\partial \ST_i} \tilde{\gamma} \wedge \gamma \, . 
\end{align*}
Substituting this relation into equation \eqref{eq:compute} we get that:
\begin{align*}
    \int_{\partial \ST_i} \tilde{s} \wedge s - \sum_{j=1}^{\ell_i}  \left( \int_{\sigma_{i,j}}  \!  s \int_{\tau_{i,j}}  \!  \tilde{s} - \int_{\sigma_{i,j}}  \!  \tilde{s}    \int_{\tau_{i,j}}  \!  s \right) 
    &=   0 \, .
\end{align*}
If \(\ell_i = 0\) then $\partial\Lambda_i$ is a union of spheres and both forms \(s, \tilde s\) are exact. The result follows from Stokes' theorem.
\end{proof}

\subsection{Helicity}\label{subsection:Relative Helicity}

As noted in Subsection \ref{sub:setup}, helicity \(H_{\Omega^1(\ST_i) }( \beta_i , \ST_i ) \) is not, in general, a well-defined function of \( \beta _i \), owing to its dependence on the choice of gauge for the associated potential. One way to obtain a well-defined quantity is to restrict the admissible class of potentials \(G_i\) so that the helicity \(H_{G_i }( \beta_i , \ST_i ) \) is uniquely determined for any given \( \beta _i \), as is done to solve Problem \eqref{prob:MRxMHD_adiabatic} by \citet{10.1063/1.4765691}.

An alternative, which we adopt here, is to introduce a gauge invariant helicity. Recall the loop-corrected helicity, expressed using a closed reference potential with matching free loop integrals,
\[
    \mathscr{H}_{a_i} ( \beta_i,  \ST_i  ) = \int_{ \ST_i } ( a_i - a_{\text{ref}} ) \wedge  \beta_i  \, .
\]
\begin{definition}\label{def:a_ref}
    The potential \(a_{\text{ref}} \in \Omega^1 (\ST_i) \) of \( \beta_{\text{ref}} = 0\) for each \(i = 1, \dots, n\) is given by any 1-form satisfying
    \[
    d a_{\text{ref}} = 0 \, , \qquad \int_{ \sigma_{i,j} } \!a_{\text{ref}} = \int_{\sigma_{i,j}}a_i \, , \; \forall j = 1, \dots , \ell_i \, ,
    \]
    where the latter condition is null if \(\ell_i = 0\).
\end{definition}
In this section we rewrite this in terms of the helicity with loop corrections and see that the definition is independent of the Alexander basis. In the next subsection we show explicitly that this quantity is gauge independent and that, like helicity, this is invariant under pushing forward by \(f \in \I\).

Additionally, we will see that \([a_{\text{ref}}]\) is unique, and the loop-corrected helicity matches the helicity with an Amperian gauge. The Amperian gauge is defined in this section and determines a primitive of each \(\beta_i\) uniquely modulo exact \(1\)-forms. We see that the loop-corrected helicity retains a clear physical interpretation without the need to impose additional constraints on the potential. The construction is known in the toroidal case, what is new is its statement on a general region \(\ST_i\), and the identification of the loop corrections.

Firstly, recalling that \(a_{ \text{ref}}\) is closed and its integral about each \(\sigma_{i,j}\) is the same as that of \(a_i\),
\begin{align*}
    \mathscr{H}_{a_i} ( \beta_i,  \ST_i  ) &= H_{G_i}( \beta_i,  \ST_i ) - \int_{ \ST_i }  a_{\text{ref}}  \wedge  \beta_i = H_{G_i}( \beta_i,  \ST_i ) + \int_{ \partial \ST_i}   a_{\text{ref}}  \wedge  a_i   \, .
\end{align*}
Seeing that \(a_{\text{ref}}\) and \(a_i\) are closed on the boundary, Lemma \ref{Lemma:compute} applies, so for \(\ell_i > 0\)
\begin{align*}
    \mathscr{H}_{a_i} ( \beta_i,  \ST_i  ) &=  H_{G_i}( \beta_i,  \ST_i ) + \sum_{j=1}^{\ell_i}  \left(\int_{\sigma_{i,j}} \! a_i \int_{\tau_{i,j}}  \! a_{\text{ref}} -\int_{\sigma_{i,j}}  \! a_{\text{ref}}   \int_{\tau_{i,j}}  \! a_i   \right)
    \\
    &=  H_{G_i}( \beta_i,  \ST_i ) - \sum_{j=1}^{\ell_i}  \int_{\sigma_{i,j}}  \! a_i  \int_{\tau_{i,j}}  \! a_i
    \, ,
\end{align*}
and \(\mathscr{H}_{a_i}(\beta_i, \ST_i) = H_{G_i} (\beta_i, \ST_i)\) for \(\ell_i = 0\), and \(G_i=\Omega^1(\ST_i)\). Refer to Lemma \ref{Lemma:well-pose} to see that the loop-corrected helicity is independent of the chosen primitive for a given \(\beta _i\). We have used that the integral of a closed 1-form \(a_{\text{ref}}\) over \(\tau_{i,j}\) is zero by Stokes' theorem.

\begin{remark}\label{rem:basis_independence}
  An Alexander basis is determined by the choice of basis \(\{[s_{i,j}]\}\) of \(H_1(\ST_i)\), the remaining classes being fixed by Theorem \ref{thr:Alexander}. The loop-corrected helicity does not depend on that choice of basis. If 
\[
    [s_{i,j}'] = \sum_k M_{jk} [s_{i,k}]
\]
 for \(M \in \mathrm{GL}(\ell_i, \R)\), then \([\sigma_{i,j}']\) transforms by \(M\), whereas the dual basis \(\{[s^*_{i,j}]\}\), and with it \(\{[T_{i,j}]\}\), \(\{[\tau_{i,j}]\}\) and the fluxes \(\psi_{i,j}\), transforms by \(M^{-\top}\). Writing \(c_j = \langle a_i , \sigma_{i,j} \rangle \), the correction term is the pairing 
\[
    \sum_j \psi_{i,j} c_j\, ,
\]
 and temporarily trading the \(j\) subscript for matrix multiplication operations
  \begin{align*}
    \sum_{j=1}^{\ell_i} \psi_{i,j}' \, c_j'
    = \big( M^{-\top} \psi_i \big)^{\!\top} \big( M c \big)
    = \psi_i^\top M^{-1} M c
    = \sum_{j=1}^{\ell_i} \psi_{i,j} \, c_j \, .
  \end{align*}
  For the same reason the Amperian gauge \eqref{eq:amperian} is independent of the basis: the span of \(\{[\sigma_{i,j}]\}\) in \(H_1(\partial \ST_i)\) is the image of \(H_1(\ST_i)\) under the splitting of Lemma \ref{Lemma:isomorphism_inclusion}, and so is intrinsic.
\end{remark}
\begin{example}[Toroidal regions]\label{ex:toroidal_helicity}
  Let \(\ST_i\) be a hollow torus, so that \(\ell_i = 2\), with the Alexander basis of Appendix \ref{append:example}, \(\sigma_{i,1}\) is a toroidal loop and \(\sigma_{i,2}\) a poloidal loop on \(\partial \ST_i\), while \(\tau_{i,1}\) bounds the poloidal cross-section \(T_{i,1}\) and \(\tau_{i,2}\) the toroidal ribbon \(T_{i,2}\), carrying the toroidal and poloidal fluxes \(\psi_{i,1} \) and \(\psi_{i,2}\) respectively. Since \(\tau_{i,j} = \partial T_{i,j}\), Stokes' theorem gives \(\langle a_i , \tau_{i,j} \rangle  = \psi_{i,j}\), so that the correction term collapses to (for an appropriate \(G_i\))
  \begin{align*}
    \mathscr{H}_{a_i} ( \beta_i, \ST_i )
    = H_{G_i}( \beta_i, \ST_i )
      - \psi_{i,1} \int_{\sigma_{i,1}} \! a_i
      - \psi_{i,2} \int_{\sigma_{i,2}} \! a_i \, .
  \end{align*}
  In vector notation (see Appendix \ref{append:vector_notation}), writing \(A_i = \big( (f^{-1})^* a_i \big)^{\sharp}\) for the vector potential on \(f(\ST_i)\) corresponding to \(a_i\), this reads
  \begin{align*}
    \mathscr{H}_{a_i} ( \beta_i, \ST_i )
    = \int_{f(\ST_i)} A_i \cdot B_i \, \varpi
      - \psi_{i,1} \oint_{f(\sigma_{i,1})} \! A_i \cdot d\ell
      - \psi_{i,2} \oint_{f(\sigma_{i,2})} \! A _i \cdot d\ell \, ,
  \end{align*}
  namely the helicity integral corrected by the circulation of \(A_i\) around each boundary loop \(f(\sigma_{i,j})\) of \(\partial f(\ST_i)\), weighted by the flux through the matching surface as in Subsection \ref{sub:coho_homo}. If \(\ST_i\) is a solid torus then \(\ell_i = 1\), the single loop \(\sigma_{i,1}\) is toroidal, \(T_{i,1}\) is a meridional disc, and only the first two terms survive. This is precisely the gauge invariant toroidal helicity of Spies, Lortz and Kaiser \citeyearpar{spies2001relaxed} and of Kaiser and Uecker \citeyearpar{kaiser2004relaxed}.

\end{example}

The content of Example \ref{ex:toroidal_helicity} is known and is the form given by Pfefferl\'e, Noakes, and Perrella \citeyearpar{Gauge_freedom} up to a difference in the sign convention. 

We refer to the following set as an Amperian gauge:
\begin{align}\label{eq:amperian}
    \mathscr{A}(\ST_i) = \bigg\{ a_i \in  \Omega^1 (\ST_i) \, : \; \int_{\sigma_{i,j}} a_i = 0 \, , \; \forall j \bigg\} \, .
\end{align}
An Amperian primitive for \(\beta _i\) is any \(a_i \in \mathscr{A}(\ST_i) \) such that \(d a_i = \beta _i\).


Selecting an Amperian gauge reduces the loop-corrected helicity to helicity,
\begin{align*}
    H_{\mathscr{A}(\ST_i)}(\beta_i,\ST_i) = \mathscr{H}_{a_i}(\beta_i,\ST_i)  \, ,
\end{align*}
for any primitive \(a_i \in \Omega^1 (\ST_i)\) of \(\beta_i\). Intuitively, the loop-corrected helicity encodes the helicity in an Amperian gauge. Therefore, any result using a loop-corrected helicity also holds for helicity under a restriction of the primitives to an Amperian gauge.

For the reader's convenience the definitions of loop-corrected helicity and the Amperian gauge are transcribed to vector calculus notation in Appendix \ref{append:vector_notation}.

\subsection{Well-Posed}\label{sub:Well-Posed}

\begin{lemma}\label{Lemma:well-pose}
    For a magnetic field represented by a closed, Dirichlet 2-form \(\beta _i \in \C\) with potential \(a_i \in  \Omega^1(\ST_i)\), the loop-corrected helicity \(\mathscr{H}_{a_i}( \beta_i, \ST_i)\) is independent of the choice of primitive \(a_i\). Additionally, if \(f \in \I\) then \(\mathscr{H}_{a_i}( \beta_i, \ST_i)= \mathscr{H}_{(f^{-1})^* a_i}( (f^{-1})^*  \beta_i, f(\ST_i ) )\). 
\end{lemma}

\begin{proof} 
Consider \(a_i, \tilde{a}_i\in  \Omega^1 (\ST_i)\) such that \(da_i = d \tilde{a}_i = \beta_i \in \C\). Then the difference in loop-corrected helicities applied to each primitive is:
\begin{align*}
    & \mathscr{H}_{a_i}(\beta_i, \ST_i) - \mathscr{H}_{\tilde{a}_i}(\beta_i, \ST_i) 
    \\ = & \int_{\partial \ST_i} \mathcal{J}^* \tilde{a}_i \wedge \mathcal{J}^* a_i
    - \left( \sum_{j=1}^{\ell_i} \int_{\sigma_{i,j}}  \!  \mathcal{J}^* a_i\int_{\tau_{i,j}}   \!  \mathcal{J}^* a_i
    - \int_{\sigma_{i,j}}  \!  \mathcal{J}^* \tilde{a}_i \int_{\tau_{i,j}}  \!  \mathcal{J}^* \tilde{a}_i \right)
    \, .
\end{align*}
The forms \(\mathcal{J}^* (\tilde{a}_i)\) and \(\mathcal{J}^* (a_i)\) are closed 1-forms in \(\Omega^1(\partial \ST_i)\), so we may apply equation \eqref{eq:compute_2}:
\begin{align*}
    \mathscr{H}_{a_i}(\beta_i, \ST_i) - \mathscr{H}_{\tilde{a}_i}(\beta_i, \ST_i) = -
    \sum_{j=1}^{\ell_i}  \Bigg( & \int_{\sigma_{i,j}}  \! \mathcal{J}^* \tilde{a}_i    \int_{\tau_{i,j}}  \! \mathcal{J}^*a_i
    -
    \int_{\sigma_{i,j}}  \!  \mathcal{J}^*a_i \int_{\tau_{i,j}}  \! \mathcal{J}^*\tilde{a}_i 
    \\
    & \! \! 
    + \int_{\sigma_{i,j}}  \! \mathcal{J}^*a_i \int_{\tau_{i,j}}  \! \mathcal{J}^*a_i
    - \int_{\sigma_{i,j}}  \! \mathcal{J}^*\tilde{a}_i \int_{\tau_{i,j}}  \! \mathcal{J}^*\tilde{a}_i \Bigg) \, .
\end{align*}
By construction, for each \(j\), there exists a \(T_{i,j}\) where \(\partial T_{i,j} = \tau_{i,j}\), therefore:
\begin{align*}
    \int_{\tau_{i,j}} \mathcal{J}^* \tilde{a}_i = \int_{T_{i,j}} \mathcal{J}^* \beta_i = \int_{\tau_{i,j}} \mathcal{J}^* a_i \, .
\end{align*}
Applying this to the difference between loop-corrected helicities computed in different gauges gives:
\begin{align*}
    \mathscr{H}_{a_i}(\beta_i, \ST_i) - \mathscr{H}_{\tilde{a}_i}(\beta_i, \ST_i) 
    = 0 \, .
\end{align*}
Hence, the loop-corrected helicity is independent of the gauge. Due to this, we now drop the subscript of \(\mathscr{H}\) that specifies the potential.

Additionally, for \(f \in \I\) and an appropriate \(G \subseteq \Omega^1(\ST_i)\): 
\begin{align*}
\mathscr{H}( \beta_i, \ST_i ) - \mathscr{H}( (f^{-1})^* \beta_i, f(\ST_i ) ) = &H_G( \beta_i, \ST_i ) - H_{(f^{-1})^* G}( (f^{-1})^* \beta_i, f(\ST_i ) )  \\
& \hspace{-3em} - \sum_{j=1}^{\ell_i} \int_{\sigma_{i,j}} a_i \int_{\tau_{i,j}}  a_i
 +
 \sum_{j=1}^{\ell_i} \int_{f(\sigma_{i,j})} (f^{-1})^*a_i \int_{f(\tau_{i,j})} (f^{-1})^* a_i
 \\
 = &H_G( \beta_i, \ST_i ) - H_{(f^{-1})^* G}( (f^{-1})^* \beta_i, f(\ST_i ) ) 
 \,  .
\end{align*}
We have the freedom to select a \(G\) without altering the loop-corrected helicity. Hence, we restrict primitives to a Biot--Savart gauge, as defined by \citet{CANTARELLA20101127}. With this gauge Cantarella and Parsley show that \(H_G( \beta_i, \ST_i ) = H_{(f^{-1})^* G}( (f^{-1})^* \beta_i, f(\ST_i ) ) \). Hence, \(\mathscr{H}( \beta_i, \ST_i ) = \mathscr{H}( (f^{-1})^* \beta_i, f(\ST_i ) )\) must hold in any gauge.
\end{proof}

By Lemma \ref{lem:vector-dictionary}, \((f^{-1})^* \beta_i = i_{B_i} \varpi\) and \(\big( (f^{-1})^* a_i \big)^{\sharp} = A_i\), so the second assertion of Lemma \ref{Lemma:well-pose} identifies \(\mathscr{H}(\beta_i, \ST_i)\) with the loop-corrected helicity of the field \(B_i\) on \(f(\ST_i)\), written \(\mathscr{H}( B_i, f(\ST_i) )\) in Appendix \ref{append:vector_notation}.

\begin{corollary}
    The objective and constraints in the variational Problem \eqref{eq:ext} are independent of the chosen gauge. Namely, the problem is well-defined in terms of the magnetic field \(\beta_i\) and \(f \in \I\). 
\end{corollary}

\begin{lemma} \label{Lemma:gauge_exist}
    For every \(\beta_i \in \C\), there exists an Amperian primitive \(a_i \in \mathscr{A}(\ST_i)\) with \(da_i = \beta_i\), unique modulo exact \(1\)-forms. 
\end{lemma}
\begin{proof}
There exists at least one primitive \(a_i \in  \Omega^1 (\ST_i)\) such that \(da_i = \beta_i\) (Lemma \ref{Lemma:Potential}). Additionally, \(a_i + s\) is also a primitive for \(\beta_i\) for any closed 1-form \(s \in  \Omega^1 (\ST_i)\). Hence, \([s]\) is any element of \(H^1_{\mathrm{dR}} (\ST_i)\) and therefore has a decomposition:
\begin{align*}
    [s] = \sum_{j=1}^{\ell_i} c_{i,j} [ s_{i,j}^* ] \, . 
\end{align*}
Pulling back to the boundary via inclusion gives
\begin{align*}
    \mathcal{J}^* [s] = \sum_j c_{i,j}  \mathcal{J}^* [s_{i,j}^*]  = \sum_j c_{i,j} [\sigma_{i,j}^*] \, . 
\end{align*}
Giving
\begin{align*}
    \int_{\sigma_{i,j}} \mathcal{J}^* (a_i + s) =  \int_{\sigma_{i,j}} \mathcal{J}^* a_i + \int_{\sigma_{i,j}}  \sum_{k=1}^{\ell_i} c_{i,k}  \sigma_{i,k}^*   = 
    c_{i,j} + \int_{\sigma_{i,j}} \mathcal{J}^* a_i    \, .
\end{align*}
Setting \(c_{i,j}\) so that the right-hand side is zero uniquely specifies \([s]\) and we have that \(a_i + s\) is an Amperian primitive for \(\beta_i\). If two Amperian primitives have the same exterior derivative, their difference is closed and has zero periods on the generators \(\sigma_{i,j}\), and is therefore exact. Conversely, adding an exact \(1\)-form preserves both the exterior derivative and the Amperian conditions.
\end{proof}

In the same vein we see that when we set \(d a_{\text{ref}} = 0\) and selected the integral of \( a_{\text{ref}}\) over \(\sigma_{i,j}\) to match that of \(a\), this uniquely specifies \([a_{\text{ref}}]\) for a given \(a\).

Lemma \ref{Lemma:gauge_exist} is important because the loop-corrected helicity given any gauge always aligns with the helicity for an Amperian gauge. 

\begin{remark}\label{remark:helicity}
    Given a \(\ST_i\) with non-trivial genus and \(\beta_i \in \C\), the helicity \(H_{G} ( \beta_i, \ST_i)\) can take any value via a choice of \(G \subseteq \Omega^1(\ST_i) \) if \(\ell_i > 0\) and \(\psi_{i,1}, \dots, \psi_{i,\ell_i} \) are not all identically zero.
\end{remark}

We consider this remark in more detail below. Under the conditions in remark \ref{remark:helicity}, \(\beta_i \neq 0\) and \(\ell_i >0\). Consider two primitives of \(\beta_i\) given by \(a_i \in G\) and \(a_i + s \in G'\) for a closed 1-form \(s \in  \Omega^1(\ST_i)\). Both \(s\) and \(\mathcal{J}^* a_i\) are closed on the boundary so formula \eqref{eq:compute_2} is applicable:
\begin{align*}
    H_G (\beta_i, \ST_i) - H_{G'} (\beta_i, \ST_i)  =  \int_{\partial \ST_i}  \!  s \wedge a_i &=  - \sum_{j=1}^{\ell_i}  \left(\int_{\sigma_{i,j}} \! s    \int_{\tau_{i,j}}  \! a_i -\int_{\sigma_{i,j}}  \! a_i \int_{\tau_{i,j}}  \! s  \right) 
    \\
    &
    =  - \sum_{j=1}^{\ell_i}  \left( \psi_{i,j} \int_{\sigma_{i,j}}  \!  s    \right) 
     \, .
\end{align*}
Here we have used that the integral of a closed 1-form \(s\) over \(\tau_{i,j}\) is zero by Stokes' theorem. Therefore, if any of the fluxes \(\psi_{i,j}\) are non-trivial, we may select \(s\) to make the difference in helicity any value (we saw a comparable situation in the proof of Lemma \ref{Lemma:gauge_exist}). However, selecting an Amperian gauge for the primitives above induces that the helicity is independent of the specific primitive choice.

Therefore, if we consider Problem \eqref{prob:MRxMHD_adiabatic} without constraining the potential for non-trivial flux: (1) the constraint \(H_{\Omega^1(\ST_i)} (\beta_i , \ST_i) = h_i\) is not well-defined as a function of \(\beta_i\) and (2) the helicity can take any value, so this constraint does not restrict \(\beta_i\) in any way. 

\subsection{Jump Conditions}

The boundary and jump conditions that we are using in equations \eqref{prob:MRxMHD} are reported as a special case of the following conditions given by Bruno and Laurence \citeyearpar{bruno_laurence_1996}:
\begin{align*}
    0 & = \llbracket B \cdot N \rrbracket \, ,  \\
    0 & = \big\llbracket \tfrac12 \| B \| ^2 + p \big\rrbracket N -  (B \cdot N) \llbracket B \rrbracket \, \, ,
\end{align*}
on \(\partial f(\ST_i)\) where \(N\) is the outward unit normal vector field to the boundary. The first of these conditions gives \(\llbracket B \rrbracket \cdot N = \llbracket B \cdot N \rrbracket = 0\), so \(\llbracket B \rrbracket\) is tangential to \(\partial f(\ST_i)\). The two terms in the second equation are therefore orthogonal, namely, the first term is parallel to \(N\) and the second term lies in the tangent space of \(\partial f(\ST_i)\), and both must be \(0\) independent of each other. Hence \(\llbracket \tfrac12 \| B \|^2 + p \rrbracket = 0\) in either case, while \((B \cdot N) \llbracket B \rrbracket = 0\) leaves two possibilities. Consider these pointwise on the boundary, either \(B\) is continuous across the boundary, in which case \(\llbracket \| B \|^2 \rrbracket = 0\) and so \(\llbracket p \rrbracket = 0\), or \(B \cdot N = 0\). Under the assumption that \(\llbracket p \rrbracket \neq 0\) across our interfaces \(f(I_q)\), the original conditions from Bruno and Laurence are equivalent to the boundary conditions in \eqref{prob:MRxMHD}.

\subsection{Lagrange Multipliers} \label{sub:LM}


We know that any \(\beta_i \in \C\) for \(\ell_i>0\) can be written as:
\begin{align} \label{eq:b_form}
    \beta_i = d \eta_i + \sum_{j=1}^{\ell_i}  \psi_{i,j}  T_{i,j}^*  
    = d \Bigg( \eta_i + \sum_{j=1}^{\ell_i}  \psi_{i,j} \Gamma_{i,j}  \Bigg)\, ,
\end{align}
for \(\eta_i \in  \Omega_D^1(\ST_i)\), \( T_{i,j}^*  \in \mathcal{H}_D^2 (\ST_i) \), and \(\Gamma_{i,j} \in  \Omega^1 (\ST_i)\). The second equality holds because there always exists a primitive \(\Gamma_{i,j}\) for \( T_{i,j}^* \) (Lemma \ref{Lemma:Potential}).

In the variational formulation \eqref{eq:ext} we incorporate the flux constraints directly into the objective. We see from equation \eqref{eq:b_form} that the harmonic component of \(\beta_i\) (which determines the flux of \(\beta_i\) through any relative \(2\)-chain) is fixed once the coefficients \(\psi_{i,1},\dots,\psi_{i,\ell_i}\) are prescribed, provided \(\ell_i>0\). Consequently, the objective and helicity are functions of \(f\in\I\) and \(a_i \). Assuming that the potential for the harmonic component is fixed, the variational Problem \eqref{eq:ext} can be posed over \(\eta_i\). The potential \(\eta_i\) is not always unique, different \(\eta_i\) may have the same exterior derivative \(\beta_i\), hence the same value for the objective and loop-corrected helicity. To incorporate variations of the potentials we formulate the optimisation over \(f\) and \(\eta_i\), giving
\begin{align}\label{prob:start_point}
\begin{gathered}
    \underset{\substack{  \eta_i \in  \Omega^1_D(\ST_i)  \\  f\in\I} }{ \text{c. p.} } \sum_{i=1}^n \Bigg( \frac{1}{2}\Big\| d\eta_i + \sum_{j=1}^{\ell_i} \psi_{i,j}\,T_{i,j}^* \Big\|_{L^2}^2 - \tilde p_i\,\myabs{f(\ST_i)} \Bigg) \\
  \text{subject to}\qquad 
   \mathscr{H}\!\Big( d\eta_i + \sum_{j=1}^{\ell_i} \psi_{i,j}\,T_{i,j}^* ,\; \ST_i\Big) = h_i,\qquad i=1,\dots,n \, . 
\end{gathered}
\end{align}
If \(\ell_i=0\) there is no harmonic contribution in the displayed expressions.

For the remainder of this section we freeze the deformation \(f\). The variational problem reduces to finding admissible potentials \(\eta_i\). A standard way to treat the constraint in \eqref{prob:start_point} is to introduce Lagrange multipliers.

Both the objective and the constraint in \eqref{prob:start_point} are inhomogeneous quadratic polynomials in \(\eta_i\), namely a quadratic form plus a linear functional plus a constant. Their increments are therefore exact algebraic identities carrying no remainder term, and the first variation along a path in the sense of Section \ref{section:Main Results} is obtained by differentiating a real polynomial in the path parameter. Consequently no topology on \(\Omega^1_D(\ST_i)\), and no notion of differentiability of maps defined on that space, is required at any point below. The only place at which anything analytic enters is Lemma \ref{lem:nondegenerate}.

Throughout this subsection we fix \(i\) and write \(\Gamma_{i,j}\) for the potential of \(T^*_{i,j}\) and
\begin{align} \label{eq:Gamma_i}
    \Gamma_i \coloneqq \sum_{j=1}^{\ell_i} \psi_{i,j} \Gamma_{i,j} \; \in \; \Omega^1(\ST_i) \, ,
\end{align}
so that 
\[
    d \Gamma_i = \sum_{j} \psi_{i,j} T^*_{i,j}
\]
 is the harmonic part fixed by the flux data, and \(\beta_i = d ( \eta_i + \Gamma_i )\) by \eqref{eq:b_form}. When \(\ell_i = 0\) we have \(d \Gamma_i = 0\). We further write
\begin{align} \label{eq:Qform}
    \mathcal{Q}( \eta , \eta') \coloneqq \int_{\ST_i} \eta \wedge d \eta' \, , \qquad
    \mathcal{Q}( \eta ) \coloneqq \mathcal{Q}( \eta , \eta ) \, , \qquad \eta , \eta' \in \Omega^1_D(\ST_i) \, .
\end{align}
Our constraint function is \(g_i : \Omega^1_D(\ST_i) \to \R\), given by \(g_i : \eta_i \mapsto \mathscr{H}(\beta_i, \ST_i)\), and our objective at fixed \(f\) up to a constant is \(F : \Omega^1_D(\ST_i) \to \R\), given by \(F : \eta_i \mapsto (\sfrac12) \|\beta_i\|^2_{L^2}\). For a fixed \(\eta_i \in \Omega^1_D(\ST_i)\), with \(\beta_i = d(\eta_i + \Gamma_i)\), we \emph{define} the two linear functionals
\begin{align} \label{eq:deriv}
    d_{\eta_i} g_i ( w ) \coloneqq 2 \int_{\ST_i} w \wedge \beta_i \, , \qquad
    d_{\eta_i} F ( w ) \coloneqq \big\langle d w , \beta_i \big\rangle_{L^2} \, , \qquad w \in \Omega^1_D(\ST_i) \, .
\end{align}
\begin{lemma} \label{lem:quad}
    Let \(\eta_i , w \in \Omega^1_D(\ST_i)\) and put \(\beta_i = d ( \eta_i + \Gamma_i )\). Then the constraint for Problem \eqref{prob:start_point} is
        \begin{align}
            g_i( \eta_i + w ) &= g_i( \eta_i ) + d_{\eta_i} g_i ( w ) + \mathcal{Q}( w ) \, . \label{eq:g_increment}
        \end{align}
        The corresponding objective is \(F( \eta_i + w ) = F( \eta_i ) + d_{\eta_i} F ( w ) + (\sfrac12) \| d w \|^2_{L^2}\). Let \(\{ \eta_{i,t} \}_{t \in (-\varepsilon,\varepsilon)} \subseteq \Omega^1_D(\ST_i)\) with a smooth map \((t,x) \mapsto \eta_{i,t}|_x\), \(\eta_{i,0} = \eta_{i}\) and \(\partial_t \eta_{i,t}|_{t=0} = \eta_i'\). Then \(t \mapsto g_i(\eta_{i,t})\) and \(t \mapsto F(\eta_{i,t})\) are smooth and
        \begin{align} \label{eq:first_variation}
            \partial_t g_i( \eta_{i,t} ) \big|_{t=0} = d_{\eta_i} g_i ( \eta_i' ) \, , \qquad
            \partial_t F( \eta_{i,t} ) \big|_{t=0} = d_{\eta_i} F ( \eta_i' ) \, .
        \end{align}
\end{lemma}

\begin{proof}
By the definition of loop-corrected helicity and \(\kappa_i\) \eqref{eq:kappa},
\[
    g_i(\eta_i) = \int_{\ST_i} ( \eta_i + \Gamma_i ) \wedge d ( \eta_i + \Gamma_i ) - \kappa_i( \eta_i +  \Gamma_i ) = \int_{\ST_i} ( \eta_i + \Gamma_i ) \wedge d ( \eta_i + \Gamma_i ) - \kappa_i( \Gamma_i ) \, .
\]
Expanding the wedge product gives
\[
    \mathcal{Q}( \eta_i ) + \int_{\ST_i} \eta_i \wedge d \Gamma_i + \int_{\ST_i} \Gamma_i \wedge d \eta_i + \int_{\ST_i} \Gamma_i \wedge d \Gamma_i \, ,
\]
and since \(d ( \Gamma_i \wedge \eta_i ) = \eta_i \wedge d \Gamma_i - \Gamma_i \wedge d \eta_i\) with \(\mathcal{J}^* \eta_i = 0\), Stokes' theorem gives
\[\int_{\ST_i} \Gamma_i \wedge d \eta_i = \int_{\ST_i} \eta_i \wedge d \Gamma_i \, .\]
So
\[
g_i(\eta_i) = \mathcal{Q}( \eta_i ) + 2 \int_{\ST_i} \eta_i \wedge d \Gamma_i + \int_{\ST_i} \Gamma_i \wedge d \Gamma_i  - \kappa_i( \Gamma_i ) \, .
\]
And 
\[
g_i(\eta_i + w) = \mathcal{Q}( \eta_i + w ) + 2 \int_{\ST_i} (\eta_i + w) \wedge d \Gamma_i + \int_{\ST_i} \Gamma_i \wedge d \Gamma_i  - \kappa_i( \Gamma_i ) \, .
\]
Applying the Stokes theorem we see that,
\[
    \mathcal{Q}( \eta_i + w ) = \mathcal{Q}( \eta_i ) + 2 \int_{\ST_i} w \wedge d \eta_i + \mathcal{Q}( w ) \, .
\]
Taking \(g_i(\eta_i + w)\) and subtracting \(g_i(\eta_i)\), the constant terms cancel and
\[
    g_i(\eta_i + w) - g_i(\eta_i)
    = 2 \int_{\ST_i} w \wedge ( d \eta_i + d \Gamma_i ) + \mathcal{Q}( w )
    = d_{\eta_i} g_i ( w ) + \mathcal{Q}( w ) \, ,
\]
which is \eqref{eq:g_increment}.

Put \(w(t) \coloneqq \eta_{i,t} - \eta_i\), which is smooth in \((t,x)\) with \(w(0) = 0\) and \(\partial_t w(t)|_{t=0} = \eta_i'\). By \eqref{eq:g_increment},
\[
    g_i(\eta_{i,t}) - g_i(\eta_i) = d_{\eta_i} g_i \big( w(t) \big) + \mathcal{Q}\big( w(t) \big) \, .
\]
Each term on the right is an integral over the compact \(\ST_i\) of an integrand smooth in \((t,x)\), so \(t \mapsto g_i(\eta_{i,t})\) is smooth and may be differentiated under the integral sign. Linearity gives \(\partial_t d_{\eta_i} g_i (w(t))|_{t=0} = d_{\eta_i} g_i(\eta_i')\), while bilinearity and symmetry give \(\partial_t \mathcal{Q}(w(t))|_{t=0} = 2 \mathcal{Q}( w(0) , \eta_i' ) = 0\). 
\end{proof}

\begin{lemma} \label{lem:nondegenerate}
    Let \(\eta_i \in \Omega^1_D(\ST_i)\) with \(\beta_i = d ( \eta_i + \Gamma_i ) \neq 0\). Then there exists \(\nu \in \Omega^1_D(\ST_i)\) with \(d_{\eta_i} g_i ( \nu ) = 1\).
\end{lemma}

\begin{proof}
Since \(\beta_i\) is smooth and not identically zero, and \(\text{int}(\ST_i)\) is dense in \(\ST_i\), there is a point \(x_0 \in \text{int}(\ST_i)\) with \(\beta_i|_{x_0} \neq 0\). Choose a smooth, real valued map \(\chi\) on \( \ST_i \) that is zero on the boundary with \(\chi \geq 0\) and \(\chi(x_0) > 0\), and set \(\nu_0 \coloneqq \chi \, \star \beta_i \in \Omega^1(\ST_i)\). Then \(\nu_0 \in \Omega^1_D(\ST_i)\), and since \(\star \star = 1\) on \(\Omega^2(\ST_i)\),
\begin{align*}
    d_{\eta_i} g_i ( \nu_0 )
    = 2 \int_{\ST_i} \chi \, ( \star \beta_i ) \wedge \star ( \star \beta_i )
    = 2 \int_{\ST_i} \chi \, \| \beta_i \|^2 \, f^* \varpi
    > 0 \, ,
\end{align*}
the integrand being non-negative everywhere and strictly positive near \(x_0\). Take \(\nu \coloneqq \nu_0 / d_{\eta_i} g_i ( \nu_0 )\).
\end{proof}

The next lemma is what the regular value theorem would otherwise supply when applying the method of Lagrange multipliers. It identifies the admissible variations of \eqref{prob:start_point} at fixed \(f\) as the kernel of \(d_{\eta_i} g_i\), and it does so by solving a scalar quadratic in one real unknown.

\begin{lemma} \label{lem:tangent}
    Let \(\eta_i \in \Omega^1_D(\ST_i)\) with \(\beta_i = d ( \eta_i + \Gamma_i ) \neq 0\), with \(\Gamma_i\) given and constrained helicity, \(h_i = g_i(\eta_i)\). For \(\eta_i' \in \Omega^1_D(\ST_i)\) we have that \(d_{\eta_i} g_i ( \eta_i' ) = 0\) is equivalent to: There exist an \(\varepsilon > 0\) and a family \(\{ \eta_{i,t} \}_{t \in (-\varepsilon,\varepsilon)} \subseteq \Omega^1_D(\ST_i)\), with \((t,x) \mapsto \eta_{i,t}|_x\) smooth, such that \(\eta_{i,0} = \eta_i\), \(g_i(\eta_{i,t}) = h_i\) for every \(t\), and \(\partial_t \eta_{i,t}|_{t=0} = \eta_i'\).
     
\end{lemma}

\begin{proof}
Suppose the latter holds. Then \(t \mapsto g_i(\eta_{i,t})\) is constant, so by Lemma \ref{lem:quad} its derivative at \(t=0\), namely \(d_{\eta_i} g_i(\eta_i')\), vanishes.

Suppose \(d_{\eta_i} g_i(\eta_i')=0\) holds. If \(\eta_i' = 0\) take \(\eta_{i,t} \equiv \eta_i\). Otherwise, take \(\nu\) from Lemma \ref{lem:nondegenerate} and we look for a family of the form
\begin{align*}
    \eta_{i,t} \coloneqq \eta_i + t \, \eta_i' + s(t) \, \nu \, , \qquad s(0) = 0 \, .
\end{align*}
By \eqref{eq:g_increment} with \(w = t \eta_i' + s \nu\), using \(d_{\eta_i} g_i ( \eta_i' ) = 0\), \(d_{\eta_i} g_i ( \nu ) = 1\) and the bilinearity of \(\mathcal{Q}\),
\begin{align} \label{eq:scalar_quadratic}
    g_i( \eta_{i,t} ) - h_i
    = s + \mathcal{Q}( \nu ) \, s^2 + 2 \mathcal{Q}( \eta_i' , \nu ) \, t s + \mathcal{Q}( \eta_i' ) \, t^2 \, . 
\end{align}
Write \(q(t) \coloneqq 1 + 2 \mathcal{Q}( \eta_i' , \nu ) \, t\) and solve \eqref{eq:scalar_quadratic} for the branch through the origin. If \(\mathcal{Q}(\nu) = 0\) take
\begin{align*}
    s(t) = \frac{ - \mathcal{Q}( \eta_i' ) \, t^2 }{ q(t) } \, ,
\end{align*}
and if \(\mathcal{Q}(\nu) \neq 0\) take
\begin{align*}
    s(t) = \frac{ - q(t) + \sqrt{ q(t)^2 - 4 \, \mathcal{Q}( \nu ) \, \mathcal{Q}( \eta_i' ) \, t^2 } }{ 2 \, \mathcal{Q}( \nu ) } \, .
\end{align*}
In both cases \(q(0) = 1\) and the radicand equals \(1\) at \(t = 0\), so \(s\) is defined and smooth on some interval \((-\varepsilon,\varepsilon)\), with \(s(0) = 0\) and the right-hand side of equation \eqref{eq:scalar_quadratic} equal to 0 there. Rearranging \eqref{eq:scalar_quadratic},
\begin{align*}
    s(t) = - \mathcal{Q}( \eta_i' ) \, t^2 - 2 \mathcal{Q}( \eta_i' , \nu ) \, t s(t) - \mathcal{Q}( \nu ) \, s(t)^2 \, ,
\end{align*}
and since \(s\) is smooth with \(s(0) = 0\) the right-hand side is \(\textit{O}(t^2)\). Hence \(s'(0) = 0\) and \(\partial_t \eta_{i,t}|_{t=0} = \eta_i'\). Smoothness in \((t,x)\) is clear because \(\eta_i'\) and \(\nu\) are fixed smooth forms and \(s\) is smooth.
\end{proof}

\begin{corollary} \label{cor:multiplier}
    Let \(\eta_i \in \Omega^1_D(\ST_i)\) with \(\beta_i \neq 0\) satisfy \(g_i(\eta_i) = h_i\). Then \(d_{\eta_i} F ( \eta_i' ) = 0\) for every \(\eta_i'\) with \(d_{\eta_i} g_i ( \eta_i' ) = 0\) if and only if there exists a \(\mu_i \in \R\) with
    \begin{align} \label{eq:multiplier}
        d_{\eta_i} F = \frac{\mu_i}{2} \, d_{\eta_i} g_i \qquad \text{on all of } \Omega^1_D(\ST_i) \, .
    \end{align}
\end{corollary}

\begin{proof}
Take \(\nu\) from Lemma \ref{lem:nondegenerate} and apply the argument of Appendix \ref{app:LM} with \(\nu_0 = \nu\), setting \(\mu_i \coloneqq 2 \, d_{\eta_i} F ( \nu )\). The argument is purely algebraic, using only the splitting \(w = \big( w - d_{\eta_i} g_i(w) \, \nu \big) + d_{\eta_i} g_i(w) \, \nu\) of \(\Omega^1_D(\ST_i)\) as a vector space. The converse is immediate.
\end{proof}

Combining Lemma \ref{lem:quad}, Lemma \ref{lem:tangent} and Corollary \ref{cor:multiplier}, at a fixed \(f\) and for \(\beta_i \neq 0\) the first variation of the objective vanishes along every admissible path if and only if \eqref{eq:multiplier} holds for some \(\mu_i \in \R\). Both directions are proved, so the resulting conditions are necessary and sufficient.

The same argument applies verbatim to conventional helicity, up to any additional constraints on the potential.

Applying the method of Lagrange multipliers therefore yields the following necessary and sufficient conditions for a solution:
\begin{align*}
    \underset{\substack{\eta_i \in \Omega^1_D(\ST_i) \\ \mu_i \in \R}}{\text{c. p.}}
    \sum_{i=1}^n
    \left(
        \frac{1}{2}\Big\| d\eta_i + \sum_j \psi_{i,j} T_{i,j}^* \Big\|_{L^2}^2
        - \tilde p_i \myabs{f(\ST_i)}
        - \frac{\mu_i}{2}\big(\mathscr{H}(\beta_i,\ST_i)-h_i\big)
    \right),
\end{align*}
for Lagrange multipliers \(\mu_1,\dots,\mu_n \in \R\), assuming \(\beta_i \neq 0\). As each domain may be considered independently we obtain the following lemma.

\begin{lemma}\label{Lemma:replace}
    Let \((\ST_1,\dots,\ST_n)\) be a partition of \(\ST\), and assume that \(\beta_i \neq 0\) for every \(i\). Fix the constants \(\psi_{i,j}\), \(\tilde p_i\), and \(h_i\). Given \(f \in \I\) and \(\Gamma_i\), Problem \eqref{eq:ext} reduces, for each \(i\) with \(\ell_i>0\), to
    \begin{align}\label{eq:ext_lagrange}
        \underset{\substack{\eta_i \in \Omega^1_D(\ST_i) \\ \mu_i \in \R}}{\text{c. p.}}
        \left(
            \frac{1}{2}\Big\| d\eta_i + \sum_{j=1}^{\ell_i}\psi_{i,j} T_{i,j}^* \Big\|_{L^2}^2
            - \frac{\mu_i}{2}\big(\mathscr{H}(\beta_i,\ST_i)-h_i\big)
        \right).
    \end{align}
    For \(i\) where \(\ell_i=0\), the harmonic terms vanish.
\end{lemma}

\section{Main Proofs}\label{section:Main Proofs}
\subsection{Proof of Proposition \ref{thr:iff_simple}}\label{sec:main_theorem}

We start from the formulation for our variational problem given by \eqref{eq:ext} with \(\beta_i \neq 0\) for every \(i = 1, \dots, n\). Regions in which \(\beta_i = 0\) are treated in Theorem \ref{prop:vanishing} which is proven in Subsection \ref{sub:proof_prop}. Consider an \(f_t \in \I\) for each \(t \in (-\varepsilon, \varepsilon)\) that is differentiable in \(t\) where \(f_0 = f \). Recall that the objective function, or magnetic energy is:
\begin{align*}
    \begin{gathered} 
    E(f_t, \beta_{1,t}, \dots, \beta _{n,t})  =  \sum_{i=1}^n \Big( \ifrac{1}{2} \big\| \beta_{i,t}\big\|_{L^2}^2 -  \tilde p_i \myabs{ f_t(\ST_i) } \Big) \, , 
    \end{gathered}
\end{align*}
where \( \beta_{i,t} \in \C\) for each \(t \in (-\varepsilon,\varepsilon)\) and \(a_{i,t}\) is smooth in \(t\), where \(da_{i,t} = \beta_{i,t}\). Let \(a_{i,0} = a_i\) and \(\beta_{i,0}= \beta _i\). 


By the definition in Section \ref{section:Main Results}, \(P_0\) is a critical point precisely when the first variation of the objective vanishes at \(t=0\) along every admissible path \(P_t \coloneqq (f_t, a_{1,t}, \dots, a_{n,t})\), namely along every path satisfying the flux and loop-corrected helicity constraints in Problem \eqref{eq:ext}. We refer to such paths as \textit{admissible paths}. Hence, a necessary and sufficient condition for \(P_0\) to be a critical point is 
\begin{align*}
    0 = \partial_t E(P_t) |_{t=0}  =
     \left(  \partial_{f_t} E(f')  + \sum_{i=1}^n \partial_{ a_{i,t}} E( a_i') \right) \Bigg|_{t=0}  \, ,
\end{align*}
where \((f', a_{1}', \dots, a_{n}') = ( \partial_t f_{t}|_{t=0}, \partial_t a_{1,t}|_{t=0}      , \dots, \partial_t a_{n,t}|_{t=0} )
\) similar to Subsection \ref{sub:LM}. Namely, \(\partial_{f_t} E(f')\) is the partial derivative with respect to \(f_t\), in the direction \(f'\) at the point \(P_t\) for each \(t\). So the conditions for \(P_0\) to be a critical point are exactly:
\begin{align} \label{eq:crit_cond_1}
    0 = \partial_{f_t} E(f') |_{t=0} \, , \qquad 
    0 = \partial_{a_{i,t}} E( a_i') |_{t=0}  \, , \; \;  i = 1, \dots, n\, ,
\end{align}
for all admissible paths passing through \(P_0\). 

The conditions \eqref{eq:crit_cond_1} are separate conditions because the constraints couple neither \(f\) to the potentials nor the potentials to one another. Admissibility of \(\beta_{i,t} \in \mathcal C(\ST_i)\), the flux conditions, and the loop-corrected helicity involve no metric quantity. Namely, the conditions \(d\beta_i = 0\) and \(\mathcal{J}^*\beta_i = 0\) are defined by the exterior derivative and by pullback, the flux is the homological pairing of Subsection \ref{sub:flux}, and the integrands above require only the wedge product. None of them refers to the Hodge star or to \(f^*\varpi\). Moreover, the potentials are fields on the reference domain \(\ST_i\), which is held fixed, so their admissible set does not move with \(f\). The constraint set of Problem \eqref{eq:ext} is therefore a product
\begin{align*}
    \mathcal I (\ST) \times A_1 \times \cdots \times A_n \, , 
\end{align*}
where
\begin{align*}
    A_i \coloneqq \Big\{ a_i \in \Omega^1(\ST_i)  : \; d a_i \in \C \, , \; \;
    \mathscr{H}(d a_i , \ST_i) = h_i \, , \; \; \int_{T_{i,j}} d a_i = \psi_{i,j} \;  \forall j \Big\} \, ,
\end{align*}
and only the objective \(E\) couples the factors, through the metric appearing in the Hodge star and in \(\myabs{f(\ST_i)}\).

Consequently, an admissible path may be prescribed one factor at a time. Given any allowable \(P_t\), the paths obtained by freezing all but one factor at its value at \(t=0\), namely \((f_t, a_1, \dots, a_n)\) and \((f, a_1, \dots, a_{i,t}, \dots, a_n)\) for each \(i\), are again allowable, and their derivatives at \(t=0\) are \(f'\) and \(a_i'\) respectively. Evaluating \(\partial_t E(P_t)|_{t=0} = 0\) along these paths yields each equation of \eqref{eq:crit_cond_1}, conversely \eqref{eq:crit_cond_1} makes the sum vanish along every admissible path by linearity of the derivative in the direction. The two formulations are therefore equivalent, and the second condition of \eqref{eq:crit_cond_1} is precisely the problem of determining, for each \(i\) separately, the constrained critical points of \(E\) at fixed \(f \in \I\). The hypothesis \(\beta_i \neq 0\) is used only to supply the direction \(\nu\) of Lemma \ref{lem:nondegenerate}, which is what licenses the multiplier rule of Subsection \ref{sub:LM}.

We now focus on the condition
\[
0 = \partial_{a_{i,t}} E( a_i') |_{t=0}
\]
for a given \(i\) with fixed metric.

We know that the energy \(E\), loop-corrected helicity, and flux are independent of the gauge, hence their derivatives in directions that do not change \(\beta_i\) are automatically zero. As such, we can select a particular gauge. By construction, we have a basis for the harmonics \(\{ d \Gamma_{i,1}, \dots, d \Gamma_{i,\ell_i} \} = \{ T_{i,1}^*, \dots, T_{i,\ell_i}^* \}\). This basis is fixed by the geometry and metric which are fixed for the partial derivative with respect to \(a_{i,t}\). Therefore, knowing that there exists a \(\Gamma_{i,j}\), we select a gauge where \(\Gamma_{i,j}\) are fixed. Then \( a'_i = \eta'_i = \partial_t \eta_{i,t}|_{t=0}\) for this restricted set of variations, knowing that the corresponding derivative of \(E\) is automatically 0 for the excluded variations.

Recall from the previous section that this is equivalent to a Lagrange multipliers problem (Lemma \ref{Lemma:replace}). Namely, for each \(i = 1, \dots, n\), the condition \(0 = \partial_{\eta_{i,t}} E(a_i')|_{t=0}\) for every allowable \(\eta_i'\) holds if and only if there exists a \(\mu_i \in \R\) such that \((\eta_i, \mu_i)\) is a critical point of
\begin{align*}
    L(a_i, \mu_i) \coloneqq \frac{1}{2} \big\| d a_i \big\|_{L^2}^2
    - \frac{\mu_i}{2} \left( \mathscr{H}( d a_i, \ST_i) - h_i \right) -  \tilde p_i \myabs{ f(\ST_i) } \, , & &
    a_i = \eta_i + \sum_{j=1}^{\ell_i} \psi_{i,j} \Gamma_{i,j} \, ,
\end{align*}
where the harmonic terms are zero whenever \(\ell_i = 0\). Here \((\eta_i, \mu_i)\) ranges over the product \(\Omega^1_D(\ST_i) \times \R\), on which no constraint is imposed, so \(\eta_i'\) and \(\mu_i'\) may be prescribed independently. Hence
\begin{align*}
    0 = \partial_t L(a_{i,t}, \mu_{i,t}) |_{t=0}
    = \partial_{\eta_{i,t}} L(\eta_i') |_{t=0} + \partial_{\mu_{i,t}} L(\mu_i') |_{t=0}
    \quad \text{for all } (\eta_i', \mu_i')
\end{align*}
if and only if
\begin{align*}
    0= \partial_{\eta_{i,t}} L(\eta_i') |_{t=0} \quad \text{and} \quad
    0= \partial_{\mu_{i,t}} L(\mu_i') |_{t=0} \, .
\end{align*}
Here \(\mu_{i,t} \in \R\) is differentiable in \(t\) with \(\mu_i' = \partial_t \mu_{i,t}|_{t=0}\). The latter equation here is given by:
\begin{align*}
     0= \partial_{\mu_{i,t}} L(\mu_i') |_{t=0} = \partial_t L(a_{i}, \mu_{i,t}) |_{t=0}  =   - \frac{\mu_i'}{2} \left(   \mathscr{H}( \beta_i,  \ST_i) -  h_i      \right)  \, ,
\end{align*}
for all \(i\) and \(\mu_i' \in \R\). Hence, this is equivalent to \(\mathscr{H}( \beta_i,  \ST_i) =  h_i\) for all \(i\). For the other derivative of \(L\) we have that \(t \eta_i' = 0\), so the loop corrections disappear and we are left with
\begin{align*}
\partial_{\eta_{i,t}} L(\eta_i') |_{t=0} & = \partial_t L(a_{i,t}, \mu_{i}) |_{t=0}   = \langle d\eta_i',  \beta_i\rangle_{L^2} -  \frac{\mu_i}{2} \int_{\ST_i} (\eta'_i \wedge \beta_i + a_i \wedge d\eta_i')    
\\
& = \langle \eta_i' , \star (d \star \beta_i - \mu_i \beta_i) \rangle_{L^2} \, .
\end{align*}
The vanishing of \(\partial_{\eta_{i,t}} L(\eta_i') |_{t=0}\) for all \(\eta'_i \in  \Omega^1_D(\ST_i)\) implies that \(\star (d \star \beta_i - \mu_i \beta_i) = 0\) on the interior of \(\ST_i\). Since \(\beta_i\) is smooth, the quantity \(\star (d \star\beta_i - \mu_i\beta_i)\) is continuous. As this vanishes on the interior of \(\ST_i\), it must vanish everywhere on \(\ST_i\). Therefore, we conclude that \(d \star\beta_i = \mu_i\beta_i\), which is a Beltrami equation in \(\beta_i\) for each \(i\). Note that the Hodge star operator encodes the volume form \(f^* \varpi\) on \(\ST_i\). Upon pushing forward via \(f\), we obtain by Lemma \ref{lem:vector-dictionary}(6) the Beltrami equation \(\nabla \times B_i = \mu_i B_i\) for the elements \(B_i\) of \(C(f(\ST_i))\) with the Euclidean metric. 

The conditions for a critical point in \eqref{eq:crit_cond_1} are equivalently written as:
\begin{align} \label{eq:crit_cond_2}
	\begin{gathered}
		0 = \partial_{f_t} E(f_t) |_{t=0} = \partial_t \left(   \sum_{i=1}^n  \frac{1}{2} \big\| \beta_{i,t}\big\|_{L^2}^2 -  \tilde p_i \myabs{ f_t(\ST_i) }       \right) \Bigg|_{t=0} \, , \\
		d \star \beta_{i} = \mu_{i} \beta_{i} \, , \quad \mathscr{H}( \beta_{i,t},  \ST_i) =  h_i \, , 
	\end{gathered}
\end{align}
for all \(i\). Along an admissible path the periods \(\langle [\beta_{i,t}] ,[ T_{i,j} ] \rangle   = \psi_{i,j}\) are prescribed and constant in \(t\). Note that the loop-corrected helicity holds for all \(t\) by definition of the admissible paths. We interpret this problem as: finding critical points of the magnetic energy under variations of the boundary through all paths \(P_t\) satisfying the prescribed flux and helicity constraints while \(P_0\) satisfies the Beltrami equation. 

Consider a derivative \(d/dt f_t : \ST \to T f_t(\ST) \). Composing this with a directional derivative \(d_t f_{t}^{-1} :T f_t(\ST) \to T (\ST) \) gives us a vector field \(V(x) \coloneqq d_t f_{t}^{-1} \ (d/dt (f_t(x))  )|_{t=0}\), for \(x \in \ST\), namely \(V:\ST \to T(\ST)\). Equivalently \(V = f^* W\), where \(W \coloneqq (\partial_t f_t) \circ f_t^{-1}|_{t=0}\) is the variation field on the laboratory domain \(f(\ST)\). Hence, \(V\) is a field on the reference domain and every interior multiplication \(i_{V}\) below acts on forms on \(\ST_i\). With this we have the following formula, essentially Cartan's formula,
\begin{align*}
    \partial_t (f^*_t \omega_t) |_{t=0} = \mathcal{L}_{V} (f^* \omega_0) +  f^*( \partial_t \omega_t |_{t=0}) =  i_{V} (d f^* \omega_0) + d i_{V} (f^* \omega_0) + f^*(\partial_t \omega_t |_{t=0}) \, ,
\end{align*}
for any \(\omega_t \in  \Omega^k (f_t(\ST_i) ) \) where \(\mathcal{L}_{V}\) is a Lie derivative with respect to \(V\). 

It is useful to perform some operations on \(f_t(\ST_i)\) because the metric here is the Euclidean metric and independent of \(t\). Let \(\epsilon_i \coloneqq f_t^* \partial_t ((f_t^{-1})^* a_{i,t}) |_{t=0}\) where \( (f_t^{-1})^* a_{i,t} \in  \Omega^1 (f_t(\ST_i))\) and \( (f_t^{-1})^*\) is a pushforward by \(f_t\). Then we have commutation of the Hodge star with a derivative: 
\begin{align*}
    f^*_t \partial_t (f_t^{-1})^* \star a_{i,t} |_{t=0} = \star f^*_t  \partial_t (f_t^{-1})^* a_{i,t} |_{t=0} = \star \epsilon_i \, ,
\end{align*}
where \(\star\) is the Hodge star in \(\ST_i\) with volume form \(f_t^* \varpi\). We use the notation \(a_i' = \partial_t a_{i,t}|_{t=0}\), the derivative of the full primitive. This gives the relations:
\begin{align*}
    a_i' &= \partial_t a_{i,t}|_{t=0} = \partial_t  f_t^*  ((f_t^{-1})^* a_{i,t}) |_{t=0} = 
    i_{V} ( \beta_{i}  ) + d i_{V} (  a_{i}  ) + \epsilon_i \, , \\
    (\star a_i)' &= \partial_t \star a_{i,t}|_{t=0} = \partial_t   f_t^*  (  (f_t^{-1})^*    \star a_{i,t} ) |_{t=0} = i_{V} ( d  \star a_{i}  ) + d i_{V} (  \star a_{i}  ) + \star \epsilon_i \, .
\end{align*}
Let's return to our problem, given by equation \eqref{eq:crit_cond_2}. For simplicity, we deal with each term in the derivative with respect to \(f_t\) independently until the end of our computation. First, we consider \(\| \beta_{i,t}\|_{L^2}^2\) for any \(i \in \{1, 2, \dots, n\} \):
\begin{align*}
    \partial_t    \big\| \beta_{i,t}\big\|_{L^2}^2    \Big|_{t=0}  & = 2  \langle da_i, d \epsilon_i \rangle + \! \int_{\ST_i} \! \! d i_{V} (da_i \wedge \star da_i) 
    \\ &= 2 \mu_i \! \int_{\ST_i} \! \! \epsilon_i \wedge\beta_i + \! \int_{\partial \ST_i}  \! \! \mathcal{J}^* ( 2 \epsilon_i \wedge \star\beta_i +  i_{V} (da_i \wedge \star da_i ) ) \, .
\end{align*}
We have used the Beltrami equation and Stokes' theorem in this computation. Note that these computations hold for any \(\mu_i \). Additionally, the loop-corrected helicity is constant in \(t\), so:
\begin{align*}
    0 & = \partial_t  \Bigg( \int_{\ST_i} a_{i,t} \wedge \beta_{i,t}  - \sum_{j=1}^{\ell_i} \int_{\sigma_{i,j}} \mathcal{J}^* a_{i,t} \int_{\tau_{i,j}} \mathcal{J}^* a_{i,t}  \Bigg) \Bigg|_{t=0}\\
    & =  \int_{\ST_i} \epsilon_i \wedge\beta_i + a_i \wedge d \epsilon_i + \mathcal{L}_{V}(a_i \wedge\beta_i) - \sum_{j=1}^{\ell_i} \psi_{i,j}  \int_{\sigma_{i,j}} \mathcal{J}^* a_i'    \\ 
     2 \int_{\ST_i}   \epsilon_i \wedge\beta_i & =  \int_{\partial \ST_i} \mathcal{J}^* (  a_i \wedge \epsilon_i -   i_{V}(a_i \wedge\beta_i) ) + \sum_{j=1}^{\ell_i} \psi_{i,j}  \int_{\sigma_{i,j}} \mathcal{J}^* a_i' \, , 
\end{align*}
where we have used Stokes' theorem. Using this equality in the derivative of \(\| \beta_{i,t}\|_{L^2}^2\):
\begin{align*}
    \partial_t    \big\| \beta_{i,t}\big\|_{L^2}^2    \Big|_{t=0}  \!  =  & \int_{\partial \ST_i}  \! \! \mathcal{J}^* ( \mu_i a_i \wedge \epsilon_i -   \mu_i i_{V}(a_i \wedge\beta_i)   +    2 \epsilon_i \wedge \star\beta_i +  i_{V} (da_i \wedge \star da_i) ) \\
    & + \mu_i \sum_{j=1}^{\ell_i} \psi_{i,j}  \int_{\sigma_{i,j}} \! \! \mathcal{J}^* a_i' \, .
\end{align*}
We see that this variation with respect to \(t\) reduces to integrals along the boundary \(\partial \ST_i\). With some rearrangement and noting that \(t \beta_i = 0\) so terms of the form \(\beta_i \wedge \omega \) are zero under \(\mathcal{J}^*\), we arrive at:
\begin{align*}
    \partial_t    \big\| \beta_{i,t}\big\|_{L^2}^2    \Big|_{t=0}   = &  \int_{\partial \ST_i} \mathcal{J}^* (   (i_{V}\beta_i + \epsilon_i ) \wedge ( 2 \star\beta_i - \mu_i a_i)  - i_{V}( \beta_i \wedge \star \beta_i) )  \\ & + \mu_i \sum_{j=1}^{\ell_i} \psi_{i,j}  \int_{\sigma_{i,j}} \mathcal{J}^* a_i' \, .
\end{align*}
This normal condition on \(\beta_i\) also implies the following:
\begin{align*}
    0 = \mathcal{J}^* \beta_{i,t} \implies 0 = \mathcal{J}^* (d \epsilon_i  + \mathcal{L}_{V}(\beta_i) ) = d \mathcal{J}^* ( \epsilon_i  + i_{V}(\beta_i) )  \, . 
\end{align*}
Additionally, for the first term we have that both terms in the wedge product are closed in \( \Omega^1 (\partial \ST_i)\). Applying equation \eqref{eq:compute_2}, gives
\begin{align*}
    \partial_t    \big\| \beta_{i,t}\big\|_{L^2}^2  &= - \sum_{j=1}^{\ell_i}  \left(   \int_{\sigma_{i,j}} \! \mathcal{J}^* (i_{V}\beta_i + \epsilon_i )    \int_{\tau_{i,j}} \!  \mathcal{J}^* ( 2 \star\beta_i - \mu_i a_i) \right. \\ & \qquad \left. -\int_{\sigma_{i,j}} \!  \mathcal{J}^* ( 2 \star\beta_i - \mu_i a_i) \int_{\tau_{i,j}} \!  \mathcal{J}^* (i_{V}\beta_i + \epsilon_i ) \right) \\ & \qquad -
    \int_{\partial \ST_i}  \mathcal{J}^* i_{V}( \beta_i \wedge \star\beta_i)  + \mu_i \sum_{j=1}^{\ell_i} \psi_{i,j} \int_{\sigma_{i,j}}  \mathcal{J}^*  a_i'  \, .
\end{align*}
Using Stokes' theorem 
\begin{align*}
    \int_{\tau_{i,j}} \mathcal{J}^* (i_{V} \beta_i + \epsilon_i ) = \int_{\tau_{i,j}} \mathcal{J}^* a_i' =  \int_{T_{i,j}} \beta_i' = \partial_t \psi_{i,j} |_{t=0} = 0 \, .
\end{align*}
So this term vanishes as flux is held constant. And via Stokes' theorem and the Beltrami equation
\begin{align*}
    \sum_{j=1}^{\ell_i}    \int_{\sigma_{i,j}} \mathcal{J}^*(i_{V}\beta_i + \epsilon_i )    \int_{\tau_{i,j}} \mathcal{J}^*( 2 \star\beta_i - \mu_i a_i)  = 
    \mu_i \sum_{j=1}^{\ell_i}   \psi_{i,j}   \int_{\sigma_{i,j}} \mathcal{J}^* a_i'   \, .
\end{align*}
So we are left with:
\begin{align*}
     \partial_t    \big\| \beta_{i,t}\big\|_{L^2}^2    \Big|_{t=0} & =   -
    \int_{\partial \ST_i} \mathcal{J}^* (i_{V}( \beta_i \wedge \star\beta_i) )  \, .
\end{align*}
This argument uses \(\ell_i > 0\), and if \(\ell_i=0\) the loop terms automatically vanish and we obtain the same result. Let's also simplify the \( \myabs{ f_t(\ST_i) } \) terms in \(\partial_{f_t} E(f') |_{t=0}\):
\begin{align*}
    \partial_t \myabs{ f_t(\ST_i) } |_{t=0} = \int_{\ST_i} \partial_t f^*_t \varpi |_{t=0} 
    = \int_{\ST_i} \mathcal{L}_{V} ( f^* \varpi ) =  \int_{\ST_i} d i_{V} ( f^* \varpi ) =  \int_{\partial \ST_i} \mathcal{J}^* ( i_{V} f^* \varpi ) \, . 
\end{align*}
Therefore the derivative \(\partial_{f_t} E(f') |_{t=0}\) is:
\begin{align*}
    0 = \partial_{f_t} E(f') |_{t=0} &= \partial_t \left(   \sum_{i=1}^n  \frac{1}{2} \big\| \beta_{i,t}\big\|_{L^2}^2 -  \tilde p_i \myabs{ f_t(\ST_i) }       \right) \Bigg|_{t=0} 
    \\ & =  -  \sum_{i=1}^n   
    \int_{\partial \ST_i}  \mathcal{J}^* (i_{V} f^* \varpi) \Big( \ifrac{1}{2}  \| \beta_i\|^2 \,  +  \tilde p_i  \Big) \, .
\end{align*}
This is required for all \(V\) with \(V|_{\partial \ST}=0\). Every
smooth $V$ vanishing on $\partial\Lambda$ is the variation field of an admissible path. We see that only the normal component in the image of \(V\) is relevant in this integral. As $\mathcal{J}^*(i_{V}f^*\varpi) = (V\!\cdot\!N)\,dA$, each interface carries opposite normals from its two sides. The component $V\!\cdot\!N$ is arbitrary on each interface and vanishes on $\partial\Lambda$, hence the jump vanishes pointwise,
\begin{align*}
    0 = \Big\llbracket \ifrac12 \| \beta \|^2 + \tilde p \Big\rrbracket \Big|_{I_q} \, , 
\end{align*}
for all interfaces \(I_q\), or equivalently by Lemma \ref{lem:vector-dictionary}(2), \(\llbracket ( \sfrac12 ) \| B \|^2  + p \rrbracket |_{f(I_q)} = 0\) on every interface \(f(I_q)\). So the necessary and sufficient conditions for a critical point with \(\beta_i \neq 0\) for every \(i\) are equivalent to 
\begin{align*} 
    0 = \Big\llbracket \ifrac12 \|\beta \|^2 +  \tilde p \Big\rrbracket \Big|_{I_q} \, , \qquad 
    d \star\beta_i = \mu_i\beta_i \, , \qquad \mathscr{H}( \beta_i,  \ST_i) =  h_i \, ,
\end{align*}
where the flux of \(\beta_i\) through any element of \(H_2(\ST_i, \partial \ST_i)\) is specified. This concludes the proof of Proposition \ref{thr:iff_simple}. \hfill $\square$




\subsection{Proof of Theorem \ref{prop:vanishing}} \label{sub:proof_prop}

First we show that, for a fixed \(i\)
\begin{align*}
    h_i = 0 \, , \psi_{i,j} = 0 \, , d \star \beta_i = \mu_i \beta_i \iff \beta_i = 0 \, .
\end{align*}
Namely these have the same set of solutions for \(\beta_i\). This effectively lets us skip the Lagrange multiplier method in the \(\beta_i=0\) case and start directly from equation \eqref{eq:crit_cond_2}. And the interface variations have no requirement that \(\beta_i \neq 0\), hence, the remainder of the proof holds.

    If \(\beta_i = 0\) then \(\langle [\beta_i],[T_{i,j}]\rangle = 0\) for every \(j\), and \(\mathscr{H}(\beta_i,\ST_i) = 0\) directly from the definition of the loop-corrected helicity. Hence, \(\beta_i = 0\) implies that the flux coefficients and helicity constraint must be 0. The Beltrami equation is trivially satisfied.

    Consider the case where all flux coefficients vanish, \(h_i=0\) and \(d \star \beta_i = \mu_i \beta_i\). As the flux is zero the loop corrections in the loop-corrected helicity are zero. The harmonic component of \(\beta_i\) in \eqref{eq:b_form} is determined by the fluxes and is therefore zero, so we may take \(a_i = \eta_i \in \Omega^1_D(\ST_i)\). As \(t\eta_i = 0\), Stokes' theorem and the Beltrami equation give
    \begin{align*}
        \|\beta_i\|_{L^2}^2 = \int_{\ST_i} d\eta_i \wedge \star\beta_i
        = \int_{\partial \ST_i} \! \mathcal{J}^*(\eta_i \wedge \star \beta_i) + \int_{\ST_i}  \! \eta_i \wedge d\star\beta_i
        = \mu_i \int_{\ST_i}  \! \eta_i \wedge \beta_i = \mu_i h_i = 0 \, ,
    \end{align*}
    hence \(\beta_i = 0\). This holds for any \(\mu_i \) and any other primitive \(a_i\).


At $\beta_i=0$ the constraint is degenerate ($d_{\eta_i}g_i = 0$), so the multiplier rule is unavailable. However, criticality is immediate since along any differentiable admissible path $\partial_t (\sfrac12) \|\beta_{i,t}\|^2_{L^2}\big|_{t=0} = \langle \beta_i,\beta_i'\rangle_{L^2} = 0$. Hence, $\beta_i = 0$ is a critical point of the variation in $a_i$ and it satisfies $d\star\beta_i=\mu_i\beta_i$ for every $\mu_i\in\mathbb{R}$.

Conversely, for data $h_i=0,\psi_{i,j}=0$, any
critical point with $\beta_i\neq0$ would satisfy Beltrami and hence, by the equivalence just proved, be zero. Hence, the MRxMHD equations, along with the corresponding constraints, provide necessary and sufficient conditions for a critical point. The interface variations hold verbatim.

As minimisers are a subset of the set of corresponding critical points, the MRxMHD equations are necessary conditions for a corresponding constrained energy minimiser. \hfill $\square$

\subsection{Proof of Theorem \ref{thr:one_gauge}} \label{subsec:proof_one_gauge}

Theorem \ref{thr:one_gauge} selects a primitive of each field modulo exact \(1\)-forms. In particular, Lemma \ref{Lemma:gauge_exist} shows that the Amperian gauge has precisely this uniqueness property; it does not select an individual primitive. Throughout, \(a_i\) denotes a representative in \(G_i\). This choice does not affect the argument: for any smooth function \(\chi\), Stokes' theorem and \(\mathcal{J}^*\beta_i=0\) give
\[
    \int_{\ST_i}d\chi\wedge\beta_i=0,
\]
while the periods of \(d\chi\) around closed loops vanish. Hence both \(H_{G_i}\) and \(\kappa_i\) are unchanged by replacing \(a_i\) with \(a_i+d\chi\). The content of the proof is that condition \eqref{eq:gauge_condition} reconciles this formulation with Proposition \ref{thr:iff_simple}.

Consider a path \(P_t\) as in Theorem \ref{thr:one_gauge} and fix the metric via \(f_t = f\) for all \(t \in (-\varepsilon, \varepsilon)\), and set \(a_i' \coloneqq \partial_t a_{i,t}|_{t=0}\), which exists by construction of the paths. The exterior derivative is linear on \(\Omega^1(\ST_i)\), so \(\beta_i' \coloneqq \partial_t \beta_{i,t}|_{t=0} = da_i'\), in particular \(\beta_i'\) is closed. Additionally, \(\beta_i'\) is Dirichlet because \(\mathcal{J}^*\beta_{i,t} = 0\) for every \(t\), and it is flux-free because
\[
    \int_{T_{i,j}} \beta_i'
    = \frac{d}{dt}\bigg|_{t=0} \int_{T_{i,j}} \beta_{i,t}
    = \frac{d}{dt}\bigg|_{t=0} \psi_{i,j} = 0 \, .
\]
By Subsections \ref{sub:Hodge} and \ref{sub:flux} the harmonic component of \(\beta_i'\) therefore vanishes, and \(\beta_i' = d\eta_i'\) for some \(\eta_i' \in \Omega^1_D(\ST_i)\). We use \(\eta_i'\) only as a primitive of the field variation, no claim is made that \(\eta_i'\), or \(a_i + t\eta_i'\), lies in \(G_i\). The two are related by \(s \coloneqq a_i' - \eta_i'\), a closed \(1\)-form on \(\ST_i\). By the hypothesis of Theorem \ref{thr:one_gauge}, every \(\eta_i' \in \Omega^1_D(\ST_i)\) arises in this way from some such path.

\begin{lemma} \label{lem:period_identity}
    Let \(\beta_i \in \mathcal{C}(\ST_i)\) with primitive \(a_i \in \Omega^1(\ST_i)\), let \(a_i' \in \Omega^1(\ST_i)\) and let \(\eta_i' \in \Omega^1_D(\ST_i)\) be any Dirichlet primitive of \(\beta_i' = d a_i'\). Then
    \begin{align}
        \int_{\partial \ST_i} \mathcal{J}^*(a_i \wedge a_i') &= \kappa_i(a_i') \, , \label{eq:period_bdry} \\
        \int_{\ST_i} a_i' \wedge \beta_i &= \int_{\ST_i} \eta_i' \wedge \beta_i + \kappa_i(a_i') \, . \label{eq:period_bulk}
    \end{align}
\end{lemma}

\begin{proof}
Both \(\mathcal{J}^*a_i\) and \(\mathcal{J}^*a_i'\) are closed on \(\partial \ST_i\), since \(d\mathcal{J}^*a_i = \mathcal{J}^*\beta_i = 0\) and likewise for \(a_i'\). As \(\tau_{i,j} = \partial T_{i,j}\) by Theorem \ref{thr:Alexander}, Stokes' theorem gives
\begin{align} \label{eq:tau_periods}
    \int_{\tau_{i,j}} \mathcal{J}^*a_i = \int_{T_{i,j}} \beta_i = \psi_{i,j} \, , \qquad
    \int_{\tau_{i,j}} \mathcal{J}^*a_i' = \int_{T_{i,j}} \beta_i' = 0 \, .
\end{align}
Lemma \ref{Lemma:compute} applied to the pair \((a_i, a_i')\), together with \eqref{eq:tau_periods}, gives
\[
    \int_{\partial \ST_i} \mathcal{J}^*(a_i \wedge a_i')
    = \sum_{j=1}^{\ell_i} \left( \int_{\sigma_{i,j}} a_i' \int_{\tau_{i,j}} a_i  - \int_{\sigma_{i,j}} a_i \int_{\tau_{i,j}} a_i'  \right)
    = \sum_{j=1}^{\ell_i} \psi_{i,j} \int_{\sigma_{i,j}} \mathcal{J}^*a_i'
    = \kappa_i(a_i') \, ,
\]
which is \eqref{eq:period_bdry}. For \eqref{eq:period_bulk}, the form \(s = a_i' - \eta_i'\) is closed, so \(d(s \wedge a_i) = -\, s \wedge \beta_i\), and by Stokes' theorem and Lemma \ref{Lemma:compute},
\[
    \int_{\ST_i} s \wedge \beta_i
    = -\int_{\partial \ST_i} \mathcal{J}^*(s \wedge a_i)
    = -\sum_{j=1}^{\ell_i} \left(     \int_{\sigma_{i,j}} a_i \int_{\tau_{i,j}} s     - \int_{\sigma_{i,j}} s \int_{\tau_{i,j}} a_i  \right)
    = \kappa_i(s) \, .
\]
Here we have used that the integral of \(s\) over each \(\tau_{i,j}\) is zero by Stokes' theorem. As \(\mathcal{J}^*\eta_i' = 0\) gives \(\kappa_i(\eta_i') = 0\), so \(\kappa_i(s) = \kappa_i(a_i')\).
\end{proof}

Since \(\kappa_i\) is linear and continuous, \eqref{eq:gauge_condition} is the condition that \(\kappa_i(a_i') = 0\) for every admissible variation. As this is currently assumed, both correction terms in Lemma \ref{lem:period_identity} vanish.

As the primitives of each \(\beta_i \in \mathcal{C}(\ST_i)\) in \(G_i\) differ only by exact \(1\)-forms, the objective and the constraint of the present problem are functions of \((f, \beta_1, \dots, \beta_n)\) alone. Fixing \(f\) and writing 
\[
    \beta_i = d\eta_i + \sum_j \psi_{i,j} T^*_{i,j}
\]
 the problem is therefore posed over \(\eta_i \in \Omega^1_D(\ST_i)\), as in Subsection \ref{sub:LM}. Moreover \eqref{eq:gauge_condition} holds along every path meeting the flux constraints, not merely those meeting the helicity constraint, so the map \(\beta_i \mapsto \kappa_i(a_i)\) has vanishing derivative at every such configuration and is constant, say equal to \(c_i\), on each path-connected component of that set. By Subsection \ref{subsection:Relative Helicity},
\begin{align} \label{eq:helicity_shift}
    H_{G_i}(\beta_i, \ST_i) = \mathscr{H}(\beta_i, \ST_i) + \kappa_i(a_i) = \mathscr{H}(\beta_i, \ST_i) + c_i \, ,
\end{align}
so on each such component \(g_i\) differs from the loop-corrected helicity by the constant \(c_i\), and Lemmas \ref{lem:quad}, \ref{lem:nondegenerate} and \ref{lem:tangent} of Subsection \ref{sub:LM} apply to it verbatim.

Write \(g_i(a_i) \coloneqq H_{G_i}(\beta_i, \ST_i) \). From \(d(a_i \wedge a_i') = a_i' \wedge \beta_i - a_i \wedge \beta_i'\) and Stokes' theorem,
\[
    \partial_t g_i(a_{i,t})\big|_{t=0}
    = \int_{\ST_i} \big( a_i' \wedge \beta_i + a_i \wedge \beta_i' \big)
    = 2\int_{\ST_i} a_i' \wedge \beta_i
      - \int_{\partial \ST_i} \mathcal{J}^*(a_i \wedge a_i') \, ,
\]
and by Lemma \ref{lem:period_identity} together with \eqref{eq:gauge_condition},
\begin{align} \label{eq:constraint_derivative_gauge}
    \partial_t g_i(a_{i,t})\big|_{t=0}
    = 2\int_{\ST_i} \eta_i' \wedge \beta_i + \kappa_i(a_i')
    = 2\int_{\ST_i} \eta_i' \wedge \beta_i \, .
\end{align}
The boundary term \eqref{eq:period_bdry} is the contribution that \eqref{eq:gauge_condition} removes at fixed geometry. Comparing \eqref{eq:constraint_derivative_gauge} with \eqref{eq:deriv}, the first variation of \(g_i\) along a path of the present problem is \(d_{\eta_i} g_i(\eta_i')\), so the admissible variations are contained in \(\ker(d_{\eta_i} g_i)\). Whenever \(\beta_i \neq 0\), Lemma \ref{lem:nondegenerate} supplies a \(\nu \in \Omega^1_D(\ST_i)\) with \(d_{\eta_i} g_i(\nu) = 1\), and Lemma \ref{lem:tangent} produces, for each \(\eta_i' \in \ker(d_{\eta_i} g_i)\), a family in \(\Omega^1_D(\ST_i)\) along which the loop-corrected helicity is constant and whose derivative at \(t=0\) is \(\eta_i'\). By the family-level exhaustiveness hypothesis of Theorem \ref{thr:one_gauge} that family is realised by a path \(a_{i,t}\) of primitives in \(G_i\), which by \eqref{eq:helicity_shift} preserves \(H_{G_i}(\beta_{i,t}, \ST_i)\). The admissible variations are therefore exactly \(\ker(d_{\eta_i} g_i)\), and Corollary \ref{cor:multiplier} gives necessary and sufficient conditions for a critical point at fixed \(f\). So there exists \(\mu_i \in \mathbb{R}\) such that \((a_i, \mu_i)\) is a critical point of
\[
    L(a_i, \mu_i)
    = \ifrac{1}{2} \big\| da_i \big\|^2_{L^2}
       - \tilde{p}_i \myabs{f(\ST_i)}
       - \ifrac{\mu_i}{2} \big( H_{G_i}(da_i, \ST_i) - h_i \big) \, ,
\]
the pressure term being independent of \(a_i\).

Still at fixed \(f\), and using \(\mathcal{J}^*\eta_i' = 0\) with \(d(\eta_i' \wedge \star\beta_i)  = d\eta_i' \wedge \star\beta_i - \eta_i' \wedge d\star\beta_i\),
\[
    \partial_t \tfrac{1}{2}\big\| \beta_{i,t} \big\|^2_{L^2}\Big|_{t=0}
    = \big\langle \beta_i' , \beta_i \big\rangle_{L^2}
    = \int_{\ST_i} d\eta_i' \wedge \star\beta_i
    = \int_{\ST_i} \eta_i' \wedge d\star\beta_i \, .
\]
Combining this with \eqref{eq:constraint_derivative_gauge} and using \(\star\star = 1\) on \(\Omega^2(\ST_i)\),
\[
    \partial_t L(a_{i,t}, \mu_i)\big|_{t=0}
    = \int_{\ST_i} \eta_i' \wedge \big( d\star\beta_i - \mu_i \beta_i \big)
    = \Big\langle \eta_i' , \star\big( d\star\beta_i - \mu_i \beta_i \big)
      \Big\rangle_{L^2} \, .
\]
Vanishing for every \(\eta_i' \in \Omega^1_D(\ST_i)\) forces \(\star(d\star\beta_i - \mu_i \beta_i) = 0\) on \(\mathrm{int}(\ST_i)\). This quantity is continuous because \(\beta_i\) is smooth, so it vanishes on all of \(\ST_i\). Therefore \(d\star\beta_i = \mu_i\beta_i\), which by Lemma \ref{lem:vector-dictionary}(6) is the Beltrami equation \(\nabla \times B_i = \mu_i B_i\) on \(f(\ST_i)\).

The \(\beta_i = 0 \) case holds via the same argument used in Subsection \ref{sub:proof_prop}.

Let \(t \mapsto (f_t, a_{1,t}, \dots, a_{n,t})\) be a path as in Theorem \ref{thr:one_gauge} that also satisfies \(H_{G_i}(\beta_{i,t}, \ST_i) = h_i\) for every \(t\). By \eqref{eq:helicity_shift} we have \(\mathscr{H}(\beta_{i,t}, \ST_i) = h_i - c_i\) for every \(t\), so the path preserves the loop-corrected helicity and is an admissible path in the sense of Subsection \ref{sec:main_theorem}. The computation given there applies verbatim: the \(\kappa_i\) term arising in \(\partial_t \|\beta_{i,t}\|^2_{L^2}\) is cancelled, there by the derivative of the loop-corrected helicity constraint and here by \eqref{eq:gauge_condition}, leaving
\[
    0 = \partial_{f_t} E(f')\big|_{t=0}
      = -\sum_{i=1}^{n} \int_{\partial \ST_i} \mathcal{J}^*\big( i_{V} f^*\varpi \big) \Big( \ifrac{1}{2}\|\beta_i\|^2 + \tilde{p}_i \Big) \, ,
\]
and hence \(\llbracket ( \sfrac{1}{2} ) \|\beta\|^2 + \tilde{p} \rrbracket|_{I_q} = 0\) across every interface \(I_q\). Together with the Beltrami equation this proves the conditions are necessary and sufficient under the corresponding assumptions. \hfill $\square$

\subsection{Consistency with SPEC} \label{subsec:spec}

The Stepped Pressure Equilibrium Code (SPEC) determines extrema of the magnetic energy subject to prescribed enclosed fluxes, prescribed helicity, and an adiabatic constraint \(p_i \myabs{f(\ST_i)}^{\gamma} = v_i\) for each volume. We verify that the gauge used by SPEC satisfies the hypotheses of Theorem \ref{thr:one_gauge}, working in vector notation and following \citet{10.1063/1.4765691} as closely as is convenient. We work with an MRxMHD Problem \eqref{prob:MRxMHD_adiabatic}, which we note is not the default search space for SPEC but this will not affect us checking the assumptions required to apply Theorem \ref{thr:one_gauge}. The primary hypothesis of Theorem \ref{thr:one_gauge}, namely equation \eqref{eq:gauge_condition}, is transcribed into vector notation in Appendix \ref{append:vector_notation}.

The verification requires no knowledge of the Beltrami field beyond the constraints imposed upon it. By remark \ref{remark:helicity} each circulation of \(A_i\) about \(f(\sigma_{i,j})\) may be assigned any value by a choice of primitive, so \(\kappa_i\) is not a functional of \(B_i\) but is fixed by the gauge alone. It is enough to decide whether the gauge conditions imposed by the code depend on the interface geometry. We first restrict ourselves to hollow tori by taking \(i \ge 2\). The \(i=1\) case is considered separately.

SPEC partitions a toroidal domain by nested toroidal interfaces \(f(I_q)\), the region \(f(\ST_i)\) being bounded by \(f(I_{i-1})\) and \(f(I_i)\) for \(i \geq 2\). Coordinates \((r,\theta,\zeta)\) adapted to the interfaces are defined inversely by
\[
    R = R(r,\theta,\zeta) \, , \qquad \phi = -\zeta \, , \qquad Z = Z(r,\theta,\zeta) \, ,
\]
where \(R\) and \(Z\) are interpolated radially between the Fourier harmonics \(R_{i,j}\), \(Z_{i,j}\) (for \(i \ge 2\)) that describe the interfaces \cite{10.1063/1.4765691}. Note that Hudson et al. use \(s\) instead of \(r\) for the corresponding radial coordinate. In the language of this paper the coordinate domain \([r_{i-1}, r_i] \times \mathbb{T}^2\) is the reference region \(\ST_i\), as in Appendix \ref{append:example}, and the coordinate map is \(f\). Varying \(\{R_{i,j}, Z_{i,j}\}\) is a variation of \(f\) with the reference domain held fixed. Here \(\ell_i = 2\) for \(i \ge 2\), and by Theorem \ref{thr:Alexander} and Appendix \ref{append:example} we may take
\[
    \sigma_{i,1} : t \mapsto (r_i,\, 0,\, t) \, ,
     \qquad
    \sigma_{i,2} : t \mapsto (r_{i-1},\, t,\, 0) \, , \qquad t \in [0, 1] \, ,
\]
a toroidal loop on the outer interface and a poloidal loop on the inner one. The dual curves are \(\tau_{i,1}\), bounding a poloidal cross-section \(T_{i,1}\) carrying the toroidal flux \(\psi_{i,1} \), and \(\tau_{i,2}\), bounding the toroidal ribbon \(T_{i,2}\) carrying the poloidal flux \(\psi_{i,2} \).

Beginning from a general covariant representation \(\bar{A}_i = \bar{A}_{i, r}\nabla r + \bar{A}_{i, \theta}\nabla\theta + \bar{A}_{i, \zeta}\nabla\zeta\), SPEC adds a gauge term \(\nabla g_i\) chosen so that
\begin{align*}
    \partial_r g_i(r,\theta,\zeta) &= -\bar{A}_{i, r}(r,\theta,\zeta) \, , \\
    \partial_\theta g_i(r_{i-1},\theta,\zeta)
        &= -\bar{A}_{i, \theta}(r_{i-1},\theta,\zeta) + \psi_{i-1, t} \, , \\
    \partial_\zeta g_i(r_{i-1},0,\zeta)
        &= -\bar{A}_{i, \zeta}(r_{i-1},0,\zeta) + \psi_{i-1, p} \, ,
\end{align*}
for prescribed constants determined, with the orientations of Appendix \ref{append:example}, by
\[
    \psi_{i,t}=\sum_{k=1}^{i}\psi_{k,1}, \qquad
    \psi_{i,p}=-\sum_{k=2}^{i}\psi_{k,2}, \qquad
    \psi_{0,t}=\psi_{0,p}=\psi_{1,p}=0.
\]
The minus sign follows from \(\psi_{i,2}=\psi_{i-1,p}-\psi_{i,p}\). Thus \(A_i = A_{i, \theta}\nabla\theta + A_{i, \zeta}\nabla\zeta\) with
\begin{equation}
    A_{i, \theta}(r_{i-1},\theta,\zeta) = \psi_{i-1, t} \, , \qquad
    A_{i, \zeta}(r_{i-1},0,\zeta) = \psi_{i-1, p} \, .
    \label{eq:spec_inner}
\end{equation}
These conditions determine \(g_i\) up to an additive constant and hence determine \(A_i\) uniquely, so \(G_i\) contains exactly one primitive of each \(\beta_i\). They are linear in the potential and solved by successive quadratures in \(r\), \(\theta\), and \(\zeta\), so the induced selection \(\beta_i \mapsto a_i\) is affine and continuous. Given \(\eta_i' \in \Omega^1_D(\ST_i)\), the family \(\beta_{i,t} = \beta_i + t \, d\eta_i'\) lies in \(\mathcal{C}(\ST_i)\) with unchanged fluxes, so \(t \mapsto a_{i,t} \) is affine in \(t\), and its derivative lies in \(W_i\) with exterior derivative \(d\eta_i'\). The successive quadratures also preserve smooth dependence on \((t,x)\) for any smooth family \(\eta_{i,t}\), giving the family-level statement. The variations are therefore exhaustive in both senses required by Theorem \ref{thr:one_gauge}, and since the conditions above are imposed in the fixed reference coordinates \((r,\theta,\zeta)\) the same holds along deformations of the interfaces.

The first condition alone removes the radial component of the potential while leaving a residual freedom \(g_i = g_i(\theta,\zeta)\), which admits multivalued gauge functions. The second and third fix this freedom on the inner interface, giving \eqref{eq:spec_inner}. Together with the tangency condition \(B_i\cdot\nabla r = 0\) on the interfaces, equivalently \(t\beta_i = 0\), and the prescribed enclosed fluxes, these conditions give at the outer interface
\begin{equation}
    A_{i, \theta}(r_i,\theta,\zeta) = \psi_{i, t} + \partial_\theta \varphi_i \, , \qquad
    A_{i, \zeta}(r_i,\theta,\zeta) = \psi_{i, p} + \partial_\zeta \varphi_i \, , \label{eq:spec_outer}
\end{equation}
where \(\varphi_i\) is single-valued on \(I_i\). The tangency condition also implies that \(\mathcal{J}^*a_i\) is closed on each boundary torus, so these circulations depend only on the classes \([\sigma_{i,j}] \in H_1(\partial \ST_i)\) and not on the representative curves. In particular, the poloidal angle re-parametrisations generated by SPEC's spectral condensation do not alter them. Writing \(A_{i,t} = \big( (f_t^{-1})^* a_{i,t} \big)^{\sharp}\), equations \eqref{eq:spec_inner} and \eqref{eq:spec_outer} give
\[
    \oint_{f_t(\sigma_{i,1})} A_{i,t} \cdot d\ell =  \psi_{i, p} \, ,
    \qquad \oint_{f_t(\sigma_{i,2})} A_{i,t} \cdot d\ell =  \psi_{i-1, t} \, ,
\]
with the orientations specified above. The enclosed fluxes \(\psi_{i, j}\) are prescribed input data and carry no dependence on \(\{R_{i,j}, Z_{i,j}\}\), so
\begin{equation} \label{eq:spec_kappa}
    \kappa_i(a_{i,t})  =  \psi_{i,1}\, \psi_{i, p} + \psi_{i,2}\, \psi_{i-1, t} 
\end{equation}
is constant along any variation of the interfaces, and \eqref{eq:gauge_condition} holds. We use \(t \in [0, 1]\) for internal consistency. It should be noted that SPEC uses \([0,2 \pi]\) resulting in slightly different normalisations, but the result is unaffected.

The innermost volume is a solid torus, for which \(\ell_1 = 1\). The surface \(r = r_0\) degenerates to the coordinate axis, where the second and third conditions on \(g_i\) cannot be imposed as written. In their place SPEC requires the potential to be analytic at the origin, which introduces radial factors that make every harmonic of \(A_1\) with non-zero poloidal mode number vanish on the axis. In particular, regularity makes \(A_{1, \theta}\) vanish on the axis. SPEC additionally fixes the gauge so that \(A_{1, \zeta}\) is independent of \(\zeta\) there. The first is the \(i = 1\) case of the first equation of \eqref{eq:spec_inner} with \(\psi_{0, t} = 0\), and with the prescribed enclosed flux it makes \(\psi_{1,1}\) the flux of \(B_1\) through \(T_{1,1}\). It carries no gauge freedom, since a single-valued gauge function is constant in \(\theta\) on the axis. The second removes the \(\zeta\) dependence of \(A_{1, \zeta}\) on the axis using single-valued gauge functions alone, and leaves the value of \(A_{1, \zeta}\) there unfixed. That value is the only quantity a multivalued gauge function can shift, so the conditions on the axis do not determine the circulation of \(A_1\).

The remaining freedom is fixed on the outer interface. SPEC imposes \eqref{eq:spec_outer} in every region, including the innermost, with \(\varphi_1\) single-valued on \(I_1\) and with \(\psi_{1, t}\) and \(\psi_{1, p}\) prescribed. This is the gauge described by Hudson et al., who use the available freedom to specify the circulations of the potential on the interfaces and record that the gauge-invariant and gauge-dependent helicities then differ by a constant \cite{10.1063/1.4765691}. For \(i \ge 2\) the constant part of \(A_{i, \zeta}\) is already fixed by \eqref{eq:spec_inner}, so \eqref{eq:spec_outer} constrains the enclosed poloidal flux, whereas for \(i = 1\) it is \eqref{eq:spec_outer} that fixes the constant and no poloidal flux is prescribed, which is why the innermost volume carries one degree of freedom rather than two. Since \(\sigma_{1,1}\) lies on \(I_1\) and the tangency condition makes \(\mathcal{J}^* a_1\) closed there, the single circulation is \(\psi_{1, p} = 0\) by the convention above, so \(\kappa_1(a_{1,t}) = \psi_{1,1} \, \psi_{1, p} = 0\) is constant along any variation of the interfaces. Only the fact that \(\psi_{1, p}\) is input data is used, not its value. The flux of \(B_1\) through a surface joining the axis to \(f(\sigma_{1,1})\) is not prescribed and does move with the interface, but in this gauge that flux is the value of \(A_{1, \zeta}\) on the axis, which the code leaves to the Beltrami solve, and not the circulation about \(f(\sigma_{1,1})\). The conditions are linear in the potential and are solved by successive quadratures in \(r\) and \(\zeta\) together with one condition at \(r_1\), so they determine \(g_1\) up to an additive constant, \(G_1\) contains exactly one primitive of each \(\beta_1\), the induced selection is affine and continuous, and the variations are exhaustive in both senses required by Theorem \ref{thr:one_gauge}. Since the conditions are imposed in the fixed reference coordinates the same holds along deformations of the interfaces.

Theorem \ref{thr:one_gauge} therefore applies, and we expect solutions of the variational problem solved by SPEC to reproduce MRxMHD equilibria, in agreement with the numerical evidence.

\section{Characterisation for small \texorpdfstring{$|\mu_i|$}{Lagrange Multiplier}}\label{section:Characterisation and Local uniqueness}

Fix $n=1$ and $f = \mathrm{id} \in \mathcal{I}(\Lambda)$, so that $\Lambda_1 = \Lambda$ and every metric quantity in this section is the Euclidean one. As there is a single region used in this section we suppress the region index, writing $\beta$, $a$, $\mu$, $h$, $\ell$ and $\psi_j$, $T_j$, $\sigma_j$ for $\beta_1$, $a_1$, $\mu_1$, $h_1$, $\ell_1$ and $\psi_{1,j}$, $T_{1,j}$, $\sigma_{1,j}$ respectively. Problem \eqref{eq:ext} becomes
\begin{align}\label{eq:ext_frozen}
  \operatorname*{c.\,p.}_{a \in \Omega^1(\Lambda)} \ \frac{1}{2} \big\| \beta \big\|_{L^2}^2 \, ,
  \qquad
  \mathscr{H}(\beta,\Lambda) = h \, ,
  \qquad
  \big\langle [\beta] , [T_{j}] \big\rangle = \psi_{j} \quad (j = 1,\dots,\ell) \, ,
\end{align}
such that $da = \beta \in \mathcal{C}(\Lambda)$, and we write $E(\beta) = ( \sfrac12 )\|\beta\|_{L^2}^2   - \tilde p \myabs{\Lambda}$, the pressure term being constant. Since $n=1$ there are no interfaces, so Proposition \ref{thr:iff_simple} at fixed $f$ states that an admissible $\beta \neq 0$ is a critical point of \eqref{eq:ext_frozen} if and only if $d \star \beta = \mu \beta$ for the Lagrange multiplier $\mu$. Theorem \ref{prop:vanishing} covers $\beta = 0$. By Subsection \ref{sec:main_theorem} the constraint set of \eqref{eq:ext} is a product over the regions and, at fixed $f$, the criticality conditions in the potentials are separate conditions on each $\ST_i$. Everything below therefore applies verbatim to a single region $\ST_i$ of a general partition at fixed geometry, with $\beta_i$, $\mu_i$, $h_i$, $\ell_i$, $\psi_{i,j}$ in place of $\beta$, $\mu$, $h$, $\ell$, $\psi_j$.

Let \(\mu_-, \mu_+\) and \(\mu_\pm\) be as in Lemma \ref{lem:helicity_bound}, and we show that when $|\mu| < \mu_\pm$ the corresponding critical point is the unique global minimiser of the magnetic energy over the entire constraint set of \eqref{eq:ext_frozen}. Both $\|da\|_{L^2}^2 / 2$ and $\mathscr{H}(da,\Lambda)$ are quadratic in $a$, so the increment of each along an admissible perturbation is computed exactly and no second derivative test is needed. The results below apply to Problem \eqref{eq:ext_frozen} with a helicity constraint with an Amperian gauge as opposed to a loop-corrected helicity constraint. 

Write $\mathcal{C}_0(\Lambda) \coloneqq \{ \beta' \in \mathcal{C}(\Lambda) : \langle [\beta'],[T_j] \rangle = 0  \; \forall j\}$ for the flux-free subspace. By Subsections \ref{sub:Hodge} and \ref{sub:flux}, a field lies in $\mathcal{C}_0(\Lambda)$ precisely when its harmonic component vanishes, so
\begin{align}\label{eq:flux_free}
  \mathcal{C}_0(\Lambda) = \big\{ d\eta' \; : \; \eta' \in \Omega^1_D(\Lambda) \big\} \, ,
\end{align}
and by definition $\eta' \in \mathscr{A}(\Lambda)$. The loop corrections in our loop-corrected helicity are weighted by the fluxes of $\beta'$, so on $\mathcal{C}_0(\Lambda)$ the helicity is independent of the primitive and
\begin{align}\label{eq:flux_free_helicity}
  \mathscr{H}(\beta',\Lambda) = \int_\Lambda \eta' \wedge \beta' \, , \qquad \beta' = d\eta' , \ \eta' \in \Omega^1_D(\Lambda) \, .
\end{align}
\begin{lemma}\label{lem:helicity_bound}
  Let $\mu_+$ and $\mu_-$ denote respectively the smallest positive and the largest negative eigenvalue of the curl operator on the flux-free fields of $\Lambda$, that is of the problem $d\star\beta' = \mu\beta'$ for $\beta' \in \mathcal{C}_0(\Lambda)\setminus\{0\}$, and put $\mu_\pm \coloneqq \min\{|\mu_+|,|\mu_-|\}$, with the convention that $|\mu_+|$ or $|\mu_-|$ is $+\infty$ when the curl operator on $\mathcal{C}_0(\Lambda)$ has no positive, respectively no negative, eigenvalue. Then $\mu_\pm > 0$ and
  \begin{align}\label{eq:helicity_bound}
    \big| \mathscr{H}(\beta',\Lambda) \big| \;\leq\; \frac{1}{\mu_\pm} \big\| \beta' \big\|_{L^2}^2
  \end{align}
  for every $\beta' \in \mathcal{C}_0(\Lambda)$.
\end{lemma}

\begin{proof}
  Let $V$ be the closure of $\mathcal{C}_0(\Lambda)$ in $L^2\Omega^2(\Lambda)$. The curl operator on $V$ is self-adjoint with compact inverse, so its spectrum is a discrete set of real eigenvalues with no finite accumulation point, its eigenfields are smooth, and $V$ has an $L^2$-orthonormal basis $\{\beta_k\}$ of eigenfields $d\star\beta_k = \mu_k\beta_k$ \cite{soperimetric_problems_helicity, exist_gerner}. These \(\beta_k\) lie in $\mathcal{C}_0(\Lambda)$ because the fluxes are continuous linear functions of the $L^2$-orthogonal projection onto the finite dimensional space $\mathcal{H}^2_D(\Lambda)$. This is the only point at which the completion is used. 

  Writing $\beta' = d\eta'$ as in \eqref{eq:flux_free}
  $$
  \|\beta'\|_{L^2}^2 = \int_\Lambda \eta' \wedge d \star \beta' \, , 
  $$
  which is zero if $d\star\beta' = 0$. An eigenvalue of 0 therefore implies that \(\beta' = 0\) which is not allowed, hence by discreteness of the eigenvalues $\mu_\pm > 0$.

  Let $\beta' \in \mathcal{C}_0(\Lambda)$ with primitive $\eta' \in \Omega^1_D(\Lambda)$ and expand 
\[
    \beta' = \sum_k c_k \beta_k
\]
 in $L^2$, so that 
\[
    \|\beta'\|_{L^2}^2 = \sum_k c_k^2
\]
 and
  \begin{align*}
    \int_\Lambda \eta' \wedge \beta_k = \frac{1}{\mu_k} \int_\Lambda \eta' \wedge d\star\beta_k = \frac{1}{\mu_k} \big\langle \beta', \beta_k \big\rangle_{L^2} = \frac{c_k}{\mu_k} \, .
  \end{align*}
  By $L^2$-continuity of the integral we may pass to the limit in the expansion of $\beta'$ and \eqref{eq:flux_free_helicity} gives that 
\[
    \mathscr{H}(\beta',\Lambda) = \sum_k c_k^2/\mu_k
\]
 is absolutely convergent. Hence 
\[
    |\mathscr{H}(\beta',\Lambda)| \leq \mu_\pm^{-1} \sum_k c_k^2\, ,
\]
 which is \eqref{eq:helicity_bound}.
\end{proof}

\begin{theorem}\label{thr:local_min}
  Let $\beta \in \mathcal{C}(\Lambda)$ be the exterior derivative of a critical point of \eqref{eq:ext_frozen} with associated Lagrange multiplier $\mu$, so that $d \star \beta = \mu\beta$, and suppose that
  \begin{align}\label{eq:mu_small}
    |\mu| < \mu_\pm \, .
  \end{align}
  Then for every $\hat\beta \in \mathcal{C}(\Lambda)$ satisfying the constraints of \eqref{eq:ext_frozen},
  \begin{align}\label{eq:coercive}
    E \big( \hat\beta \big) - E \big( \beta \big)
    \;\geq\; \frac{1}{2}\left( 1 - \frac{|\mu|}{\mu_\pm} \right) \big\| \hat\beta - \beta \big\|_{L^2}^2
    \;\geq\; 0 \, ,
  \end{align}
  with equality if and only if $\hat\beta = \beta$. In particular $\beta$ is the unique global minimiser of the magnetic energy over the constraint set of \eqref{eq:ext_frozen}.
\end{theorem}

\begin{proof}
Write $\tilde\beta \coloneqq \hat\beta - \beta$. As $\mathcal{C}(\Lambda)$ is a linear space, $\tilde\beta \in \mathcal{C}(\Lambda)$, and the flux constraints give
  \begin{align}\label{eq:admissible}
    \big\langle [\tilde\beta],[T_j] \big\rangle = \psi_j - \psi_j = 0 \, , \qquad j = 1,\dots,\ell \, ,
  \end{align}
  so $\tilde\beta \in \mathcal{C}_0(\Lambda)$, which is the class to which Lemma \ref{lem:helicity_bound} applies. Fix $\tilde\eta \in \Omega^1_D(\Lambda) $ with $d\tilde\eta = \tilde\beta$, as in \eqref{eq:flux_free}.

  Let $a \in \mathscr{A}(\Lambda)$ be an Amperian primitive of $\beta$, which exists by Lemma \ref{Lemma:gauge_exist}. Then $a + \tilde\eta$ is an Amperian primitive of $\hat\beta$, so 
  \begin{align*}
    \mathscr{H}\big( \hat\beta,\Lambda \big) = \int_\Lambda (a+\tilde\eta) \wedge d(a+\tilde\eta) = \mathscr{H}(\beta,\Lambda) + 2 \int_\Lambda  \tilde\eta \wedge \beta  + \mathscr{H}(\tilde\beta,\Lambda) \, .
  \end{align*}
  Applying the helicity constraint $\mathscr{H}(\hat\beta,\Lambda) = h = \mathscr{H}(\beta,\Lambda)$ gives
  \begin{align}\label{eq:constraint_exact}
    2 \int_\Lambda \tilde\eta \wedge \beta = - \mathscr{H}(\tilde\beta,\Lambda) \, .
  \end{align}
  We emphasise that \eqref{eq:constraint_exact} is exact and no smallness assumption on $\tilde\eta$ has been made. Criticality enters through the Beltrami equation,
  \begin{align}\label{eq:criticality}
    \big\langle \beta , \tilde\beta \big\rangle_{L^2} = \int_\Lambda \tilde\eta \wedge d\star\beta = \mu \int_\Lambda \tilde\eta \wedge \beta \, .
  \end{align}
  Expanding the objective, which is quadratic, and substituting \eqref{eq:criticality} and then \eqref{eq:constraint_exact},
  \begin{align*}
    E \big( \hat\beta \big) - E \big( \beta \big)
    = \big\langle \beta , \tilde\beta \big\rangle_{L^2} + \frac{1}{2} \big\| \tilde\beta \big\|_{L^2}^2
    = \mu \int_\Lambda \tilde\eta \wedge \beta + \frac{1}{2} \big\| \tilde\beta \big\|_{L^2}^2
    = \frac{1}{2} \Big( \big\| \tilde\beta \big\|_{L^2}^2 - \mu \, \mathscr{H}(\tilde\beta,\Lambda) \Big) \, ,
  \end{align*}
  the pressure term cancelling because $f$ is frozen. Estimating by $-\mu x \geq -|\mu||x|$ and then by \eqref{eq:helicity_bound},
  \begin{align*}
    \big\| \tilde\beta \big\|_{L^2}^2 - \mu \, \mathscr{H}(\tilde\beta,\Lambda)
    \;\geq\; \big\| \tilde\beta \big\|_{L^2}^2 - |\mu| \, \big| \mathscr{H}(\tilde\beta,\Lambda) \big|
    \;\geq\; \left( 1 - \frac{|\mu|}{\mu_\pm} \right) \big\| \tilde\beta \big\|_{L^2}^2 \, ,
  \end{align*}
  which on dividing by $2$ is \eqref{eq:coercive}. By \eqref{eq:mu_small} the constant is strictly positive, so the right-hand side vanishes if and only if $\tilde\beta = 0$.
\end{proof}

\section{Conclusion}\label{section:Conclusion}

In this work, we developed a variational framework for solutions to the multi-region relaxed magnetohydrodynamics (MRxMHD) equations on compact, oriented domains embedded in \(\mathbb{R}^3\) with smooth boundary. We use loop-corrected helicity, represented by a curl-free reference potential whose free loop integrals match those of the chosen potential. We demonstrated the existence of an Amperian gauge, noting that the loop-corrected helicity aligns with conventional helicity in this gauge. Within each subdomain, we imposed physically motivated conditions, including divergence-free fields, boundary tangency, and prescribed loop-corrected helicity, flux, and pressure constraints.

We showed that under these constraints non-zero critical points of the magnetic energy correspond precisely to solutions of the MRxMHD equilibrium equations. We isolated a gauge condition which, in the toroidal setting, is imposed implicitly through the gauge function being single-valued, and showed that the variation of the loop-corrected helicity cancels the corresponding boundary contribution on a general partition. We verified that it is satisfied by the gauge used in the Stepped Pressure Equilibrium Code. In the presence of a pressure jump, we established the equivalence between the MRxMHD formulation and a stronger formulation by Bruno and Laurence.


Our results extend the class of domains where one may apply an MRxMHD model beyond previously considered geometries.

\section*{Acknowledgements}
This research was supported by an Australian Government Research Training Program (RTP) Scholarship at The University of Western Australia. The authors gratefully acknowledge this funding assistance.

%
\section*{Declaration of interests}
The authors report no conflict of interest.

\section*{Data availability}
No datasets were generated or analysed during the current study, as all necessary information is fully contained within the article.



\bibliographystyle{abbrvnat}
\bibliography{ref}

\newpage
\section*{Appendix}
\appendix



\section{De Rham Cohomology and Singular Homology}\label{section:Homology cohomology construction}

\subsection{Construction} \label{sub:construct}

For each \(k \in \mathbb{Z}\), the set of all smooth \(k\)-forms on \(\ST_i\) is \(\Omega^k(\ST_i)\). This set is an abelian group with a homomorphism given by the exterior derivative \(d :\Omega^k(\ST_i)\longrightarrow\Omega^{k+1}(\ST_i)\) satisfying \(d^2 = 0\). For \(k < 0\) or \(k > 3\) we have that \(\Omega^k(\ST_i) \cong 0\). The sequence \(\dots , \, \Omega^0(\ST_i), \, \Omega^1(\ST_i), \, \dots\) with operations \(d\) is commonly called the \emph{de Rham complex} of \(\ST_i\). The \(k\)-th \emph{de Rham cohomology} group of \(\ST_i\) is the \(k\)-th cohomology of the de Rham complex,
\[
H^k_{\mathrm{dR}}(\ST_i) \coloneqq
\frac{\ker\big(d:\Omega^k(\ST_i)\to\Omega^{k+1}(\ST_i)\big)}
     {\operatorname{im}\big(d:\Omega^{k-1}(\ST_i)\to\Omega^{k}(\ST_i)\big)} \, .
\]
If \(\omega\in\Omega^k(\ST_i)\) is closed (\(d\omega=0\)), its cohomology class \([\omega]\in H^k_{\mathrm{dR}}(\ST_i)\) is the coset
\[
[\omega]=\{\omega + d\nu:\nu\in\Omega^{k-1}(\ST_i)\} \, ,
\]
so two closed \(k\)-forms represent the same class precisely when their difference is exact. Similarly, one defines the \(k\)-th cohomology \(H^k_{\mathrm{dR}} (\partial \ST_i) \).

The pullback \(\mathcal{J}^*:\Omega^k(\ST_i)\to\Omega^k(\partial\ST_i)\) of an inclusion map \(\mathcal{J}:\partial\ST_i\hookrightarrow\ST_i\) commutes with the exterior derivative. Therefore, the kernel of \(\mathcal{J}^*\), denoted \(\Omega^k_D(\ST_i) \subset \Omega^k(\ST_i)\), forms a sequence 
\[
\dots, \, \Omega^0_D(\ST_i), \, \Omega^1_D(\ST_i), \, \dots
\] 
of abelian groups with connecting homomorphisms \(d\). This is referred to as the \emph{de Rham complex} of \(\ST_i\) \emph{relative} to \(\partial \ST_i\). The corresponding \(k\)-th \emph{de Rham cohomology} of \(\ST_i\) \emph{relative} to \(\partial \ST_i\) is
\[
H^k_{\mathrm{dR}}(\ST_i,\partial\ST_i)
\coloneqq
\frac{\ker\big(d:\Omega^k_D(\ST_i)\to\Omega^{k+1}_D(\ST_i)\big)}
     {\operatorname{im}\big(d:\Omega^{k-1}_D(\ST_i)\to\Omega^{k}_D(\ST_i)\big)} \, .
\]
Relative de Rham cohomology classes are represented by closed \(k\)-forms that vanish under \(\mathcal{J}^*\) modulo exact \(k\)-forms with primitives vanishing under \(\mathcal{J}^*\).

For each integer \(k\ge 0\) and points \(v_0,\dots,v_k \in \R^k\), the \emph{geometric \(k\)-simplex} spanned by \(v_0,\dots,v_k\) is the set,
\[
[v_0,\dots,v_k]
\coloneqq
\Bigg\{\sum_{j=0}^k t_j v_j \, :\; t_j\in[0,1]\text{ for all }j,\ \sum_{j=0}^k t_j=1\Bigg\} \, .
\]
If the vertices \(v_0,\dots,v_k\) are affinely independent this set is a \(k\)-dimensional simplex, otherwise it lies in a lower-dimensional affine subspace. Let \(e_0 \coloneqq 0\in\mathbb{R}^k\) and, for \(i=1,\dots,k\), let \(e_i\) be the \(i\)-th standard basis vector of \(\mathbb{R}^k\). The \emph{standard \(k\)-simplex} is \([e_0,e_1,\dots,e_k] \subset \R^k\).

A \emph{singular \(k\)-simplex} in \(\ST_i\) is a continuous map
\[
    \sigma: [e_0,e_1,\dots,e_k] \longrightarrow \ST_i \, .
\]
The \(k\)th \emph{singular chain group} of \(\ST_i\) with coefficients in the real numbers is the real vector space of finite, formal linear combinations of singular \(k\)-simplices:
\[
C_k(\ST_i)
\coloneqq
\Bigg\{\sum_{j=1}^N t_j \sigma_j \, : \; 0 < N \in\mathbb{N},\   \ t_j\in\mathbb{R} \text{ and } \sigma_j \text{ is a singular \(k\)-simplex for all } j \Bigg\} \, .
\]
Elements of \(C_k(\ST_i)\) are called \emph{singular \(k\)-chains}. For \(k < 0\), we set \(C_k(\ST_i) = \{0\}\).

Consider a boundary operator \( \partial\) that acts on any singular \(k\)-simplex \(\sigma\) and returns a singular \((k-1)\)-chain via
\[
 \partial \sigma
=\sum_{j=0}^k (-1)^j\, \sigma \circ \mathfrak{d}_j \, ,
\]
where the face map \(\mathfrak{d}_j\) is the affine map
\[
    \mathfrak{d}_j : [e_0, \dots, e_{k-1}] \longrightarrow [e_0, \dots, e_{k}]
\]
carrying the vertices \(e_0, \dots, e_{k-1}\) to \(e_0, \dots, e_{j-1}, e_{j+1}, \dots, e_k\) in that order. Since \(e_0 = 0\), a point of \([e_0, \dots, e_{k-1}] \subset \R^{k-1}\) is written \((t_1, \dots, t_{k-1})\) with \(t_r \geq 0\) and 
\[
    \sum_{r=1}^{k-1} t_r \leq 1\, ,
\]
 and
\[
    \mathfrak{d}_0 : (t_1 , \dots, t_{k-1}) \mapsto \Big( 1 - \sum_{r=1}^{k-1} t_r , \, t_1 , \dots,  t_{k-1} \Big) \, , 
\]
\[
    \mathfrak{d}_j : (t_1 , \dots, t_{k-1}) \mapsto (t_1 , \dots, t_{j-1}, 0, t_j , \dots,  t_{k-1}) \, , \qquad j = 1, \dots, k \, .
\]
The face opposite \(e_0\) is the set on which the coordinates sum to \(1\), so it is not obtained by inserting a zero, which is why \(\mathfrak{d}_0\) is not of the second form. For \(k=2\) this gives \(\partial \sigma : w \mapsto \sigma(1-w,w) - \sigma(0,w) + \sigma(w,0)\), the formula used in Appendix \ref{append:example}.
This boundary operator is uniquely extended to a boundary operator on singular \(k\)-chains via 
\[
\partial : C_k(\ST_i)\longrightarrow C_{k-1}(\ST_i) \, , \quad \text{via} \quad \partial_k w = \partial_k \sum_{j=1}^N t_j \sigma_j =  \sum_{j=1}^N t_j \partial_k \sigma_j \, .
\]
for a singular \(k\)-chain \(w \) where \(t_j \in \R\) and \(\sigma_j\) is a singular \(k\)-simplex for all \(j\). We will identify the domain for the boundary map \(\partial\) whenever it is unclear from context.

The operators \(\partial\) acting on \(k\)-chains are homomorphisms satisfying \(\partial\circ\partial=0\) \cite{lee}, hence a sequence \(\dots , \, C_1(\ST_i) \mathbin{,} \, C_0(\ST_i) \mathbin{,} \, \dots \) with connecting boundary operators is a chain complex called the \emph{singular chain complex} of \(\ST_i\). The \(k\)th \emph{singular homology} group of \(\ST_i\) with real coefficients is the quotient vector space 
\[
  H_k(\ST_i)\coloneqq  \frac{\ker(\partial: C_k(\ST_i)\longrightarrow C_{k-1}(\ST_i)) }{  \operatorname{im}(\partial: C_{k+1}(\ST_i)\longrightarrow C_{k}(\ST_i))} \, .
\]
Similarly, one defines a homology group for the boundary \(H_k(\partial \ST_i)\). An inclusion \(\mathcal{J}:\partial \ST_i \hookrightarrow \ST_i\) induces an inclusion of chain complexes \( \mathcal{J}_* : C_k(\partial \ST_i )\hookrightarrow C_k(\ST_i)\) via a pushforward \(\mathcal{J}_*\). The quotient
\[
  C_k(\ST_i,\partial \ST_i )\coloneqq C_k(\ST_i)/C_k(\partial \ST_i ) \, ,
\]
is therefore a vector space for each \(k\), and the boundary map \(\partial\) acting on \(C_k(\ST_i)\) descends to a well-defined linear map on the quotient \citep[page 115]{hatcher2005algebraic}
\[
   \partial: C_k(\ST_i,\partial \ST_i )\longrightarrow C_{k-1}(\ST_i,\partial \ST_i ) \, .
\]
So \(\dots , \, C_1(\ST_i, \partial \ST_i ) , \, C_0(\ST_i, \partial \ST_i ), \, \dots \) with the boundary operations is a chain complex. The \(k\)th homology group of \(\ST_i\) relative to \(\partial \ST_i\) is 
\[
  H_k(\ST_i,\partial \ST_i)\coloneqq  \frac{\ker( \partial: C_k(\ST_i, \partial \ST_i )\longrightarrow C_{k-1}(\ST_i, \partial \ST_i)) }{  \operatorname{im}(  \partial: C_{k+1}(\ST_i, \partial \ST_i)\longrightarrow C_{k}(\ST_i, \partial \ST_i))} \, .
\]

\subsection{Proof of Lemma \ref{Lemma:finite}}\label{sub:finite}

Recall Lemma \ref{Lemma:finite}:
\begin{center}
    \textit{The homologies \(H_k (\ST_i)\) and \(H_k (\overline{\ST}_i^c)\) are finite dimensional vector spaces for all \(k,i\).}
\end{center}

These homologies are vector spaces over the reals by construction, so we prove the finite dimensionality. We follow a discussion given by Bott and Tu \citeyearpar{bott2013differential}.





\begin{proof}
  A finite open cover \(\mathcal{U}=\{U_1,\dots,U_N\}\) of an open subset of \(\R^3\) is called a \emph{good cover} if every non-empty finite intersection of the sets \(U_j\) is diffeomorphic to \(\mathbb{R}^3\). As \(\ST_i\) is a manifold with boundary we do not cover \(\ST_i\) itself but an open set containing it. Let \((\ST_i)_\epsilon\) be the set of all points of \(\R^3\) within \(\epsilon>0\) of \(\ST_i\). Since \(\partial \ST_i\) is smooth and compact, the collar neighbourhood theorem supplies an \(\epsilon\) for which \((\ST_i)_\epsilon\) is an open Riemannian \(3\)-manifold that deformation retracts onto \(\ST_i\), the retraction restricting to the identity on \(\ST_i\).

Fix a point \(x\in \ST_i\), which is an interior point of \((\ST_i)_\epsilon\). By the existence and uniqueness theorem for ordinary differential equations, for each \(v\in T_x(\ST_i)_\epsilon\) there exists \(\varepsilon>0\) and a unique geodesic
\[
    \gamma_{v}:(-\varepsilon,\varepsilon)\to (\ST_i)_\epsilon
\]
such that \(\gamma_{v}(0)=x\) and \(\gamma_{v}'(0)=v\). The exponential map at \(x\) is therefore defined on a neighbourhood \(U\) of the origin in \(T_x(\ST_i)_\epsilon\) such that for every \(v \in U\), the geodesic \(\gamma_{v}\) is also defined on \([0,1]\). We then define the exponential map by
\[
    \exp_x : U \subseteq  T_x(\ST_i)_\epsilon \longrightarrow (\ST_i)_\epsilon \, , 
    \qquad 
    \exp_x(v)\coloneqq \gamma_{v}(1) \, .
\]
By the homogeneity of geodesics, \(\exp_x(sv)=\gamma_{v}(s)\) for \(s \in [0,1]\). Differentiating at \(s=0\) gives
\[
    d(\exp_x)_0 (v)
    =
    \frac{d}{ds}\Big|_{s=0}\exp_x(sv)
    =
    \gamma_{v}'(0)
    =
    v \, ,
\]
and hence \(d(\exp_x)_0\) is the identity on \(T_x (\ST_i)_\epsilon\). Only the value of the differential at the origin is needed, and the inverse function theorem implies that there exists a ball \(\B(r) \subset  T_x(\ST_i)_\epsilon\) with radius \(r>0\) about \(0\) such that 
\[
    \exp_x : U = \B(r)\longrightarrow V_x\coloneqq \exp_x\bigl(\B(r)\bigr)
\]
is a diffeomorphism onto an open neighbourhood \(V_x\) of \(x\) in \((\ST_i)_\epsilon\).

By choosing \(r\) sufficiently small \(V_x\) is geodesically convex \cite{lee} and \(V_x \subseteq (\ST_i)_\epsilon\). That is, for any two points \(y,z\in V_x\), there exists a unique geodesic segment joining \(y\) to \(z\), and this segment lies entirely in \(V_x\).

The collection \(\{V_x : x\in \ST_i\}\) is therefore an open cover of the compact set \(\ST_i\), so it admits a finite subcover, namely, there exist points \(x_1,\dots,x_N\in \ST_i\) such that
\[
    \ST_i \subseteq V \coloneqq V_{x_1}\cup\cdots\cup V_{x_N} \subseteq (\ST_i)_\epsilon \, .
\]
Moreover, each \(V_{x_j}\) is geodesically convex.

Consider the geodesically convex neighbourhoods \(V_{x_1},V_{x_2}\) that are small enough so that the geodesic between any two points within the neighbourhoods is unique, as in the above discussion. Then in \(V_{x_1} \cap V_{x_2}\), the geodesic is unique in both sets separately, and contained in both sets separately, so \(V_{x_1} \cap V_{x_2}\) is geodesically convex. 

Each \(V_x\) is diffeomorphic to \(\B(r) \subset T_x (\ST_i)_\epsilon \cong \R^3\), and by the paragraph above every non-empty finite intersection of the \(V_{x_j}\) is geodesically convex and therefore has the same property. Hence \(\{V_{x_1},\dots,V_{x_N}\}\) is a finite good cover of \(V\).

Recall the Poincaré Lemma: \(H^k_{\mathrm{dR}}(\mathbb{R}^3)\cong 0 \) for \(k>0\), and \(H^0_{\mathrm{dR}}(\mathbb{R}^3)\cong\mathbb{R}\). We prove by induction on the number of geodesically convex sets in a union that the de Rham cohomology of such a union is finite dimensional. Applying this to the finite good cover of \(V\) will give the desired statement.

For a single geodesically convex neighbourhood \(V_{x_1}\) the Poincaré Lemma (and the fact that \(V_{x_1}\) is diffeomorphic to an open ball in \(\mathbb{R}^3\)) implies \(H^k_{\mathrm{dR}}(V_{x_1})\) is finite dimensional.

Now suppose, for the sake of induction, that \(H^k_{\mathrm{dR}}(V)\) is finite dimensional where
\[
V=V_{x_1}\cup\cdots\cup V_{x_j} \, .
\]
The Mayer--Vietoris long exact sequence for de Rham cohomology gives a segment
\[
\cdots \longrightarrow
H^{k-1}_{\mathrm{dR}}(V\cap V_{x_{j+1}})
\xrightarrow{\;d^*\;}
H^{k}_{\mathrm{dR}}(V\cup V_{x_{j+1}})
\xrightarrow{\;\mathcal{J}_*\;}
H^{k}_{\mathrm{dR}}(V)\oplus H^{k}_{\mathrm{dR}}(V_{x_{j+1}})
\longrightarrow\cdots \, .
\]
Exactness at \(H^{k}_{\mathrm{dR}}(V\cup V_{x_{j+1}})\) yields
\[
\ker(\mathcal{J}_*)=\operatorname{Im}(d^*) \, ,
\]
hence, a short exact sequence of vector spaces is
\[
0 \longrightarrow \operatorname{Im}(d^*) \longrightarrow
H^{k}_{\mathrm{dR}}(V\cup V_{x_{j+1}}) \longrightarrow
\operatorname{Im}(\mathcal{J}_*) \longrightarrow 0 \, .
\]
Consequently
\[
\dim H^{k}_{\mathrm{dR}}(V\cup V_{x_{j+1}})
= \dim\operatorname{Im}(d^*) + \dim\operatorname{Im}(\mathcal{J}_*) \, .
\]
But \(\operatorname{Im}(d^*)\) is a quotient (hence a homomorphic image) of \(H^{k-1}_{\mathrm{dR}}(V\cap V_{x_{j+1}})\), so it is finite dimensional whenever \(H^{k-1}_{\mathrm{dR}}(V\cap V_{x_{j+1}})\) is finite dimensional. Also \(\operatorname{Im}(\mathcal{J}_*)\) is a subspace of the finite dimensional space \(H^{k}_{\mathrm{dR}}(V)\oplus H^{k}_{\mathrm{dR}}(V_{x_{j+1}})\), so it is finite dimensional. Therefore, if
\[
H^{k-1}_{\mathrm{dR}}(V\cap V_{x_{j+1}}) \, ,\qquad
H^{k}_{\mathrm{dR}}(V) \, ,\qquad
H^{k}_{\mathrm{dR}}(V_{x_{j+1}}) \, , 
\]
are finite dimensional, then \(H^{k}_{\mathrm{dR}}(V\cup V_{x_{j+1}})\) is finite dimensional. We still need to verify that \(H^{k-1}_{\mathrm{dR}}(V\cap V_{x_{j+1}})\) is finite. Note that
\[
V\cap V_{x_{j+1}}
=\bigcup_{i=1}^j \bigl(V_{x_i}\cap V_{x_{j+1}}\bigr) \, .
\]
The sets \(V_{x_1}\cap V_{x_{j+1}},\dots,V_{x_j}\cap V_{x_{j+1}}\) are geodesically convex, as is every non-empty finite intersection of them, so together they form a finite good cover of \(V\cap V_{x_{j+1}}\). Hence, \(H^{k-1}_{\mathrm{dR}}(V\cap V_{x_{j+1}})\) is finite dimensional. Therefore, \(\operatorname{Im}(d^*)\) is finite dimensional and the argument above shows \(H^{k}_{\mathrm{dR}}(V\cup V_{x_{j+1}})\) is finite dimensional.

By induction, we conclude that the finite union
\[
V = V_{x_1}\cup\cdots\cup V_{x_N}
\]
produces a finite dimensional vector space \(H^k_{\mathrm{dR}}(V)\). Write \(\mathcal{J} : \ST_i \hookrightarrow V\) for the inclusion. The retraction of \((\ST_i)_\epsilon\) onto \(\ST_i\) restricts to a continuous map \(V \to \ST_i\) whose composition with \(\mathcal{J}\) is the identity of \(\ST_i\). Thus the induced map on real singular cohomology has a right inverse. By the natural de Rham isomorphism, the same holds for \(\mathcal{J}^*\) on de Rham cohomology. Hence \(\mathcal{J}^* : H^k_{\mathrm{dR}}(V) \to H^k_{\mathrm{dR}}(\ST_i)\) is surjective and \(H^k_{\mathrm{dR}}(\ST_i)\) is finite dimensional.

Using the de Rham theorem we see that \(H^k_{\mathrm{dR}} (\ST_i)\) has the same dimension as \( H_{k} (\ST_i) \), so \(H_{k} (\ST_i)\) is finite dimensional. The same argument applies to \(\overline{\ST}_i^c\), which is likewise a compact manifold with smooth boundary in \(\mathbb{S}^3\), so that the same neighbourhood and retraction are available there.
\end{proof}




\subsection{Homology of Complements}\label{sub:complement}

Consider a result by \citet{CANTARELLA20101127}: Let \(\mathbb{S}^k\) and \(\mathbb{D}^k\) be the usual \(k\)-dimensional sphere and ball. Thinking of \(\R^3\) topologically as \(\mathbb{S}^3 \setminus \{x\}\) for \(x \in \mathbb{S}^3\setminus\ST_i\) we have
\begin{align*}
    H_1( \R^3 \setminus \ST_i ) \cong H_1( \mathbb{S}^3 \setminus \ST_i ) \, . 
\end{align*}
\begin{proof}
    Let \(\mathbb{D}^3\) be a ball about \(x\) and not intersecting \(\ST_i\). The union and intersection of \(\mathbb{D}^3\) and \(\R^3 \setminus \ST_i \cong (\mathbb{S}^3 \setminus\{x\}) \setminus \ST_i\) is
    \begin{align*}
        \mathbb{D}^3 \cup ( (\mathbb{S}^3 \setminus \{x\}) \setminus \ST_i ) = \mathbb{S}^3 \setminus \ST_i \, ,
        \qquad 
        \text{and}
        \qquad 
        \mathbb{D}^3 \cap ( (\mathbb{S}^3 \setminus \{x\}) \setminus \ST_i ) = \mathbb{D}^3  \setminus \{x\}  \simeq \mathbb{S}^2 \, .
    \end{align*}
    This provides a segment of a Mayer--Vietoris long exact sequence:
    \begin{align*}
        H_1(\mathbb{S}^2) \to H_1(\mathbb{D}^3) \oplus H_1( \R^3 \setminus \ST_i )  \to H_1( \mathbb{S}^3 \setminus \ST_i ) \to H_0(\mathbb{S}^2) \, .
    \end{align*}
    Observe that $H_0(S^2)\to H_0(D^3)\oplus H_0(\mathbb{R}^3\setminus\ST_i)$ is injective, so the connecting map in the above sequence vanishes. With this, the result is implied by the fact that \(\mathbb{D}^3\) is contractible and the first homology group in the exact sequence is zero. Since \(\partial \ST_i\) is smooth and compact, a collar of \(\partial \ST_i\) in \(\overline{\ST}_i^c\) makes the inclusion \(\mathbb{S}^3 \setminus \ST_i \hookrightarrow \overline{\ST}_i^c\) a homotopy equivalence, so the statement above is also a statement about the closed complement \(\overline{\ST}_i^c\) of Subsection \ref{sub:coho_homo}.
\end{proof}

\section{Magnetic Energy Comparison}\label{section:energy}

The purpose of this appendix is to record the (known) fact that several formulations of the MRxMHD variational problem possess the same critical points, so that Proposition \ref{thr:iff_simple} applies to each of them. Throughout this appendix we assume familiarity with the notation and the results of the main text. Although the appendix is referred to in the Introduction, and we use results derived later in the paper, no circularity is introduced. 

Although many of the quoted results use explicit changes to the domain, we quote them as though they are changes through the function \(f\) because these are equivalent. Namely, every smooth deformation of the interfaces is induced by a variation of \(f\), and conversely, so that the two notions of first variation agree.

\begin{remark}\label{rem:helicity_gauge}
The formulations quoted below impose the helicity constraint with respect to a particular gauge, rather than through the loop-corrected helicity of Subsection \ref{subsection:Relative Helicity}, and they do not all use the same helicity: \citet{10.1063/1.4765691} constrains the ungauged helicity integral and observes that it is not gauge invariant, \citet{10.1063/1.4795739} adds the toroidal and poloidal loop integrals so as to obtain a gauge invariant helicity, and \citet{Dewar_Yoshida_Bhattacharjee_Hudson_2015} constrains the ungauged integral while restricting to single-valued gauge transformations. We write the constraint as \(H_{G_i}(\beta_i,\ST_i) = h_i\) in all four problems below, so that the same helicity is used in each. The comparison carried out here concerns only the pressure terms, which involve neither the potentials of \(\beta_i\) nor the helicity constraint, and the helicity constraint enters each of the four problems in exactly the same way. Consequently, every statement below is unchanged if the constraint \(H_{G_i}(\beta_i,\ST_i) = h_i\) is replaced by the corresponding loop-corrected helicity constraint, provided that the same constraint is used in all four problems. 
\end{remark}

Fix a partition \((\ST_1,\dots,\ST_n)\) of \(\ST\) and let \(\gamma > 1\) denote the ratio of specific heats. Throughout we work in units in which \(\mu_0 = 1\). For \(f \in \I\), \(\beta_i \in \mathcal{C}(\ST_i)\) and pressures \(\tilde p_i \colon \ST_i \to \R\) we recall that the magnetic energy \eqref{eq:original_energy} is written in terms of exterior calculus and pulled back to \(\ST\) as
\begin{align}\label{eq:classical_energy}
    W(f, \beta_1, \dots, \beta_n, \tilde p_1, \dots, \tilde p_n) =
    \sum_{i=1}^n \left( \frac{1}{2} \big\| \beta_i \big\|_{L^2}^2
    + \int_{\ST_i} \frac{\tilde p_i}{\gamma - 1} \, f^* \varpi \right) \, ,
\end{align}
A formulation for the MRxMHD problem is due to Dennis, Hudson, Dewar and Hole \citeyearpar{10.1063/1.4795739}, who take the pressures to be unknown functions and prescribe the entropy \(\mathcal{S}_i\) of each region,
\begin{align} \label{prob:MRxMHD_entropy}
    \begin{gathered}
        \underset{\substack{f \in \I \\ \beta_i \in \mathcal{C}(\ST_i) \\
        \hspace{0.9em} \tilde p_i \in C^\infty(\ST_i), \; \tilde p_i > 0}}{\text{c. p.}}
        W(f, \beta_1, \dots, \beta_n, \tilde p_1, \dots, \tilde p_n) \, , \\
        H_{G_i}(\beta_i, \ST_i) = h_i \, , \qquad
        \int_{T_{i,j}} \beta_i = \psi_{i,j} \, , \qquad
        \int_{\ST_i} \tilde p_i^{1/\gamma} \, f^* \varpi = \mathcal{S}_i \, .
    \end{gathered}
\end{align}
No sign condition on the pressure is imposed in \citet{10.1063/1.4795739}, namely the restriction \(\tilde p_i > 0\) is a difference from the precise variational problem in the corresponding paper, but one we impose. This ensures that \(\smash{ \tilde p_i^{1/\gamma} }\) is defined and smooth, and that the set of admissible pressures is open, so that critical points are characterised by two-sided variations.

The second is due to Hudson, Dewar, Dennis, Hole, McGann, von Nessi and Lazerson \citeyearpar{10.1063/1.4765691}, who assume that the pressure is piecewise constant and equal to \(\tilde p_i\) on each region \(\ST_i\) and impose an adiabatic law,
\begin{align} \label{prob:MRxMHD_adiabatic_app}
    \begin{gathered}
        \underset{\substack{f \in \I \\ \beta_i \in \mathcal{C}(\ST_i) \\
        \hspace{0.9em} \tilde p_i \in C^\infty ( \ST_i ) }}{\text{c. p.}} W(f, \beta_1, \dots, \beta_n, \tilde p_1, \dots, \tilde p_n) \, , \\
        H_{G_i}(\beta_i, \ST_i) = h_i \, , \qquad  \int_{T_{i,j}} \beta_i = \psi_{i,j} \, , \qquad  \tilde p_i \myabs{f(\ST_i)}^\gamma = v_i \, ,
    \end{gathered}
\end{align}
where \(v_1, \dots, v_n\) are prescribed constants. Here \(\tilde p_i\) is constant on \(\ST_i\), as assumed in \citet{10.1063/1.4765691}, so that \(\tilde p_i \myabs{f(\ST_i)}^\gamma = v_i\) is an identity between real numbers. This is equivalent to the formulation \eqref{prob:MRxMHD_adiabatic} that we quote in the Introduction. 

The third is the formulation of Dewar, Yoshida, Bhattacharjee and Hudson \citeyearpar{Dewar_Yoshida_Bhattacharjee_Hudson_2015}, in which the pressure \(\tilde p_i\) and the density \(\tilde \rho_i\) are unknown functions subject to a prescribed mass \(M_i\) and a prescribed thermodynamic entropy \(\mathcal{E}_i\) of an ideal gas,
\begin{align} \label{prob:MRxMHD_thermo}
    \begin{gathered}
        \underset{\substack{f \in \I \\ \beta_i \in \mathcal{C}(\ST_i) \\
        \hspace{0.9em} \tilde p_i, \tilde \rho_i \in C^\infty(\ST_i), \; \tilde p_i, \tilde \rho_i > 0}}{\text{c. p.}}
        W(f, \beta_1, \dots, \beta_n, \tilde p_1, \dots, \tilde p_n) \, , \\
        H_{G_i}(\beta_i, \ST_i) = h_i \, , \;
        \int_{T_{i,j}} \beta_i = \psi_{i,j} \, , \;
        \int_{\ST_i} \tilde \rho_i \, f^* \varpi = M_i \, , \;
        \int_{\ST_i} \frac{\tilde \rho_i}{\gamma - 1}
        \ln \! \left( \kappa \frac{\tilde p_i}{\tilde \rho_i^{\gamma}} \right) f^* \varpi
        = \mathcal{E}_i \, ,
    \end{gathered}
\end{align}
where \(M_1, \dots, M_n\) are positive constants and \(\kappa > 0\) is a constant fixing the units of the entropy. Positivity of \(\tilde p_i\) and \(\tilde \rho_i\) is not recorded as an assumption in \citet{Dewar_Yoshida_Bhattacharjee_Hudson_2015}, but it is implicit there, since the logarithm requires the argument to be positive and the entropy density is obtained from the Sackur--Tetrode equation for an ideal gas of positive number density and positive temperature. The first two problems are usually stated on a domain \(\ST\) that is a solid torus, with the interfaces \(I_q\) nested toroidal surfaces as in Subsection \ref{sub:back}, whereas \citet{Dewar_Yoshida_Bhattacharjee_Hudson_2015} already permits partitions that are not nested, such as those formed by a magnetic island chain. We place no restriction of either kind on \(\ST\). The fourth formulation is the one used in this paper, in which \(\tilde p_1, \dots, \tilde p_n\) are prescribed and the objective is the energy \(E\),
\begin{align}\label{prob:MRxMHD_fixed_append}
    \begin{gathered}
        \underset{\substack{f \in \I \\ \beta_i \in \mathcal{C}(\ST_i)}}{\text{c. p.}}
        \; \sum_{i=1}^n \left( \frac{1}{2} \big\| \beta_i \big\|_{L^2}^2
        - \tilde p_i \myabs{f(\ST_i)} \right) \, , \\
        H_{G_i}(\beta_i, \ST_i) = h_i \, , \qquad
        \int_{T_{i,j}} \beta_i = \psi_{i,j} \, .
    \end{gathered}
\end{align}
As in Section \ref{section:Main Proofs}, a critical point of a constrained problem is a configuration at which the first variation of the objective vanishes for every admissible variation tangent to the constraints.

Each of the four problems is called an MRxMHD problem in the literature, and we relate them in Proposition \ref{prop:energy_comparison} below. The only geometric input required is the first variation of the volume of a region, which we isolate first.

Let \(\{ f_t \}_{t \in (-\varepsilon, \varepsilon)} \subset \I\) be a smooth family with \(f_0 = f\), and let \(V = f^* W\) denote the associated variation field on \(\ST\), where \(W \coloneqq (\partial_t f_t) \circ f_t^{-1}|_{t=0}\) is the variation field on \(f(\ST)\) as in Subsection \ref{sec:main_theorem}, \(V\) vanishes on \(\partial \ST\). Then, for each \(i = 1, \dots, n\),
\begin{align} \label{lem:volume_variation}
    \frac{d}{dt} \bigg|_{t = 0} \myabs{ f_t(\ST_i) }
    = \frac{d}{dt} \bigg|_{t = 0} \int_{\ST_i} f_t^* \varpi
    = \int_{\partial \ST_i} \mathcal{J}^* \big( i_{V} ( f^* \varpi ) \big) \, . 
\end{align}

\begin{proposition}\label{prop:energy_comparison}
Let \(\gamma > 1\), fix a partition \((\ST_1, \dots, \ST_n)\) and prescribe the helicity data \(h_i\) and the flux data \(\psi_{i,j}\). Let \(f \in \I\) and \(\beta_i \in \mathcal{C}(\ST_i)\) satisfy these constraints, and let \(\tilde p_1, \dots, \tilde p_n\) be positive constants. Then the following are equivalent.
\begin{enumerate}[label={\arabic*. }  ]
    \item \((f, \beta_1, \dots, \beta_n)\) is a critical point of \eqref{prob:MRxMHD_fixed_append} with pressures \(\tilde p_1, \dots, \tilde p_n\).
    \item \((f, \beta_1, \dots, \beta_n, \tilde p_1, \dots, \tilde p_n)\) is a critical point of \eqref{prob:MRxMHD_entropy} with entropy data \(\mathcal{S}_i = \tilde p_i^{1/\gamma} \myabs{f(\ST_i)}\), the multiplier associated with the entropy constraint in \(\ST_i\) being 
    \[
    \nu_i = - \frac{\gamma}{\gamma - 1} \tilde p_i^{1 - 1/\gamma} \, . \]
    \item \((f, \beta_1, \dots, \beta_n, \tilde p_1, \dots, \tilde p_n)\) is a critical point of \eqref{prob:MRxMHD_adiabatic_app} with data \(v_i = \tilde p_i \myabs{f(\ST_i)}^\gamma\).
    \item For any choice of positive constants \(M_1, \dots, M_n\), the configuration \((f ,\allowbreak \beta_1, \allowbreak \dots, \allowbreak \beta_n, \allowbreak \tilde p_1, \allowbreak \dots, \allowbreak \tilde p_n, \allowbreak \tilde \rho_1, \allowbreak \dots, \tilde \rho_n)\) is a critical point of \eqref{prob:MRxMHD_thermo} with mass data \(M_i\),
    \begin{align*}
        \tilde \rho_i = M_i \myabs{f(\ST_i)}^{-1} \, ,
    \qquad \text{ and } \qquad
        \mathcal{E}_i = \frac{M_i}{\gamma - 1}
        \ln \! \left( \kappa \, \tilde p_i \, \myabs{f(\ST_i)}^{\gamma} M_i^{-\gamma} \right) \, ,
    \end{align*}
    the multipliers associated with the entropy and mass constraints in \(\ST_i\) being
    \begin{align*}
        \nu_i = - \frac{\tilde p_i \myabs{f(\ST_i)}}{M_i} \, , \qquad
        \lambda_i = - \frac{\nu_i}{\gamma - 1}
        \left( \ln \! \left( \kappa \, \tilde p_i \, \myabs{f(\ST_i)}^{\gamma} M_i^{-\gamma} \right)
        - \gamma \right) \, .
    \end{align*}
\end{enumerate}
Moreover, at any critical point of \eqref{prob:MRxMHD_entropy} the pressure \(\tilde p_i\) is constant on \(\ST_i\), and at any critical point of \eqref{prob:MRxMHD_thermo} both \(\tilde p_i\) and \(\tilde \rho_i\) are constant on \(\ST_i\), so that the restriction to constant pressures in \eqref{prob:MRxMHD_adiabatic_app} and \eqref{prob:MRxMHD_fixed_append} is no loss of generality.
\end{proposition}

\begin{proof}
In each of the four problems the norm of \(\beta_i\), the helicity, and flux constraints are identical, and the pressure terms are independent of \(\beta_1, \dots, \beta_n\). The variations of the objective with respect to \(\beta_i\) therefore produce the same Euler--Lagrange equations in all four problems, and it suffices to compare the first variation of the pressure terms with respect to \(f\), together with the variations with respect to \(\tilde p_i\) and \(\tilde \rho_i\) when the latter are independent variables. We show that in each problem the pressure terms contribute
\begin{align}\label{eq:pressure_variation}
    - \sum_{i = 1}^n \tilde p_i \, \int_{\partial \ST_i} \mathcal{J}^* \big( i_{V} ( f^* \varpi ) \big)
\end{align}
to the first variation along an admissible variation field \(V\) where \(\tilde p_i\) is constant. 

\emph{Problem \eqref{prob:MRxMHD_entropy}:} Since \(\tilde p_i > 0\), the map \(p \mapsto p^{1/\gamma}\) is smooth on the range of \(\tilde p_i\), and the derivative of the entropy constraint with respect to \(\tilde p_i\) is 
\[
\frac{1}{\gamma} \int_{\ST_i} \tilde p_i^{1/\gamma - 1} f^* \varpi > 0 \, .
\]
The Lagrange multiplier rule is therefore applicable, and there are constants \(\nu_1, \dots, \nu_n\) such that the first variation of
\begin{align*}
    W + \sum_{i = 1}^n \nu_i \left( \int_{\ST_i} \tilde p_i^{1/\gamma} \, f^* \varpi
    - \mathcal{S}_i \right)
\end{align*}
vanishes on the variations tangent to the helicity and flux constraints. The multiplier \(\nu_i\) is a constant on \(\ST_i\) because the entropy constraint is a single scalar condition in each region.

Let \(q \in C^\infty(\ST_i)\) and vary \(\tilde p_i\) in the direction \(q\), with \(f\) and \(\beta_1, \dots, \beta_n\) held fixed. Neither the helicity constraint nor the flux constraint involves the pressure, so
\begin{align*}
    \frac{d}{dt} \bigg|_{t = 0}
    \int_{\ST_i} \left( \frac{\tilde p_i + t q}{\gamma - 1} + \nu_i (\tilde p_i + t q)^{1/\gamma} \right) f^* \varpi
    = \int_{\ST_i} \left( \frac{1}{\gamma - 1}
    + \frac{\nu_i}{\gamma} \tilde p_i^{1/\gamma - 1} \right) q \, f^* \varpi \, .
\end{align*}
As \(q\) is arbitrary and \(f^* \varpi\) is a positive volume form, the integrand vanishes pointwise on \(\ST_i\), giving
\begin{align*}
    \frac{1}{\gamma - 1} + \frac{\nu_i}{\gamma} \tilde p_i^{1/\gamma - 1} = 0 \, .
\end{align*}
In particular \(\nu_i \neq 0\), since \(\gamma \neq 1\). Rearranging gives
\begin{align}\label{eq:pressure_from_nu}
    \tilde p_i = \left( - \frac{\gamma}{\nu_i (\gamma - 1)} \right)^{\gamma / (1 - \gamma)}
    \qquad \text{equivalently} \qquad
    \nu_i = - \frac{\gamma}{\gamma - 1} \tilde p_i^{1 - 1/\gamma} \, .
\end{align}
The right-hand side of the first identity in \eqref{eq:pressure_from_nu} is a constant, so \(\tilde p_i\) is constant on \(\ST_i\). Moreover, for \(\gamma > 1\) the map \(\nu \mapsto ( - \gamma / (\nu (\gamma - 1)) )^{\gamma / (1 - \gamma)}\) is a decreasing bijection from \((-\infty, 0)\) onto \((0, \infty)\). Hence, every positive constant pressure arises from exactly one multiplier \(\nu_i\), and conversely.

Now vary \(f\), holding \(\tilde p_1, \dots, \tilde p_n\) and \(\beta_1, \dots, \beta_n\) fixed. Since \(\tilde p_i\) is constant, the pressure terms in \(\ST_i\) reduce to
\begin{align*}
    \int_{\ST_i} \left( \frac{\tilde p_i}{\gamma - 1} + \nu_i \tilde p_i^{1/\gamma} \right) f^* \varpi
    = \left( \frac{\tilde p_i}{\gamma - 1} + \nu_i \tilde p_i^{1/\gamma} \right) \myabs{f(\ST_i)} \, ,
\end{align*}
and equation \eqref{lem:volume_variation} shows that their first variation along \(V\) is
\begin{align*}
    \left( \frac{\tilde p_i}{\gamma - 1} + \nu_i \tilde p_i^{1/\gamma} \right) \int_{\partial \ST_i} \mathcal{J}^* \big( i_{V} ( f^* \varpi ) \big) \, .
\end{align*}
Substituting \(\nu_i\) from \eqref{eq:pressure_from_nu} gives that the coefficient above equals
\begin{align*}
    \frac{\tilde p_i}{\gamma - 1} - \frac{\gamma \, \tilde p_i}{\gamma - 1} = - \tilde p_i \, ,
\end{align*}
which gives \eqref{eq:pressure_variation} for the interface variation.

\emph{Problem \eqref{prob:MRxMHD_adiabatic_app}:} Here the pressures are constants that are determined by the constraint, namely \(\tilde p_i = v_i \myabs{f(\ST_i)}^{-\gamma}\), so no multiplier is required, and no positivity assumption is needed. Substituting this expression into the objective, the pressure terms in \(\ST_i\) become
\begin{align*}
    \int_{\ST_i} \frac{\tilde p_i}{\gamma - 1} \, f^* \varpi
    = \frac{\tilde p_i \myabs{f(\ST_i)}}{\gamma - 1}
    = \frac{v_i}{\gamma - 1} \myabs{f(\ST_i)}^{1 - \gamma} \, ,
\end{align*}
and the only remaining dependence on the unknowns is through \(\myabs{f(\ST_i)}\). Varying \(f\) and applying equation \eqref{lem:volume_variation} gives
\begin{align*}
    \frac{v_i (1 - \gamma)}{\gamma - 1} \myabs{f(\ST_i)}^{-\gamma} \int_{\partial \ST_i} \mathcal{J}^* \big( i_{V} ( f^* \varpi ) \big)
    & = - v_i \myabs{f(\ST_i)}^{-\gamma} \int_{\partial \ST_i} \mathcal{J}^* \big( i_{V} ( f^* \varpi ) \big)
    \\ &= - \tilde p_i \, \int_{\partial \ST_i} \mathcal{J}^* \big( i_{V} ( f^* \varpi ) \big) \, ,
\end{align*}
which is again \eqref{eq:pressure_variation}.

\emph{Problem \eqref{prob:MRxMHD_thermo}:} Write \(s_i \coloneqq \ln ( \kappa \tilde p_i \tilde \rho_i^{-\gamma} )\), which is well-defined since \(\tilde p_i, \tilde \rho_i, \kappa > 0\). The derivative of the entropy constraint with respect to \(\tilde p_i\) equals 
\[
    (\gamma - 1)^{-1} \int_{\ST_i} \tilde \rho_i \tilde p_i^{-1} f^* \varpi > 0\, ,
\]
 and the mass constraint has vanishing derivative with respect to \(\tilde p_i\). The derivative of the mass constraint with respect to \(\tilde \rho_i\) equals \(\myabs{f(\ST_i)} > 0\). The derivative of the pair of constraints is therefore surjective onto \(\R^2\) in each region, so the Lagrange multiplier rule applies, and there are constants \(\nu_i\) and \(\lambda_i\) such that the first variation of
\begin{align*}
    W + \sum_{i = 1}^n \nu_i \left( \int_{\ST_i} \frac{\tilde \rho_i}{\gamma - 1} s_i \, f^* \varpi
    - \mathcal{E}_i \right)
    + \sum_{i = 1}^n \lambda_i \left( \int_{\ST_i} \tilde \rho_i \, f^* \varpi - M_i \right)
\end{align*}
vanishes on the variations tangent to the helicity and flux constraints. As before \(\nu_i\) and \(\lambda_i\) are constants on \(\ST_i\), each constraint being a single scalar condition in each region.

Varying \(\tilde p_i\) in an arbitrary direction \(q \in C^\infty(\ST_i)\), and using \(\partial s_i / \partial \tilde p_i = \tilde p_i^{-1}\), the integrand of the first variation is 
\[
\Big( \frac{1}{\gamma - 1} + \frac{\nu_i}{\gamma - 1} \tilde \rho_i \tilde p_i^{-1} \Big) q \, , \]
so that
\begin{align}\label{eq:thermo_p}
    \nu_i \tilde \rho_i = - \tilde p_i \qquad \text{on } \ST_i \, .
\end{align}
In particular \(\nu_i \neq 0\), and \(\tilde \rho_i / \tilde p_i\) is constant on \(\ST_i\). Varying \(\tilde \rho_i\) in an arbitrary direction \(r \in C^\infty(\ST_i)\), and using \(\partial s_i / \partial \tilde \rho_i = - \gamma \tilde \rho_i^{-1}\), gives
\begin{align}\label{eq:thermo_rho}
    \frac{\nu_i}{\gamma - 1} \big( s_i - \gamma \big) + \lambda_i = 0
    \qquad \text{on } \ST_i \, .
\end{align}
Since \(\nu_i\) and \(\lambda_i\) are constants and \(\nu_i \neq 0\), the quantity \(s_i\), and hence \(\kappa \tilde p_i \tilde \rho_i^{-\gamma}\), is constant on \(\ST_i\). Combining this with the constancy of \(\tilde \rho_i / \tilde p_i\) obtained from \eqref{eq:thermo_p}, we may write \(\tilde \rho_i = c_i \tilde p_i\) for a constant \(c_i > 0\). Hence, \(\kappa c_i^{-\gamma} \tilde p_i^{1 - \gamma}\) is constant. As \(\gamma \neq 1\) this forces \(\tilde p_i\), and therefore \(\tilde \rho_i\), to be constant on \(\ST_i\). The mass constraint then gives \(\tilde \rho_i = M_i \myabs{f(\ST_i)}^{-1}\), and the stated expressions for \(\mathcal{E}_i\), \(\nu_i\) and \(\lambda_i\) follow from the entropy constraint, \eqref{eq:thermo_p} and \eqref{eq:thermo_rho} respectively.

Now vary \(f\), holding \(\tilde p_i\) and \(\tilde \rho_i\) fixed. Since both are constant, the pressure terms in \(\ST_i\) reduce to
\begin{align*}
    \left( \frac{\tilde p_i}{\gamma - 1} + \frac{\nu_i \tilde \rho_i}{\gamma - 1} s_i + \lambda_i \tilde \rho_i \right) \myabs{f(\ST_i)} \, ,
\end{align*}
and equation \eqref{lem:volume_variation} shows that their first variation along \(V\) is the bracketed coefficient times the volume derivative. Eliminating \(\lambda_i\) by \eqref{eq:thermo_rho} and then \(\nu_i \tilde \rho_i\) by \eqref{eq:thermo_p}, the coefficient equals
\begin{align*}
    \frac{\tilde p_i}{\gamma - 1} + \frac{\nu_i \tilde \rho_i}{\gamma - 1} s_i
    - \frac{\nu_i \tilde \rho_i}{\gamma - 1} \big( s_i - \gamma \big)
    = \frac{\tilde p_i}{\gamma - 1} + \frac{\gamma \, \nu_i \tilde \rho_i}{\gamma - 1}
    = \frac{\tilde p_i - \gamma \, \tilde p_i}{\gamma - 1}
    = - \tilde p_i \, ,
\end{align*}
which is once more \eqref{eq:pressure_variation}. The mass data \(M_i\) enters the argument only through \(\tilde \rho_i = M_i \myabs{f(\ST_i)}^{-1}\).

\emph{Problem \eqref{prob:MRxMHD_fixed_append}:} The pressure term in \(\ST_i\) is \(- \tilde p_i \myabs{f(\ST_i)}\) with \(\tilde p_i\) a prescribed constant, so the variation of equation \eqref{lem:volume_variation} gives \eqref{eq:pressure_variation} immediately. This is the computation carried out in the proof of Proposition \ref{thr:iff_simple} in Section \ref{section:Main Proofs}.

In all four problems the first variation of the objective therefore consists of the same magnetic contribution together with \eqref{eq:pressure_variation}, and is taken subject to the same helicity and flux constraints. The four problems thus have the same Euler--Lagrange equations, and the stated correspondence between the data \((\tilde p_i)\), \((\mathcal{S}_i)\), \((v_i)\) and \((M_i, \mathcal{E}_i)\) follows from \eqref{eq:pressure_from_nu}, from \(\tilde p_i = v_i \myabs{f(\ST_i)}^{-\gamma}\), and from \eqref{eq:thermo_p}--\eqref{eq:thermo_rho}.
\end{proof}

\begin{remark}
The correspondence between the constraint data in Proposition \ref{prop:energy_comparison} depends on the critical configuration itself, through the volumes \(\myabs{f(\ST_i)}\) and the same is true of the multipliers \(\nu_i\) and \(\lambda_i\). The four problems are therefore not interchangeable for a fixed choice of data. What the proposition asserts is that the four families of problems, taken over all admissible data, have the same critical configurations, and that passing from one family to another amounts to a relabelling of the data at each such configuration.
\end{remark}

\section{Alexander Basis Exposition}\label{section:Alexander Basis}

This section provides some background on the construction of an Alexander basis. We present a general formulation in Subsection \ref{append:general_formulation}, and provide  an explicit example for a hollow torus in Subsection \ref{append:example}. The material here is drawn from work by Cantarella and Parsley (\citeyear{CANTARELLA20101127}, appendix B).

Consider a basis \(\{[s_{i,1}], \dots, [s_{i,{\ell_i}}]\}\) of \(H_1 (\ST_i) \) where \(\ST_i\) is a compact, oriented, connected domain in \(\R^3\) with a smooth boundary \(\partial \ST_i\). The genus of \(\partial \ST_i\) is \(0 <{\ell_i} < \infty\). We keep the notation consistent with the remainder of the document, this makes the partition subscript in \(\ST_i\) and other variables superfluous. 

The goal in Subsection \ref{append:general_formulation} is to prove Theorem \ref{thr:Alexander}. Additionally, we consider orientation preserving diffeomorphisms \(\varphi:\widehat{\ST}_i\to\ST_i\) that extend to orientation preserving ambient diffeomorphisms of \(\mathbb{S}^3\). The maps \(f\in\I\) in the body of the paper admit such extensions by isotopy extension. We show that pulling back the elements of an Alexander basis for \(H_1(\partial \ST_i)\) gives an Alexander basis for \(H_1(\partial \widehat{\ST}_i)\). This is used throughout our paper in order to have a consistent definition of flux, loop-corrected helicity, and Amperian gauge on domains \(f(\ST_i)\) with \(f \in \mathcal{I}( \ST )\).

First we provide and prove well known results.
\begin{lemma}[{\citealt[Proposition C.4]{CANTARELLA20101127}}]\label{Lemma:Potential}
    Any closed, Dirichlet \(2\)-form on \(\ST_i\) is exact.
\end{lemma}

\begin{proof}
    Let \(b \in \Omega^2_D( \ST_i )\) be a closed, Dirichlet differential \(2\)-form. Applying a HMF decomposition to the 1-form \(\star b\), then applying a \(\star\) we obtain
    \begin{align*}
        b = d \eta_1 + \delta \eta_2 + d \eta_3 + \eta_4^* \, , 
    \end{align*}
    where \(\eta_1 \in \Omega^1 _D (\ST_i)\), \(\eta_2 \in  \Omega^3_N (\ST_i)\), \( d \eta_3 \in \mathcal{H}^2 (\ST_i) \), where \(\eta_3 \in \Omega^1  (\ST_i)\), and \(\eta_4^* \in \mathcal{H}_N^2 (\ST_i) \). Following a similar approach as in Section \ref{sub:Hodge}, we apply the closed, and Dirichlet conditions:
    \begin{align*}
        0 = db = d   \delta \eta_2  \, , \qquad 0 = tb =  t (  \delta \eta_2 + d \eta_3 + \eta_4^* ) \, .
    \end{align*}
    Taking note of the first equality, and using Stokes' theorem:
    \begin{align*}
        \|  \delta \eta_2 \|_{L^2}^2 = -\int_{\ST_i}  \delta \eta_2 \wedge  d \star \eta_2 = \int_{\ST_i} d \delta \eta_2 \wedge \star \eta_2 -\int_{\partial \ST_i} \mathcal{J}^* (\delta \eta_2 \wedge \star \eta_2) = 0 \, . 
    \end{align*}
    And by continuity we have that \(\delta \eta_2 = 0\).

    The Dirichlet condition is now \(0=t(d\eta_3+\eta_4^*)\). The Mayer--Vietoris sequence for \(\mathbb{S}^3=\ST_i\cup\overline{\ST}_i^c\), using collar thickenings, shows that \(\mathcal{J}_*:H_2(\partial\ST_i)\to H_2(\ST_i)\) is surjective because \(H_2(\mathbb{S}^3)=0\). Thus every class in \(H_2(\ST_i)\) is represented by \(\mathcal{J}_*z\) for a boundary cycle \(z\), and
    \[
        \int_{\mathcal{J}_*z}\eta_4^*
        =\int_z\mathcal{J}^*\eta_4^*
        =-\int_zd(\mathcal{J}^*\eta_3)=0.
    \]
    By de Rham's theorem, \([\eta_4^*]=0\). Uniqueness of the Neumann harmonic representative then gives \(\eta_4^*=0\), so \(b=d(\eta_1+\eta_3)\) is exact.
\end{proof}

\subsection{General Formulation}\label{append:general_formulation}

Given any choice of bases for \(H_1(\ST_i)\) and \(H_1 (\overline{\ST}^c_i )\) one can, of course, define a basis for \(H_1 (\partial \ST_i)\) (Lemma \ref{Lemma:form_basis}). Theorem \ref{thr:Alexander} asserts that we already have a basis for \(H_1(\ST_i)\) given by \(\{ [s_{i,1}], \dots, [s_{i,{\ell_i}}] \}\), all that remains is to construct a suitable basis for \(H_1(\overline{\ST_i}^c)\) that satisfies the requirements of an Alexander basis. We proceed in two steps: First, by constructing \(\{[t_{i,1}], \dots, [t_{i,\ell_i}]\}\) (a basis for \(H_1 (\overline{\ST}^c_i )\)). Secondly, by verifying this basis satisfies the conditions in Theorem \ref{thr:Alexander}.

\bigskip 

\textit{Step 1. Construct a Basis for \(H_1(\partial \ST_i)\) via \(\{[s_{i,1}], \dots, [s_{i,{\ell_i}}] \}\).}

As a basis for \(H_1(\ST_i)\) is given we construct basis elements \(T_{i,j}, \sigma_{i,j}\) and the duals \([s_{i,j}^*]\), \([T_{i,j}^*]\), and \([\sigma_{i,j}^*]\) as in the main body. We will see an explicit example to compute these in the next subsection. 

We want to see how \(H_2(\ST_i,\partial\ST_i)\) relates to \(H_1(\overline{\ST_i}^c)\). Let \((\ST_i)_\epsilon\) be a small open collar thickening of \(\ST_i\) in \(\mathbb{S}^3\), and let \(\overline{\ST_i}^c_\epsilon\) be its closed complement, contained in the interior of \(\overline{\ST_i}^c\). Excision, followed by a deformation retraction of pairs, shows that the inclusion \(\mathcal{J}^0:(\ST_i,\partial\ST_i)\hookrightarrow(\mathbb{S}^3,\overline{\ST_i}^c)\) induces an isomorphism
\[
    H_k(\ST_i,\partial\ST_i)\cong H_k(\mathbb{S}^3,\overline{\ST_i}^c)
\]
\citep[see][page 119, Theorem 2.20]{hatcher2005algebraic}.

The relevant part of the long exact sequence is
\[
    0=H_2(\mathbb{S}^3)\longrightarrow H_2(\mathbb{S}^3,\overline{\ST_i}^c)
    \xrightarrow{\partial}H_1(\overline{\ST_i}^c)\longrightarrow H_1(\mathbb{S}^3)=0.
\]
This gives isomorphisms \(H_2(\ST_i,\partial\ST_i)\cong H_2(\mathbb{S}^3,\overline{\ST_i}^c)\cong H_1(\overline{\ST_i}^c)\), induced by \(\mathcal{J}^0_*\) and \(\partial\). Interchanging the domains similarly gives \(H_2(\overline{\ST_i}^c,\partial\overline{\ST_i}^c)\cong H_1(\ST_i)\). We define a basis for \(H_1(\overline{\ST_i}^c)\) by \([t_{i,j}]\coloneqq\partial\mathcal{J}^0_*[T_{i,j}]\).

From the \([t_{i,j}]\) we define \([t_{i,j}^*], [S_{i,j}], [S_{i,j}^*], [\tau_{i,j}], [\tau_{i,j}^*]\) as in the main body. We also see that the basis for \(H_{1} (\overline{\ST_i}^c )\) has dimension \(\ell_i\).

\bigskip 

\textit{Step 2. Verifying that the basis constructed in step 1 is an Alexander basis for \(H_1(\partial \ST_i)\).}

We first show that \(\partial [T_{i,j}] = [\tau_{i,j}]\).

Consider inclusion maps \(\mathcal{J}^0 : (\ST_i, \partial \ST_i) \xhookrightarrow{} (\mathbb{S}^3 , \overline{\ST}_i^c)\), \(\mathcal{J} :  \partial \ST_i \xhookrightarrow{}  \ST_i\) and \( \mathcal{J}^c : \partial \ST_i \xhookrightarrow{} \overline{ \ST }_i^c \). Recalling that \((\mathcal{J}_*, \mathcal{J}^c_*) : H_1(\partial \ST_i) \to H_1(\ST_i) \oplus H_1(\overline{\ST}^c_i)\) is an isomorphism (Lemma \ref{Lemma:isomorphism_inclusion}), for any \([y] \in H_2(\ST_i, \partial \ST_i)\):
\begin{align*}
    (\mathcal{J}_*, \mathcal{J}^c_*) \partial [y] = ( \mathcal{J}_*  \partial [y], \mathcal{J}^c_*  \partial [y]) 
    =  ( \mathcal{J}_*  \partial [y],   \partial \mathcal{J}^0_* [y])  \, . 
\end{align*}
Here the boundary acts as \(\partial : H_2(\ST_i, \partial \ST_i) \to H_{1}(\partial \ST_i)\) except for the last boundary operator, acting on spaces \(\partial : H_2(\mathbb{S}^3 , \overline{\ST}_i^c) \to H_{1}(\overline{\ST}_i^c)\). Recall a portion of the long exact sequence of the pair \((\ST_i,\partial\ST_i)\):
\[
\cdots \longrightarrow H_2(\ST_i) \longrightarrow H_2(\ST_i,\partial\ST_i)
\xrightarrow{\; \partial \;} H_1(\partial\ST_i)
\xrightarrow{\;\mathcal{J}_*\;} H_1(\ST_i)\longrightarrow\cdots .
\]
Exactness implies \(\operatorname{im}(\partial ) = \ker( \mathcal{J}_* )\). Hence \(\mathcal{J}_*(\partial [y])=0\),
showing that \((\mathcal{J}_*, \mathcal{J}^c_*)  \partial [y]  =  ( 0,   \partial \mathcal{J}^0_* [y])\). As \((\mathcal{J}_*, \mathcal{J}^c_*) \) is an isomorphism, it has an inverse, so, setting \([y] = [T_{i,j}]\) for any \(i,j\),
\begin{align*}
    \partial [T_{i,j}]  = (\mathcal{J}_*, \mathcal{J}^c_*)^{-1} ( 0, [t_{i,j}]) = [\tau_{i,j}] \, ,
\end{align*}
as required. 

Theorem \ref{thr:Alexander} claims that for all \( j, k \in \{ 1, \dots, {\ell_i}\} \):
\begin{align*}
    \begin{gathered}
        \int_{\partial \ST_i} \sigma_{i,j}^* \wedge \sigma_{i,k}^* = 
         0 \,,    \qquad 
         \int_{\partial \ST_i} \tau_{i,j}^* \wedge \tau_{i,k}^*
        = 0 \, ,
        \qquad 
         \int_{\partial \ST_i} \tau_{i,j}^* \wedge \sigma_{i,k}^* = \delta_{jk}  \,.
    \end{gathered}
\end{align*}  
The first 2 integrals we can compute directly, namely, via Stokes' theorem we have that 
\begin{align*}
    \int_{\partial \ST_i} \sigma_{i,j}^* \wedge \sigma_{i,k}^* = \int_{ \ST_i} d( s_{i,j}^* \wedge s_{i,k}^* ) = \int_{ \ST_i} 0  = 0 \, .
\end{align*}
Similarly,
\begin{align*}
    \int_{\partial \ST_i} \tau_{i,j}^* \wedge \tau_{i,k}^* = - \int_{ \overline{\ST}^c_i} d( t_{i,j}^* \wedge t_{i,k}^* ) = - \int_{ \overline{\ST}^c_i } 0  = 0 \, .
\end{align*}
For the last integral, we note that \([T_{i,j}^*] \in H^2_{\mathrm{dR}} (\ST_i, \partial \ST_i) \). By the Hodge theorem, we have a unique harmonic representative \(T_{i,j}^* \in \mathcal{H}^2_D(\ST_i)\). Importantly, \(T_{i,j}^*\) is a closed, Dirichlet 2-form on \(\ST_i\), so by Lemma \ref{Lemma:Potential}, \(T_{i,j}^* = d \Gamma_{i,j} \) for some \(\Gamma_{i,j} \in \Omega^1(\ST_i)\). We know that the collection of \([\sigma_{i,j}^*], [\tau_{i,j}^*]\) form a basis for \(H^1_{\mathrm{dR}}(\partial \ST_i)\) (Lemma \ref{Lemma:form_basis}), so given \([\mathcal{J}^* \Gamma_{i,j} ] \in H^1_{\mathrm{dR}} (\partial \ST_i)\), 
\begin{align*}
    [\mathcal{J}^* \Gamma_{i,j} ] = \sum_{r=1}^{\ell_i} \Big( k_{j,r} [\tau_{i,r}^*] +   c_{j,r} [\sigma_{i,r}^*] \Big) \, .
\end{align*}
By definition of \([T_{i,j}^*]\),
\begin{align*}
    \delta_{jk}  = \int_{T_{i,j}} T_{i,k}^*
    =   \int_{T_{i,j}} d \Gamma_{i,k} 
    =    \int_{\tau_{i,j}} \mathcal{J}^*  \Gamma_{i,k} 
    =
    \sum_{r=1}^{\ell_i} \left( k_{k,r} \int_{\tau_{i,j}}  \tau_{i,r}^* +  c_{k,r}  \int_{\tau_{i,j}} \sigma_{i,r}^* \right) = k_{k,j}
    \, .
\end{align*}
We have explicitly used \(\partial [T_{i,j}] = [\tau_{i,j}]\) within this computation. Recall that Section \ref{section:Problem Construction} shows \(\tau_{i,j}^*\), \(\sigma_{i,j}^*\) are the duals of \(\tau_{i,j}\), \(\sigma_{i,j}\) respectively. Substituting this back into the non-degenerate Lefschetz pairing
\begin{align*}
    \delta_{jk} &= \int_{ \ST_i} s_{i,j}^* \wedge d \Gamma_{i,k} 
    = - \int_{\partial \ST_i} \sigma_{i,j}^* \wedge \mathcal{J}^*\Gamma_{i,k}
    = - \int_{\partial \ST_i} \sigma_{i,j}^* \wedge \Big( \tau_{i,k}^* + \sum_{r=1}^{\ell_i} c_{k,r} \sigma_{i,r}^* \Big)
    \\ &=  \int_{\partial \ST_i}  \tau_{i,k}^*  \wedge \sigma_{i,j}^*
    \, .
\end{align*}
This completes the proof for the integral relations. But it remains to be shown that \(\partial [S_{i,j}] = [\sigma_{i,j}]\).

Now for the last claim, that \([\sigma_{i,j}] = \partial [S_{i,j}]\). We just showed \(\mathcal{J}_*(\partial [y])=0\) for all \([y] \in H_2( \ST_i, \partial \ST_i) \). By exactly the same argument for all \([z] \in H_2 ( \overline{\ST}_i^c ,  \partial \overline{\ST}_i^c )\) we have \(\mathcal{J}^c_*(\partial [z])=0\). Therefore 
\[\partial [S_{i,j}] = \sum_{k=1}^{\ell_i} c_{i,j,k} [\sigma_{i,k}] \, . \]
Select some \(\eta \in \Omega^1(\overline{\ST}^c_i)\) such that \(\eta |_{\partial \ST_i} = \sigma_{i,k}^*\). On the boundary we see that \(d\eta |_{\partial \ST_i} = d \sigma_{i,k}^* = 0\), so \([d\eta] \in H^2_{\mathrm{dR}}(\overline{\ST}^c_i, \partial \overline{\ST}^c_i ) \). Therefore we have, using Lefschetz duality and Stokes' theorem 
\begin{align*}
    \int_{\partial S_{i,j}} \sigma_{i,k}^* =  \int_{S_{i,j}}  d\eta = \int_{\overline{\ST}^c_i} t_{i,j}^* \wedge d\eta  = - \int_{\overline{\ST}^c_i} d (t_{i,j}^* \wedge \eta)  
    =    \int_{\partial \ST_i} \tau_{i,j}^* \wedge \sigma_{i,k}^* = \delta_{jk} \, . 
\end{align*}
By selecting an outward orientation for \(\partial \ST_i\) Stokes' theorem provides a negative sign when moving to the boundary of \(\overline{\ST}^c_i\). This gives the relation when placing in basis elements
\begin{align*}
    \delta_{jk} = \int_{\partial S_{i,j} }  \sigma_{i,k}^*
    = \int_{ \sum_{r=1}^{\ell_i}  c_{i,j,r} \sigma_{i,r} }  \sigma_{i,k}^*
    = \sum_{r=1}^{\ell_i}  c_{i,j,r}  \int_{ \sigma_{i,r} }  \sigma_{i,k}^*
    =   c_{i,j,k} \, .
\end{align*}
Hence \( \partial [S_{i,j}] =  [\sigma_{i,j}]  \). This completes the proof of Theorem \ref{thr:Alexander}. \hfill $\square$

Let \(\varphi:\widehat{\ST}_i \to \ST_i\) be an orientation preserving diffeomorphism with an orientation preserving ambient extension. This extension also maps the complements diffeomorphically and preserves both kernels of the boundary inclusion maps. Since \(\varphi\) is a diffeomorphism the induced maps
\[
\varphi_*:H_k(\widehat{\ST}_i) \to H_k(\ST_i), \qquad
\varphi^*:H^k_{\mathrm{dR}}(\ST_i) \to H^k_{\mathrm{dR}}(\widehat{\ST}_i)
\]
are isomorphisms (with inverses \(\varphi^{-1}_*\) and \((\varphi^{-1})^*\) respectively). Let \(\varphi_\partial:=\varphi\big|_{\partial\widehat{\ST}_i}:\partial\widehat{\ST}_i\to\partial\ST_i\) be the corresponding map on the boundaries. 

Let \(\mathcal{J}:\partial \ST_i \hookrightarrow \ST_i\) denote the boundary inclusion. Recall that an Alexander basis
\[\{[\sigma_{i,1}],\dots,[\sigma_{i,\ell_i}],[\tau_{i,1}],\dots,[\tau_{i,\ell_i}]\}\subset H_1(\partial\ST_i)\] 
has the property that the inclusion-induced map \(\mathcal{J}_*:H_1(\partial\ST_i)\to H_1(\ST_i)\) sends \(\,[\sigma_{i,j}]\mapsto[s_{i,j}]\) and \(\mathcal{J}^c :\partial \ST_i \hookrightarrow \overline{\ST}^c_i\) induces a map \(\,[\tau_{i,j}]\mapsto[t_{i,j}]\).

Define \([\widehat{s}_{i,j}] \coloneqq \varphi^{-1}_* [s_{i,j}]\), then \(\{[\widehat{s}_{i,1}],\dots,[\widehat{s}_{i,\ell_i}]\}\) is a basis of \(H_1(\widehat{\ST}_i)\). Define \( [\widehat \sigma_{i, j}] =  (\varphi_\partial^{-1})_* [\sigma_{i, j}]\) on the boundary. Denote the resulting basis \(\{[\widehat{\sigma}_{i,1}],\dots,[\widehat{\sigma}_{i,\ell_i}], [\widehat{\tau}_{i,1}],\dots, \allowbreak [\widehat{\tau}_{i,\ell_i}]\}\). Similarly, define \(\widehat{S}_{i,j}, \widehat{T}_{i,j} \) and the corresponding duals in the same way as their non-hatted counterparts. So, for \(\varphi_* \coloneqq (\varphi^{-1})^*\) on forms, \([s^*_{i,j}] = \varphi_* [\widehat s^*_{i,j}]\) and
\begin{align*}
    \int_{T_{i, j}} \! \varphi_*( \widehat T_{i, k}^* )  = \int_{\ST_i}  \! s_{i, j}^* \wedge \varphi_*( \widehat T_{i, k}^* ) = \int_{\ST_i}  \! \varphi_*( \widehat s_{i, j}^* \wedge  \widehat T_{i, k}^* )
    =
    \int_{\widehat \ST_i}  \! \widehat s_{i, j}^* \wedge \widehat T_{i, k}^* = \delta_{jk}
    =
    \int_{T_{i, j}}  \! T_{i, k}^* 
    \, .
\end{align*}
This implies that \(\varphi_*([ \widehat T_{i, j}^*] ) = [T_{i, j}^*]\), so \([ \widehat T_{i, j}^*]  = \varphi^* [T_{i, j}^*]\), and on the boundary this induces \( [ \widehat \tau_{i,j}^*]  =  \varphi_\partial^* [\tau_{i,j}^*]\). Hence, in the current context, the Alexander basis is also pulled back by \(\varphi\). Namely, \[
\{  (\varphi_\partial^{-1})_* [\sigma_{i, 1}] , \dots, (\varphi_\partial^{-1})_* [\sigma_{i, \ell_i}]  ,    (\varphi_\partial^{-1})_* [\tau_{i, 1}] , \dots, (\varphi_\partial^{-1})_* [\tau_{i, \ell_i}]\}\] 
is an Alexander basis for \(H_1(\partial \widehat \ST_i) \). This is helpful when considering an MRxMHD problem where interfaces are pushed forward by \(f \in \I\) because the ambient extension also transports the corresponding Alexander basis and related definitions.

The construction also applies directly to knotted domains, as in the following example, but an arbitrary parametrisation must use longitudes compatible with the complement.


\begin{figure}
    \centering
    \begin{overpic}[width=0.5\linewidth]{"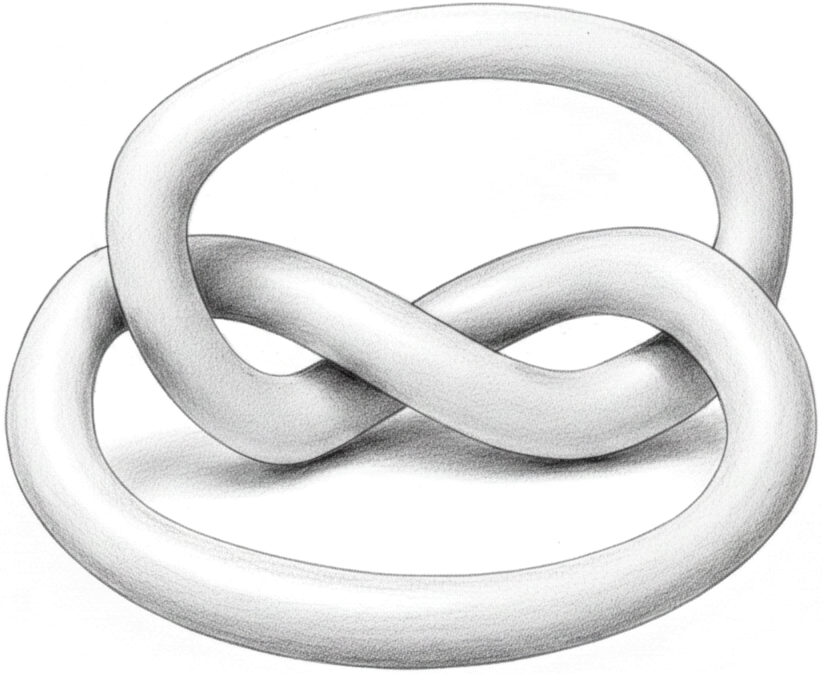"}
    \end{overpic}
    \caption{A hollow torus embedded in \(\R^3\) with a trefoil knot.}
    \label{fig:trefoil} 
\end{figure}

\subsection{Hollow Torus Example} \label{append:example}

We apply our construction of an Alexander basis to a hollow torus \(\HT\subset\R^3\). Rather than contract \(\HT\) to a torus, we work with an explicit orientation preserving diffeomorphism
\[
\iota:[1,2]\times \mathbb{T}^2 = U \hookrightarrow \HT\subset\R^3 \, ,
\]
where \(\mathbb{T} \coloneqq \R / \mathbb{Z}\) and the coordinate system \((r,\vartheta_{P},\vartheta_{T})\) with \(r\in[1,2]\) and \(\vartheta_{P},\vartheta_{T}\in[0,1)\). The subscript \(i\) is retained for consistency with the rest of the paper. In the language of the body of the paper the parameter space \(U\) is the reference region \(\ST_i\) and \(\iota\) provides the coordinate map, so that the forms \(s_{i,j}^*\), \(T_{i,j}^*\) and the chains constructed below are objects on \(\ST_i\), while the magnetic field lives on the physical region \(\iota(U) = \HT \subseteq \R^3\). On the parameter space define singular \(1\)-cycles
\[
s_{i,1}(\omega )=(2,0, \omega ),\qquad s_{i,2}(\omega )=(1,\omega ,0),\qquad \omega \in\Delta_1=[0,1] \, .
\]
Clearly \(\partial s_{i,j}=0\), so \([s_{i,1}]\) and \([s_{i,2}]\) represent classes in \(H_1( U  )\). Since this homology is two-dimensional (the boundary has total genus \(2\)), it suffices to verify linear independence. This follows either from the K\"unneth theorem or from the geometric observation that no singular \(2\)-chain relates \(s_{i,1}\) to \(s_{i,2}\), hence \(\{[s_{i,1}],[s_{i,2}]\}\) is a basis of \(H_1(  U  )\). 

Choose \(\vartheta_{P}\) as the meridional angle and \(\vartheta_{T}\) with zero framing, so that the outer toroidal loop is a preferred longitude, null homologous in the exterior. The inner poloidal loop bounds a meridional disk in the removed inner solid torus. Thus both \(s_{i,1}\) and \(s_{i,2}\), viewed on the boundary, are null homologous in the complement and give \([\sigma_{i,1}]\) and \([\sigma_{i,2}]\).

With this choice of framing, the following construction on the parameter space gives an Alexander basis for \(H_1(\partial\HT)\). In the unknotted case these loops agree with those used for loop-corrected helicity by Pfefferl\'{e}, Noakes, and Perrella \citeyearpar{Gauge_freedom}. The domain \(\HT\) may also be knotted, as in Figure \ref{fig:trefoil}; a Seifert surface supplies the preferred longitude \citep[see][]{seifert1935geschlecht}.

Our computation of \([T_{i, j}]\) relies on cohomology elements \([s_{i,j}^*]\). We may compute a representative closed 1-form \(s_{i,j}^* \in \Omega^1 (\ST_i) \) through the relation
\begin{align*}
    \delta_{jk}  = \int_{s_{i,j}} s_{i,k}^* \, . 
\end{align*}
Selecting \(s_{i,1}^* = d \vartheta_{T} \) and \(s_{i,2}^* = d \vartheta_{P} \) we verify that 
\begin{align*}
    \int_{s_{i,1}} s_{i,1}^* &= \int_{\omega \mapsto  (2, 0,  \omega )} d \vartheta_{T}  = \int_{0}^1 d  \omega  = 1 \, , \qquad
    \int_{s_{i,1}} s_{i,2}^* = \int_{\omega \mapsto  (2, 0,  \omega )} d \vartheta_{P}  =  0 \, , \\
    \int_{s_{i,2}} s_{i,2}^* &= \int_{\omega \mapsto  (1,  \omega , 0)} d \vartheta_{P}  = 1 \, , \qquad \hspace{4em}
    \int_{s_{i,2}} s_{i,1}^* = \int_{\omega \mapsto  (1,  \omega , 0)} d \vartheta_{T}  =0 \, .
\end{align*}
There is some freedom, up to a choice of volume in the parameter space, which we set by selecting a volume form \(\varpi = dr \wedge d \vartheta_{P} \wedge d \vartheta_{T}\). Then the corresponding \(T_{i, j}^*\) are defined by
\begin{align*}
\begin{gathered}
        \delta_{jk} =  \int_{ U } s_{i,j}^* \wedge T_{i, k}^* \,  .
\end{gathered}
\end{align*}
As \(T_{i, j}^*\) is a closed, Dirichlet 2-form we select 
\begin{align*}
    T_{i, j}^* = f_1(r, \vartheta_{P}) \,  d r \wedge d \vartheta_{P} + f_2(r, \vartheta_{T}) \, d r \wedge d \vartheta_{T} + f_3(\vartheta_{P}, \vartheta_{T}) \, d \vartheta_{P} \wedge d \vartheta_{T} \, . 
\end{align*}
Let's first consider \({T_{i,1}}^*\):
\begin{align*}
    1 &= \int_U (d \vartheta_{T})  \wedge f_1(r, \vartheta_{P}) d r \wedge d \vartheta_{P} = \int_U f_1(r, \vartheta_{P})  d r \wedge d \vartheta_{P} \wedge d \vartheta_{T}  \, , \\
    0 &= \int_U (d \vartheta_{P})  \wedge f_2(r, \vartheta_{T}) d r \wedge d \vartheta_{T} = - \int_U f_2(r, \vartheta_{T}) d r \wedge d \vartheta_{P} \wedge d \vartheta_{T}   \, .
\end{align*}
Selecting \(f_1(r, \vartheta_{P}) =  1\), \(f_2(r, \vartheta_{T}) = 0\) and \(f_3(\vartheta_{P}, \vartheta_{T}) = 0\) gives us a closed 2-form \(T_{i,1}^* = d r \wedge d \vartheta_{P} \) that satisfies the above conditions. And pulling back to the boundary, \(\mathcal{J}^* T_{i, 1}^* = \mathcal{J}^* (  d r \wedge d \vartheta_{P}) = 0 \) as \(\mathcal{J}^* (dr)=0\). 

Performing the same for \(T_{i,2}^*\),
\begin{align*}
    0 &= \int_U (d \vartheta_{T} )  \wedge f_1(r, \vartheta_{P}) d r \wedge d \vartheta_{P} = \int_U f_1(r, \vartheta_{P})  d r \wedge d \vartheta_{P} \wedge d \vartheta_{T}  \, , \\
    1 &= \int_U (d \vartheta_{P} )  \wedge f_2(r, \vartheta_{T}) d r \wedge d \vartheta_{T} = - \int_U f_2(r, \vartheta_{T}) d r \wedge d \vartheta_{P} \wedge d \vartheta_{T}  \, .
\end{align*}
Selecting \(f_2(r, \vartheta_{T}) = - 1\), \(f_1(r, \vartheta_{P}) = 0\) and \(f_3(\vartheta_{P}, \vartheta_{T}) = 0\) gives us a closed 2-form \(T_{i,2}^* = -  d r \wedge d \vartheta_{T} \). Additionally, the boundary condition \(\mathcal{J}^*T_{i,2}^* = 0  \) is satisfied. 

We note that this choice for \(T_{i,1}^*, T_{i,2}^*\) agrees with \citet{Gauge_freedom} up to a sign. Clearly \(T_{i, j}^*\) are closed and 0 when pulled back to the boundary, and therefore generators for elements of \(H^2_{\mathrm{dR}}(U, \partial U)\).

For the singular relative 2-chains \(T_{i, j}\) we require that,
\begin{align*}
    \delta_{jk} =  \int_{T_{i, j}} T_{i, k}^* \, . 
\end{align*}
For \((w_1, w_2) \in  \Delta_2\), we try the 2-chain:
\begin{align*}
    T_{i,2} (w_1, w_2) = ( 1 + w_1 + w_2, 0, -  w_2    ) + ( 1 + w_2, 0, -1+w_1    ) \, . 
\end{align*}
The two singular simplices triangulate the square with corners \((1,0,0)\), \((2,0,0)\), \((2,0,-1)\) and \((1,0,-1)\), whose image in an unknotted hollow torus \(\HT\) is a toroidal ribbon. The vertices are ordered so that the shared diagonal is traversed the same way by both simplices and enters the two boundaries with opposite signs, which is what makes the cancellation below exact. 
\begin{align*}
        \int_{T_{i,2}} T_{i,2}^* &= \int_{( 1 + w_1 + w_2, 0, -     w_2    ) + ( 1 + w_2, 0, -1 + w_1    )} - dr \wedge d \vartheta_{T} 
     =    2 \int_0^1 \int_0^{1-w_2}   d w_1 \wedge d  w_2 
     =     1     \, .  \\
        \int_{T_{i,2}} T_{i,1}^* &= \int_{( 1 + w_1 + w_2, 0, -     w_2    ) + ( 1 + w_2, 0, -1 + w_1    )}  dr \wedge d \vartheta_{P} = 0 \, . 
\end{align*}
Note that the \( [T_{i,2}] \in H_2( U, \partial U )\) only if \(\partial T_{i,2} \in C_1 (\partial U) \), which we now verify.

Similarly, we may verify the following choice for \(T_{i,1}\):
\begin{align*}
    T_{i,1} (w_1, w_2) = ( 1 + w_1 + w_2,       w_2  , 0  ) + ( 1 + w_2,      1 - w_1 , 0   ) \, . 
\end{align*}
To complete our construction of an Alexander basis, we take the boundary of \([T_{i, j}]\):
\begin{align*}
    \partial T_{i, j}  
     : w \mapsto  T_{i, j} ( 1-w, w ) - T_{i, j} ( 0, w ) + T_{i, j} ( w, 0 ) \, ,
\end{align*}
where \(w \in [0,1] =  \Delta_1\). Starting with the representative \(T_{i,2}\):
\begin{align*}
    \partial T_{i,2} : w &\mapsto   T_{i,2} ( 1-w, w ) - T_{i,2} ( 0, w ) + T_{i,2} ( w, 0 ) \\
     &=  ( 2 , 0, -     w    ) + ( 1  + w, 0, -    w    ) \\
     & \; \; \; \; - (( 1 + w , 0, -     w    ) + ( 1 + w , 0, -   1    ) ) \\
     & \; \; \; \; + ( 1 + w, 0, 0   ) + ( 1 , 0, - 1 + w      ) \, .
\end{align*}
The two copies of \(w \mapsto ( 1 + w, 0, -w )\) are the shared diagonal and cancel as singular chains. The terms \(w \mapsto ( 1 + w, 0, -1 )\) and \(w \mapsto ( 1 + w, 0, 0 )\) are the same singular \(1\)-simplex, since \(\vartheta_{T}\) takes values in \(\mathbb{T} = \R / \mathbb{Z}\) and \(-1 = 0\) there, so they cancel as well. Hence
\begin{align*}
    \tau_{i,2} = \partial T_{i,2} : w &\mapsto  ( 1 , 0,  - 1 +   w    )  + ( 2 , 0, -    w    ) =  ( 1 , 0,     w    )  + ( 2 , 0, -    w    )\, ,
\end{align*}
which lies in \(C_1(\partial U)\), the two simplices lying on \(r=1\) and \(r=2\) respectively.
Similarly for \(T_{i,1}\):
\begin{align*}
    \partial T_{i,1} : w &\mapsto   T_{i,1} ( 1-w, w ) - T_{i,1} ( 0, w ) + T_{i,1} ( w, 0 ) \\
     &=  ( 2 ,       w ,0   ) + ( 1  + w,       w  ,0  ) \\
     & \; \; \; \; - (( 1 + w ,       w  ,0  ) + ( 1 + w ,      1  ,0  ) ) \\
     & \; \; \; \; + ( 1 + w,  0  ,0 ) + ( 1 ,      1 - w ,0   ) \, .
\end{align*}
The diagonal cancels as before, and \(w \mapsto ( 1 + w, 1, 0 )\) and \(w \mapsto ( 1 + w, 0, 0 )\) agree because \(\vartheta_{P}\) is also valued in \(\mathbb{T}\), leaving
\begin{align*}
    \tau_{i,1} = \partial T_{i,1} : w &\mapsto   ( 1 ,   1 -     w  ,0  ) + ( 2 ,      w  ,0  ) =  ( 1 ,   -     w  ,0  ) + ( 2 ,      w  ,0  )       \, .
\end{align*}
Together the classes for \(\sigma_{i,j}\) and \(\tau_{i,j}\) form an Alexander basis for the parameter space of a hollow torus. We observe that our definition of loop-corrected helicity in Subsection \ref{subsection:Relative Helicity} is given in this context by \citet{Gauge_freedom}.

\subsection{Hodge Decomposition}\label{append:Hodge}

Throughout the paper we utilise both a Hodge decomposition, and a Hodge--Morrey--Friedrichs decomposition. Here we provide these decompositions without proof. A Hodge-Morrey decomposition for the space \( \Omega^1  (\ST_i)\) is given by:
\begin{align*}
     \Omega^1  (\ST_i) = &  \{      d\omega \in  \Omega^1 (\ST_i)  \; : \, \omega \in \Omega^0 (\ST_i) \, , \; t \omega = 0     \} \\  
     \oplus & 
    \{      \delta \omega \in  \Omega^1 (\ST_i)    \; : \, \omega \in \Omega^2 (\ST_i) \, , \; n \omega = 0       \} \\  
    \oplus 
    &
     \{      \omega \; \, \in \Omega^1 (\ST_i)   \; : \, d \omega = \delta \omega = 0       \} \, .
\end{align*}
These spaces are \(L^2\)-orthogonal to each other, giving a unique algebraic decomposition for any \(\omega \in  \Omega^1  (\ST_i)\) \citep[Lemma 2.4.3]{schwarz-1995}:
\begin{align*}
    \omega = d \phi + \delta (\star \eta) + \gamma  \, , 
\end{align*}
where \( \eta \in  \Omega_D^{1} (\ST_i) \), \(\phi \in  \Omega^0 (\ST_i)\) where \(\phi|_{\partial \ST_i} = 0\), and \(\gamma \in \mathcal{H}^1( \ST_i)\). This is further decomposed to give a Hodge--Morrey--Friedrichs decomposition \citep[Corollary 2.4.9]{schwarz-1995}:
\begin{align*}
     \Omega^1  (\ST_i) = &  \{      d\omega \in  \Omega^1 (\ST_i)  \; : \, \omega \in \Omega^0 (\ST_i) \, , \; t \omega = 0   \} \\  \oplus &
    \{      \delta \omega \in  \Omega^1  (\ST_i)    \; : \, \omega \in \Omega^2 (\ST_i) \, , \; n \omega = 0   \} \\  \oplus &
    \{      \omega \; \, \in  \Omega^1  (\ST_i)   \; : \, d \omega = \delta \omega = n \omega= 0       \}\\ \oplus&
    \{      d \omega \in  \Omega^1  (\ST_i)   \; : \, \omega \in \Omega^0 (\ST_i) \, , \; \delta d \omega  = 0        \} \, .
\end{align*}
The first and fourth summands meet only in \(\{0\}\), a function that is both Dirichlet and harmonic being zero.
These spaces are \(L^2\)-orthogonal to each other and give a unique algebraic decomposition for any \(\omega \in  \Omega^1  (\ST_i)\):
\begin{align*}
    \omega = d \phi + \delta (\star \eta) + \star \lambda + dz  \, , 
\end{align*}
where \( \eta \in  \Omega_D^{1} (\ST_i) \), \(\phi \in  \Omega^0(\ST_i)\), \(\star \lambda \in \mathcal{H}^1_N (\ST_i)\), and \(z \in  \Omega^0 (\ST_i)\) for \(\phi|_{\partial \ST_i} = 0\), and \(z\) is a harmonic function \(\Delta z = 0\) where \(\Delta \) is the Laplacian.

\subsection{Lagrange Multipliers}\label{app:LM}

Let \(F : \Omega_D^1(\ST_i) \to \R\) be the objective and \(g_i\) the constraint of Subsection \ref{sub:LM}, and let \(\eta_i\) satisfy \(g_i(\eta_i) = h_i\) with \(\beta_i \neq 0\). By Lemma \ref{lem:tangent} the set of admissible variations at \(\eta_i\), namely the set of derivatives at \(t=0\) of families in \(\Omega^1_D(\ST_i)\) along which \(g_i\) is constant, is exactly \(\ker( d_{\eta_i} g_i )\). By Lemma \ref{lem:quad} the first variation of \(F\) along such a family is \(d_{\eta_i} F\) applied to that derivative. So \(\eta_i\) is a critical point of \(F\) subject to \(g_i = h_i\) precisely when
\begin{align*}
    d_{\eta_i} F(w) = 0 \quad \text{for all} \quad w \in \ker( d_{\eta_i} g_i ) \, .
\end{align*}

We show that this is equivalent to the existence of a \(\mu_i \in \R\) with
\begin{align*}
    d_{\eta_i} F = \frac{\mu_i}{2} \, d_{\eta_i} g_i \qquad \text{on the whole space } \Omega_D^1(\ST_i) \, ,
\end{align*}
which is the convention carried by the Lagrangian of \eqref{eq:ext_lagrange}. By Lemma \ref{lem:nondegenerate} there is a \(\nu_0 \in \Omega^1_D(\ST_i)\) with \(d_{\eta_i} g_i(\nu_0) = 1\), so \(\Omega_D^1(\ST_i)\) splits as a vector space,
\[ \Omega_D^1(\ST_i) = \ker(d_{\eta_i} g_i) \oplus \text{span}\{\nu_0\} \, . \]
Indeed any \(w\) decomposes as \(w = \big( w - d_{\eta_i} g_i(w) \nu_0 \big) + d_{\eta_i} g_i(w) \nu_0\), the first term lying in \(\ker( d_{\eta_i} g_i )\) because \(d_{\eta_i} g_i( w - d_{\eta_i} g_i(w) \nu_0 ) = d_{\eta_i} g_i(w) - d_{\eta_i} g_i(w) \, d_{\eta_i} g_i(\nu_0) = 0\), and the intersection of the two summands being trivial because \(d_{\eta_i} g_i( c \nu_0 ) = 0\) forces \(c = 0\). Only this algebraic splitting is used, and no topology on \(\Omega^1_D(\ST_i)\) is required.

Writing \(w = w_k + c \nu_0\) with \(w_k \in \ker(d_{\eta_i} g_i)\), and applying \(d_{\eta_i} g_i\) to both sides, the scalar \(c\) is exactly \(d_{\eta_i} g_i(w)\). Hence
\begin{align*}
    d_{\eta_i} F(w) = d_{\eta_i} F(w_k) + c \, d_{\eta_i} F(\nu_0)
    = d_{\eta_i} g_i(w) \, d_{\eta_i} F(\nu_0) \, ,
\end{align*}
using \(d_{\eta_i} F(w_k) = 0\). Setting \(\mu_i \coloneqq 2 \, d_{\eta_i} F(\nu_0)\) gives \(d_{\eta_i} F = ( \mu_i / 2 ) \, d_{\eta_i} g_i\) on \(\Omega_D^1(\ST_i)\). This confirms the existence of a Lagrange multiplier at a critical point. Going the other way, we need to show that \(d_{\eta_i} F = ( \mu_i / 2 ) \, d_{\eta_i} g_i\) implies \(\eta_i\) is a critical point. This is immediate, since \(w \in \ker (d_{\eta_i} g_i )\) implies that \(d_{\eta_i} F(w) = 0\).

\section{Existence of Solutions} \label{append:exist}

With \(n=1\) there are no internal interfaces, so the diffeomorphism \(f\in\I\) plays no role. Hence, we take \(f=\mathrm{id}\) and drop partition subscripts. In particular \(f^*\varpi = \varpi\) and \(f(\ST) = \ST\), so that \(\mathcal{C}(\ST)\) and \(C(f(\ST))\) correspond under the identity, every metric quantity below is the Euclidean one, and no transcription is required. Otherwise, the notation in this section is consistent with the remainder of the paper.

 We look for a solution \(\beta\) to
\begin{align}  \label{problem:a_only}
    \begin{gathered} 
    \underset{\beta  \in  \mathcal{C} ( \ST ) }{ \text{min} }  \| \beta   \|_{L^2}^2  \, , \\
    H_{\mathscr{A}(\ST)} (\beta,  \ST) = h \, , \qquad \int_{T_{j}} \beta  = \psi_{j} \quad j = 1,2, \dots, \ell \, ,
    \end{gathered}
\end{align}
for an Amperian primitive \(a \in \mathscr{A}(\ST)\), where \(\{[T_1], \dots, [T_\ell]\}\) is a basis of \(H_2(\ST,\partial\ST)\) whose boundaries coincide with the usual \([\tau_j]\). Fixing an Amperian gauge makes the helicity a function of \(\beta \) alone. The pressure is irrelevant and is omitted.

\begin{lemma}\label{lemma:c1}
    The following set is non-empty:
    \[
    \mathscr Q(\psi_1,\dots,\psi_\ell,h,\ST)
    :=\Big\{\beta \in\mathcal C(\ST)\ :\ H_{\mathscr{A}(\ST)}(\beta,\ST)=h,\ \langle[ \beta ],[T_j]\rangle=\psi_j\ \forall j\Big\},
    \]
    with the flux conditions absent when \(\ell=0\). 
\end{lemma}
\begin{proof}
As $\ST$ is a non-empty, compact 3-manifold there exists an open 3-dimensional ball $\B(r) \subset \text{int}(\ST)$ of some radius $r>0$. Set \(\widehat{\ST}\coloneqq \ST\setminus \B(r)\). Consider \(\ell>0\) and a closed \(w \in \Omega^1(\partial \ST)\), so
\begin{align*}
    [w] = \sum_{j=1}^\ell c^1_j [\tau_j^*] + c^2_j [\sigma_j^*] \, .
\end{align*}
We select \(c^1_j = \psi_j\) and \(c^2_j = 0\). In Euclidean coordinates \((x_1, x_2, x_3)\), via a partition of unity, there are \(\widehat g_j \in \Omega^0(\partial \ST)\) for \(j=1,2,3\) such that
\begin{align*}
    w = \sum_{j=1}^3 \widehat g_{j}\,dx_j \, .
\end{align*}
Since $\partial {\ST}$ is compact in \(\R^3\) and contained in $\ST$, we may extend $\widehat g_j$ to a smooth function \(g_j\) on $\ST$ (\(j=1,2,3\)) by pulling back along the projection of a tubular neighbourhood of $\partial \ST$ in $\ST$ and patching with a partition of unity. 

Let \(\chi\in \Omega^0 (\ST)\) be a bump function with \(\chi=0\) on \(\B(r_0)\) (with \(r_0 < r\)) and \(\chi=1\) on \(\widehat{\ST}\). Define
\[
a_{\ST}\coloneqq \chi \sum_{j=1}^3 g_{j}\,dx_j   \in\Omega^1(\ST),\qquad \text{and} \qquad \beta_{\ST}\coloneqq d a_{\ST} \in\Omega^2(\ST) \, ,
\]
giving \(a_{\ST}|_{\B(r_0)}= 0\). Additionally, \(\beta _{\ST}\) is exact by construction and vanishes when pulled back to \(\partial\ST\)
\begin{align*}
    t(\beta _{\ST}) = t(d a_{\ST}) = d t(a_{\ST}) = d w = 0 \, .
\end{align*}
Hence \(\beta _{\ST} \in \mathcal{C}(\ST)\) and satisfies the flux constraint:
\begin{align*}
    \int_{T_j} \beta _{\ST} = \int_{\tau_j} \mathcal{J}^* a_{\ST} = \int_{\tau_j} w = \psi_j \, .
\end{align*}
Additionally, 
\begin{align*}
    \int_{\sigma_j} \mathcal{J}^* a_{\ST} = \int_{\sigma_j} w = 0 \, ,
\end{align*}
indicating that \(a_{\ST} \in \mathscr{A}(\ST)\).

 There exists a field \(\beta _\B \in \mathcal{C}(\ST)\) that is only non-zero within \(\B(r_0)\) and able to achieve any value for helicity \citep[supplementary material]{PhysRevLett.117.274501}. Such \(\beta _\B\) were constructed explicitly by Kedia, Foster, Dennis, and Irvine. These authors show that a smooth potential \(a_\B\) may be constructed from two linked trefoil knots to achieve any helicity, and that \(a_\B, \beta _\B = 0\) outside of a bounding ball and therefore the flux of \(\beta _{\ST} + \beta _\B \in \mathcal{C}(\ST)\) is given by \(\beta _{\ST}\). As \(a_\B = 0\) on the boundary this implies \(a_\B \in \mathscr{A}(\ST)\). 

The helicity for \(\beta _{\ST} + \beta _\B\) is:
\begin{align*}
    H_{\mathscr{A}(\ST)}( \beta_\B + \beta_\ST , \ST) = \int_{\ST} ( a_\ST + a_\B ) \wedge ( \beta_\ST + \beta_\B ) =  H_{\mathscr{A}(\ST)} (\beta_\ST , \ST ) +  H_{\mathscr{A}(\ST)} ( \beta_\B , \ST ) \, ,
\end{align*}
the cross terms vanishing because \(a_\ST\) and hence \(\beta_\ST\) are zero on \(\B(r_0)\), while \(a_\B\) and \(\beta_\B\) are zero outside it.
Immediately we see that, after fixing a selection for \({\beta}_\ST\) with the desired flux, we may select a \(  \beta_\B \) to generate any value for \(H_{\mathscr{A}(\ST)}(\beta , \ST) \), without affecting the flux. We now have a way to construct any flux and helicity with a closed Dirichlet 2-form \(\beta  \in \mathcal{C} (\ST)\).

If \(\ell = 0\) then there is no flux and one only needs to reconstruct the helicity. The above arguments tell us that \(\mathscr{Q}( \psi_1, \dots, \psi_\ell , h, \ST )\) is non-empty for any \(\ell\). 
\end{proof}

We now apply the direct method in the calculus of variations to a weak formulation of Problem \eqref{problem:a_only}. We prove the existence of a weak minimiser and establish smoothness in the interior.

By Lemma \ref{lemma:c1}, we may choose
\[
\beta \in \mathscr{Q}(\psi_1,\dots,\psi_\ell,h,\ST).
\]
Moreover, there exists a potential for \(\beta\) given by
\[
a=\eta+\Gamma \in \mathscr{A}(\ST),
\]
where \(d\Gamma \in \mathcal{H}^2_D (\ST)\) is the unique harmonic term and \(\eta \in \Omega^1_D(\ST)\). Given that \(\eta\) is zero when pulled back to the boundary, only the \(\Gamma\) term contributes to the periods of \(a\) around the loops \(\sigma_j\). Therefore, \(\Gamma\) is solely responsible for ensuring \(a \in \mathscr{A}(\ST)\). Without loss of generality, we fix \(\Gamma\).

Orthogonality of the HMF decomposition (see Appendix \ref{append:Hodge}) implies 
\begin{align*}
    \| \beta  \|_{L^2}^2 = \|d \eta \|_{L^2}^2 + \| d \Gamma \|_{L^2}^2 \, .
\end{align*}
The last term is determined by a choice of \(\psi_1, \dots, \psi_\ell\). Hence, Problem \eqref{problem:a_only} descends to an equivalent problem in \(d \eta\), namely, to determine if \(\|d \eta \|_{L^2}^2\) has any critical points with the constraint \(H_{\mathscr{A}(\ST)}( \beta , \ST) = h\).

It is convenient to enlarge the admissible space of \(2\)-forms containing \(d\eta\). We use the closed Hilbert subspace
\[
    V\coloneqq\overline{d\Omega_D^1(\ST)}^{\,L^2}\subset L^2\Omega^2(\ST),
\]
which retains the closed, flux-free admissibility conditions in their weak form. The fixed harmonic component is \(d\Gamma\).

For $1$-forms, we equip $\Omega_D^1(\ST)$ with the $H^1$ norm
\[
\|\omega\|_{H^1}^2 \coloneqq \|\omega\|_{L^2}^2 + \|D\omega\|_{L^2}^2,
\qquad \omega \in \Omega_D^1(\ST),
\]
where $D$ denotes the weak covariant derivative used by Schwarz \citeyearpar{schwarz-1995}. The corresponding $H^1$ inner product is
\[
\langle \omega,\omega_0\rangle_{H^1}
\coloneqq
\langle \omega,\omega_0\rangle_{L^2}
+
\langle D\omega, D\omega_0\rangle_{L^2},
\qquad \omega,\omega_0 \in \Omega_D^1(\ST).
\]

More generally, for the definitions of the spaces $H^k\Omega_D^1(\ST)$, $H^k_{\mathrm{loc}}\Omega_D^1( \text{int}( \ST ) )$, and $C^k\Omega_D^1(\ST)$ on a manifold with boundary, together with their associated norms and the tangential trace operator $t$, we refer the reader to \citet{schwarz-1995}. We now turn to the minimisation problem
\begin{align}\label{prob:inter} 
    \begin{gathered} 
    \underset{  d \eta \in V}{ \text{min}}   \| d \eta  \|_{L^2}^2  \, , \qquad 
    \hat H ( \eta) = h  \, ,
    \end{gathered}
\end{align}
where 
\begin{align*}
    \hat H  : H^1 \Omega^1_D(\ST) \to \R \qquad \text{is given by} \qquad \hat H  ( \eta) = \int_{\ST} (\eta + \Gamma) \wedge d (\eta + \Gamma)  \, ,
\end{align*}
where the helicity \(\hat H\) is dependent upon \( d \eta\) irrespective of the primitive \(\eta\) chosen. This is an affine analogue of the flux-free problem treated by Gerner \citep[Theorem 2.1]{exist_gerner}. We adapt the direct-method argument below with the harmonic component fixed, and then prove interior smoothness.

For a minimising sequence, \( d \eta_j \) in \(j\), boundedness of the minimising sequence together with Hilbert-space weak compactness of \(V\) implies that \(d \eta_j \to d \eta\) weakly in \(L^2\) as \(j \to \infty\) (up to a subsequence). The continuous linear inverse-curl operator of Gerner \citep[Lemma 3.4]{exist_gerner}, applied to \(\star d\eta\) for \(d\eta\in V\), selects \(\eta\in H^1\Omega_D^1(\ST)\) with the required exterior derivative. Hence, \( \eta_j \to  \eta\) converges weakly in \(H^1\). By the Rellich–Kondrachov theorem there is an embedding into  \(L^2 \Omega^1_D(\ST)\) where a subsequence of \(  \eta_j\) converges strongly to \(\eta\). So
\begin{align*}
    \hat H  (   \eta_j ) - \hat H  (   \eta ) &=  \int_{\ST} \left(   \eta_j + \Gamma \right)  \wedge \beta _j
    - 
    \int_{\ST} \left(   \eta +  \Gamma \right)  \wedge \beta   
    \\ &=   \int_{\ST} \left(   \eta_j - \eta \right)  \wedge \beta _j 
    - 
    \int_{\ST} \left(   \eta +  \Gamma \right)  \wedge  (\beta - \beta_j)   \, .
\end{align*}
The first term tends to zero because \( \eta_j - \eta \to 0\) strongly in \(L^2\) and \( \| \beta_j \|_{L^2}\) is bounded. The second term goes to zero because \(\beta - \beta_j\) goes to zero in \(L^2\) when tested against the test function \(\eta +  \Gamma  \) by weak convergence. We see that \(\hat H  (   \eta_j ) \to \hat H  (   \eta ) \), namely \(\hat H\) is sequentially weakly continuous. This implies that \(\hat H^{-1} (h)\) is weakly closed.

The objective \(\| d \eta  \|_{L^2}^2\) is weakly lower semicontinuous and trivially coercive. It follows from the direct method on Banach spaces that there exists a minimiser \(d \eta\).

We now determine the Euler--Lagrange equations for Problem \eqref{prob:inter}. Lemma \ref{lem:quad} extends from \(\Omega^1_D(\ST)\) to \(H^1 \Omega^1_D(\ST)\). Indeed, each identity there is continuous in the \(H^1\) topology, since
\begin{align*}
    \Big| \int_{\ST}  \eta  \wedge d \eta \Big| = \big| \langle \eta  , \star d \eta \rangle \big|
    \leq
    \| \eta \|_{L^2} \, \|  d \eta  \|_{L^2}
    \leq
    C \| \eta \|_{H^1}^2 \, ,
\end{align*}
using \(\| \eta \|_{L^2} \leq C \| \eta \|_{H^1}\) and, \(d\) being a first-order differential operator on a compact smooth volume, \(\| d \eta \|_{L^2} \leq C \| \eta \|_{H^1}\), and since \(| d_{\eta} \hat H ( w ) | \leq 2 \| w \|_{L^2} \| \beta \|_{L^2}\). As the identities hold on the dense subspace \(\Omega^1_D(\ST)\) they hold throughout. So \(\hat H\) is again an exact quadratic polynomial, with no loop correction because the primitives are Amperian, and
\begin{align*}
    d_{\eta} \hat H(  w  ) = 2 \int_{\ST} w \wedge  \beta  \, , \qquad w \in H^1 \Omega^1_D(\ST) \, .
\end{align*}

It remains to exhibit a \(\nu\) as in Lemma \ref{lem:nondegenerate}. Suppose \(\| \beta \|_{L^2} > 0\). Compactly supported smooth \(1\)-forms on \(\text{int}(\ST)\) lie in \(H^1 \Omega^1_D(\ST)\) and are dense in \(L^2 \Omega^1(\ST)\), so we may choose such a \(\nu_0\) with \(\| \nu_0 - \star \beta \|_{L^2} < \| \beta \|_{L^2}\). Then, using \(\star \star = 1\) on \(\Omega^2(\ST)\),
\begin{align*}
    d_{\eta} \hat H (\nu_0)
    = & \, 2 \big\langle \nu_0 , \star \beta \big\rangle_{L^2}
    = 2 \| \beta \|_{L^2}^2 + 2 \big\langle \nu_0 - \star \beta , \star \beta \big\rangle_{L^2}
    \\ \geq & \, 2 \| \beta \|_{L^2} \big( \| \beta \|_{L^2} - \| \nu_0 - \star \beta \|_{L^2} \big)
    > 0 \, ,
\end{align*}
and \(\nu \coloneqq \nu_0 / d_{\eta} \hat H (\nu_0)\) satisfies \(d_{\eta} \hat H (\nu) = 1\). Lemma \ref{lem:tangent} and Corollary \ref{cor:multiplier} therefore apply with \(\hat H\) in place of \(g_i\), families now being understood as differentiable curves in \(H^1 \Omega^1_D(\ST)\), the increments along them being the same polynomial identities. They give necessary and sufficient conditions for a critical point of \eqref{prob:inter} whenever \(\| \beta \|_{L^2} > 0\). If \(\|\beta\|_{L^2}=0\), the field is zero and is already a global energy minimiser. The resulting conditions are those of
\begin{align*}
    \underset{\substack{\eta \in  H^1 \Omega^1_D(\ST) \\  \mu \in \R} }{ \text{c. p.} }  \left(  \frac{1}{2} \Big\|d \eta + \sum_j  \psi_{j}  T_{j}^*  \Big\|_{L^2}^2   - \frac{\mu}{2} \big(   \hat H (\eta) -   h       \big) \right) \, .
\end{align*}
We have that stationary points of the above satisfy
\begin{align*}
    0 &=    - \frac{\mu'}{2} \big(   \hat H (\eta) -   h       \big) \, ,
\end{align*}
for all \(\mu'\) which implies that the helicity constraint holds. And the derivative with respect to \(\eta\) in the direction \(\eta'\) is
\begin{align*}
      0 &= \int_{\ST} d \eta' \wedge \star \beta   - \frac{\mu}{2} \int_{\ST} \eta'  \wedge \beta + a \wedge d \eta'  \\
      0 &= \int_{\ST} \eta' \wedge (d\star \beta  - \mu \beta ) \, ,
\end{align*}
for all \(\eta'\), where we have used Stokes' theorem and that \(\beta = da\) to simplify this condition. The condition implies that
\begin{align*}
    \langle   (d\star \beta  - \mu \beta )   , \star \eta' \rangle = 0 \, . 
\end{align*}
On the interior of \(\ST\) we see that \( d\star \beta  - \mu \beta \) projected onto any element of \(\Omega^2( \text{int} (\ST ))\) is zero. Hence, in \(\text{int} (\ST)\) we have that \( d\star \beta  - \mu \beta  = 0\) in the sense of distributions with \(\beta \in H^0 \Omega^2 (\ST)\). 

We may apply \((d\star   - \mu   )\), as a distributional operator, to \( d\star \beta  - \mu \beta  = 0\). Namely,
\begin{align*}
    0 = (d\star   - \mu   )^2 \beta &= (d\star d \star   - 2 \mu d\star + \mu^2   ) \beta =  (d\star d \star   -  \mu^2   ) \beta = (\Delta - \mu^2 ) \beta  \, ,
\end{align*}
where \(\Delta \beta = (d \delta + \delta d) \beta = d \delta \beta= d \star d \star  \beta \) and \(\Delta\) is the Hodge Laplacian on 2-forms. A Hodge decomposition gives \(\beta = d ( \eta + \Gamma) \) where \(\Gamma\) is known and smooth, and \(d \Gamma \) is harmonic, so
\begin{align*}
    (\Delta - \mu^2 ) d \eta &=  \mu^2   d \Gamma \, . 
\end{align*}
Since the right-hand side of this equation is smooth, interior elliptic regularity theory implies that \(d \eta \in H^k_{\mathrm{loc}} (  \text{int}(\ST)  )\) for all \(k\), namely, \(d \eta \) is smooth within the domain interior \cite{gilbarg1998elliptic}. 


Hence, there exists a minimiser \(\eta\) to Problem \eqref{prob:inter}, giving a \(\beta = d (\eta + \Gamma)\) and if the minimiser satisfies \(\| \beta \|_{L^2} > 0 \) then \(d \eta\) and therefore \(\beta\) must also be smooth within the domain interior. We note that if at least one of the flux values \(\psi_1, \dots, \psi_\ell\) is non-zero this implies that (by \(L^2\) orthogonality of the Hodge decomposition) \(\| \beta \|_{L^2}^2 \geq \|d \Gamma \|_{L^2}^2 > 0\).

\section{Vector Calculus Notation}\label{append:vector_notation}

The body of this paper is written in the language of exterior calculus. For the convenience of readers we record the translation of objects and results into vector calculus notation.

We fix \(f\in\I\). Our vector calculus objects live in the laboratory domain, namely each \(B_i\) and \(A_i\) are ordinary vector fields on the region \(f(\ST_i)\subseteq\R^{3}\) with Euclidean metric and associated operators. 

As \(f\) is an orientation preserving isometry of \(\ST_i\) onto \(f(\ST_i)\) the pullback by \(f\) satisfies \(f^{*}(i_{V}\omega)=i_{f^{*}V}(f^{*}\omega)\) where \(f^{*}V\coloneqq(df)^{-1}\circ V\circ f\), and preserves pointwise inner products and integrals. The dictionary below is the resulting correspondence
\[
  \beta_i = f^{*}\big(i_{B_i}\varpi\big) , \qquad a_i = f^{*}\big(A_i^{\flat}\big) ,
\]
with inverses
\[
  B_i = \big((f^{-1})^{*}(\star\beta_i)\big)^{\sharp} , \qquad
  A_i = \big((f^{-1})^{*}a_i\big)^{\sharp} .
\]
Curves, surfaces and interfaces are named on \(\ST_i\), so every domain of integration in the vector calculus column of Table \ref{tab:dictionary} carries an \(f\). For the homological constructions, we extend \(f\) to an ambient homeomorphism of \(\R^{3}\) by the identity, which is consistent because \(f|_{\partial\ST}=\mathrm{id}\), so that \(f(\ST_i^{c})=f(\ST_i)^{c}\). Finally, \(\tilde p_i = f^{*}p_i = p_i\circ f\) is the pressure of Subsection \ref{sub:setup} transcribed to \(\ST_i\). For the constant pressures of Problem \eqref{prob:MRxMHD_fixed} the two pressures agree as real numbers and we keep the tilde only to record on which region the function is defined.

\subsection{Dictionary}

Table \ref{tab:dictionary} lists the translations used in the body of the paper. All of them follow from the known identities collected in Lemma \ref{lem:vector-dictionary} below.

\begin{table}
  \centering
  \begin{tabular}{@{} | l | lll | @{}}
    \hline
    Object & Exterior calculus on \(\ST_i\) & Vector calculus on \(f(\ST_i)\) & Reference \\
    \hline
    Magnetic field
      & \(\beta_i = f^{*}\big(i_{B_i}\varpi\big)\)
      & \(B_i\) 
      & Lemma \ref{lem:vector-dictionary}(1) \\
    Divergence-free
      & \(d\beta_i = 0\)
      & \(\nabla\cdot B_i = 0\) 
      & Lemma \ref{lem:vector-dictionary}(3) \\
    Tangent to boundary
      & \(t\beta_i = 0\)
      & \(B_i\cdot N = 0\) on \(\partial f(\ST_i)\)
      & Lemma \ref{lem:vector-dictionary}(4) \\
    Admissible fields
      & \(\mathcal{C}(\ST_i)\)
      & \(C(f(\ST_i))\)
      & Lemma \ref{lem:vector-dictionary} \\
    Pointwise norm
      & \(\|\beta_i\|\)
      & \(\|B_i\|\circ f\)
      & Lemma \ref{lem:vector-dictionary}(2) \\
    Pressure
      & \(\tilde p_i = f^{*}p_i\)
      & \(p_i\) 
      & Subsection \ref{sub:setup} \\
    Volume
      & \( \myabs{f(\ST_i)}= \displaystyle \int_{\ST_i}f^{*}\varpi\)
      & \(\myabs{f(\ST_i)}= \displaystyle  \int_{f(\ST_i)}\varpi\)
      & \eqref{eq:energy-vector} \\
   Squared \(L^2\) norm
      & \(\|\beta_i\|_{L^{2}}^2\)
      & \( \displaystyle \int_{f(\ST_i)}\|B_i\|^{2}\, \varpi \)
      & \eqref{eq:energy-vector} \\
    Potential
      & \(da_i = \beta_i\)
      & \(\nabla\times A_i = B_i\) 
      & Lemma \ref{lem:vector-dictionary}(5) \\
    Helicity
      & \( \displaystyle \int_{\ST_i} a_i\wedge\beta_i\)
      & \( \displaystyle \int_{f(\ST_i)} A_i\cdot B_i\, \varpi \)
      & \eqref{eq:helicity-vector} \\
    Flux
      & \( \displaystyle \int_{T}\beta_i\)
      & \( \displaystyle \int_{f(T)} B_i\cdot dS\)
      & Lemma \ref{lem:vector-dictionary}(4) \\
    Loop integral
      & \( \displaystyle \int_{\gamma} a_i\)
      & \( \displaystyle \oint_{f(\gamma)} A_i\cdot d\ell\)
      & \eqref{eq:circulation-is-flux} \\
    Beltrami equation
      & \(d\star\beta_i = \mu_i\beta_i\)
      & \(\nabla\times B_i = \mu_i B_i\) 
      & Lemma \ref{lem:vector-dictionary}(6) \\
    Force balance
      & \(\delta\beta_i\wedge\star\beta_i = \star\,d\tilde p_i\)
      & \((\nabla\times B_i)\times B_i = \nabla p_i\)
      & Lemma \ref{lem:vector-dictionary}(8) \\
    Jump condition
      & \(\llbracket (\sfrac12) \|\beta\|^{2}+\tilde p \rrbracket\big|_{I_q} = 0\)
      & \(\llbracket (\sfrac12) \|B\|^{2}+p \rrbracket\big|_{f(I_q)} = 0\)
      & \eqref{eq:MRxMHD-vector} \\
    Loop-corrected helicity
      & \(\mathscr{H}(\beta_i,\ST_i)\)
      & \(\mathscr{H}\big(B_i,f(\ST_i)\big)\)
      & \eqref{eq:relhel-vector} \\
    \hline
  \end{tabular}
  \caption{Translation between the exterior calculus of the body of the paper and vector calculus. Throughout, \(A_i = \big((f^{-1})^{*}a_i\big)^{\sharp}\) is the vector field on \(f(\ST_i)\) corresponding to the \(1\)-form \(a_i\) on \(\ST_i\), and \(T\) and \(\gamma\) are a surface and a curve in \(\ST_i\), with images \(f(T)\) and \(f(\gamma)\).}
  \label{tab:dictionary}
\end{table}

\begin{lemma}\label{lem:vector-dictionary}
Let \(B_i\) be a smooth vector field on \(f(\ST_i)\) and set
\[
  \beta_i \coloneqq f^{*}\big(i_{B_i}\varpi\big) = i_{f^{*}B_i}\big(f^{*}\varpi\big)\in\Omega^{2}(\ST_i) .
\]
Let \(a_i\in\Omega^{1}(\ST_i)\) and write \(A_i \coloneqq \big((f^{-1})^{*}a_i\big)^{\sharp}\), a vector field on \(f(\ST_i)\), equivalently \(a_i = f^{*}\big(A_i^{\flat}\big)\). Then
\begin{enumerate}[label={\arabic*. }  ]
  \item \(\beta_i = f^{*}\big(\star\big(B_i^{\flat}\big)\big) = \star\big((f^{*}B_i)^{\flat}\big)\) where the Hodge star changes in the metric accordingly, equivalently \(\big((f^{-1})^{*}(\star\beta_i)\big)^{\sharp} = B_i\),
  \item \(\|\beta_i\| = \|B_i\|\circ f\) pointwise on \(\ST_i\), and consequently \(\|\beta_i\|_{L^{2}(\ST_i)}^{2} = \|B_i\|_{L^{2}(f(\ST_i))}^{2}\),
  \item \(d\beta_i = \big((\nabla\cdot B_i)\circ f\big)\,f^{*}\varpi\),
  \item \(t\beta_i = \big((B_i\cdot N)\circ f\big)\,f^{*}dS\), where \(N\) and \(dS\) are the outward unit normal and the area element of \(\partial f(\ST_i)\) and \(f^{*}dS\), the pullback by \(f|_{\partial\ST_i}\), is the area element of \(\partial\ST_i\) and more generally
        \[
          \int_{T}\beta_i = \int_{f(T)} B_i\cdot dS \, ,
        \]
        for any oriented embedded surface \(T\subseteq\ST_i\),
  \item \(da_i = f^{*}\big(\star\big((\nabla\times A_i)^{\flat}\big)\big)\), so that \(da_i = \beta_i\) if and only if \(\nabla\times A_i = B_i\) on \(f(\ST_i)\),
  \item \(d\star\beta_i = f^{*}\big(\star\big((\nabla\times B_i)^{\flat}\big)\big)\),
  \item \(a_i\wedge\beta_i = \big((A_i\cdot B_i)\circ f\big)\,f^{*}\varpi\),
  \item for smooth vector fields \(u,v\) on \(f(\ST_i)\) one has \(u^{\flat}\wedge v^{\flat} = \star\big((u\times v)^{\flat}\big)\), and for \(p\in C^{\infty}(f(\ST_i))\) one has \(d\big(f^{*}p\big) = f^{*}\big(dp\big)\).
\end{enumerate}
In particular \(\beta_i\in\mathcal{C}(\ST_i)\) if and only if \(B_i \in C(f(\ST_i))\). Additionally, \(d\star\beta_i = \mu_i\beta_i\) if and only if \(\nabla\times B_i = \mu_i B_i\), and \(\delta\beta_i\wedge\star\beta_i = \star\,d\tilde p_i\) if and only if \((\nabla\times B_i)\times B_i = \nabla p_i\).
\end{lemma}

\begin{proof}
Since \(f\) is an orientation preserving isometry from \(\ST_i\) with the pulled-back metric onto \(f(\ST_i)\) with the Euclidean metric, \(f^{*}\varpi\) is precisely the Riemannian volume form of \(\ST_i\), and \(f^{*}\) commutes with \(d\), \(\star\), \(\flat\), \(\sharp\) and interior multiplication, in the sense that \(f^{*}(i_{V}\omega) = i_{f^{*}V}(f^{*}\omega)\). Moreover, \(f^{*}\) preserves pointwise inner products and integrals. It therefore suffices to establish each identity on \(f(\ST_i)\) with the Euclidean metric and to transport it by \(f^{*}\).

For (1), on an oriented Riemannian \(3\)-manifold with volume form \(\varpi\) one has \(i_{V}\varpi = \star(V^{\flat})\) for every smooth vector field \(V\) \cite{lee}. Applying this to \(B_i\) on \(f(\ST_i)\) gives the first equality, and commuting \(f^{*}\) past \(\star\) and \(\flat\) gives the second. The last follows because \(\star\star = \mathrm{id}\) on forms of every degree in three dimensions with a Riemannian metric, so that \(\star\beta_i = f^{*}(B_i^{\flat})\) and hence \((f^{-1})^{*}(\star\beta_i) = B_i^{\flat}\). Statement (2) holds because \(\star\) is a pointwise isometry and \(f^{*}\) preserves pointwise norms. The \(L^{2}\) identity then follows from the definition of an \({L^{2}}\) norm and the invariance of the integral.

For (3), Cartan's formula and \(d\varpi = 0\) give \(d\,i_{B_i}\varpi = \mathcal{L}_{B_i}\varpi = (\nabla\cdot B_i)\,\varpi\) on \(f(\ST_i)\). Then pulling back by \(f\) and using (1) yields the claim.

For (4), write \(\mathcal{J}\colon\partial\ST_i\hookrightarrow\ST_i\) and \(\mathcal{J}_f\colon\partial f(\ST_i)\hookrightarrow f(\ST_i)\) for the two inclusions, so that \(f\circ\mathcal{J} = \mathcal{J}_f\circ f|_{\partial\ST_i}\) and hence \(t\beta_i = (f|_{\partial\ST_i})^{*}\big(\mathcal{J}_f^{*}(i_{B_i}\varpi)\big)\). Decompose \(B_i = (B_i\cdot N)N + B_i^{\mathsf{T}}\) along \(\partial f(\ST_i)\), where \(B_i^{\mathsf{T}}\) is tangential to the boundary. Any three vectors tangent to a surface in a \(3\)-manifold are linearly dependent, so \(i_{B_i^{\mathsf{T}}}\varpi\) vanishes upon pullback to the surface, while \(\mathcal{J}_f^{*}(i_{N}\varpi)\) is the induced area form \(dS\). The second assertion is the change of variables formula, 
\[\int_{T} f^{*}(i_{B_i}\varpi) = \int_{f(T)} i_{B_i}\varpi \, . \]
Statement (5) is the definition of the curl on an oriented Riemannian \(3\)-manifold, \((\nabla\times A_i)^{\flat} \coloneqq \star\,d(A_i^{\flat})\), combined with \(\star\star = \mathrm{id}\) and with \(f^{*}(A_i^{\flat}) = a_i\), which is the definition of \(A_i\) above. The stated equivalence then follows from (1). Statement (6) is (5) applied to the field \(B_i\), using that \(\star\beta_i = f^{*}(B_i^{\flat})\) from (1). For (7), for \(1\)-forms \(\omega_1\wedge\star\omega_2 = \langle\omega_1,\omega_2\rangle\,\varpi\), so that \(A_i^{\flat}\wedge i_{B_i}\varpi = (A_i\cdot B_i)\,\varpi\) on \(f(\ST_i)\), then pulling back gives the claim.

For (8), both sides of \(u^{\flat}\wedge v^{\flat} = \star\big((u\times v)^{\flat}\big)\) are bilinear and alternating in \((u,v)\), so it suffices to check it on the standard basis of \(\R^{3}\), where \(dx^{1}\wedge dx^{2} = \star\,dx^{3}\) and its cyclic permutations. The second identity is naturality of the exterior derivative. For the final assertion, \(\delta = (-1)^{k}\star d\star\) on \(k\)-forms gives \(\delta\beta_i = \star\,d\star\beta_i\), and by (6) and \(\star\star=\mathrm{id}\) this equals \(f^{*}\big((\nabla\times B_i)^{\flat}\big)\). Hence, using (1) and then (8),
\[
  \delta\beta_i\wedge\star\beta_i
  = f^{*}\Big((\nabla\times B_i)^{\flat}\wedge B_i^{\flat}\Big)
  = f^{*}\Big(\star\big(((\nabla\times B_i)\times B_i)^{\flat}\big)\Big) ,
\]
while \(\star\,d\tilde p_i = \star\,f^{*}(dp_i) = f^{*}\big(\star\,(\nabla p_i)^{\flat}\big)\). As \(f^{*}\) and \(\flat\) are injective the two are equal precisely when \((\nabla\times B_i)\times B_i = \nabla p_i\).
\end{proof}


\subsection{Energy and Equilibrium Equations}

By Lemma \ref{lem:vector-dictionary}(2), the magnetic energy of a configuration is
\begin{equation}\label{eq:energy-vector}
  E(f,\beta_{1},\dots,\beta_{n})
  =\sum_{i=1}^{n}\Big(\ifrac12\big\|\beta_{i}\big\|_{L^{2}}^{2}
    -\tilde p_i\,\myabs{f(\ST_i)}\Big)
  =\sum_{i=1}^{n}\left(\int_{f(\ST_i)}\ifrac12\|B_i\|^{2}\, \varpi 
    -p_i\,\myabs{f(\ST_i)}\right) \, . 
\end{equation}
By Lemma \ref{lem:vector-dictionary}(2) and (6), the MRxMHD equilibrium equations \eqref{prob:MRxMHD_forms} read
\begin{equation}\label{eq:MRxMHD-vector}
  \nabla\times B_i=\mu_i B_i \, , 
  \qquad
  \Big\llbracket\,\ifrac12\|B\|^{2}+p\,\Big\rrbracket\Big|_{f(I_q)}=0 \, .
\end{equation}

\subsection{Helicity}

Let \(\beta\in\mathcal{C}(\ST)\) and let \(a\in G\) be an admissible potential of \(\beta\). Writing \(A=\big((f^{-1})^{*}a\big)^{\sharp}\), a vector field on \(f(\ST)\), Lemma \ref{lem:vector-dictionary}(5) shows that \(A\) is a vector potential for \(B\) in the usual sense, \(B=\nabla\times A\) on \(f(\ST)\), and Lemma \ref{lem:vector-dictionary}(7) gives
\begin{equation}\label{eq:helicity-vector}
  H_{G}(\beta,\ST)=\int_{\ST}a\wedge\beta=\int_{f(\ST)}A\cdot B\, \varpi  \, . 
\end{equation}

\subsection{Loop-corrected Helicity}

In vector notation the classes of an Alexander basis are represented by closed curves in \(\partial\ST_i\), for which we retain the symbols \(\sigma_{i,j}\) and \(\tau_{i,j}\) and the images \(f(\sigma_{i,j})\) and \(f(\tau_{i,j})\) are closed curves in \(\partial f(\ST_i)\). The curves \(\sigma_{i,j}\) with a corresponding direction generate the first homology of \(\ST_i\) and bound only outside \(\ST_i\), whereas each \(\tau_{i,j}\) bounds the surface \(T_{i,j}\subseteq\ST_i\) carrying the flux \(\psi_{i,j}\) (Theorem \ref{thr:Alexander}).

On a hollow torus the \(f(\sigma_{i,j})\) are a poloidal and a toroidal loop on \(\partial f(\ST_i)\), while \(f(T_{i,2})\) is a toroidal ribbon and \(f(T_{i,1})\) a poloidal cross-section, so that \(\psi_{i,2}\) and \(\psi_{i,1}\) are the poloidal and toroidal fluxes respectively (Appendix \ref{append:example}).

Since \(\tau_{i,j}=\partial T_{i,j}\), Stokes' theorem gives
\begin{equation}\label{eq:circulation-is-flux}
  \oint_{f(\tau_{i,j})}A_i\cdot d\ell=\int_{f(T_{i,j})}B_i\cdot dS=\psi_{i,j},
\end{equation}
for every vector potential \(A_i\) of \(B_i\) on \(f(\ST_i)\), so the loop-corrected helicity of Subsection \ref{subsection:Relative Helicity} becomes
\begin{align} \label{eq:relhel-vector}
  \mathscr{H}\big(B_i,f(\ST_i)\big)
  = \int_{f(\ST_i)} A_i\cdot B_i \,  \varpi 
    - \sum_{j=1}^{\ell_i}\psi_{i,j}\oint_{f(\sigma_{i,j})} A_i\cdot d\ell \,  .
\end{align}
Here \(\mathscr{H}(B_i,f(\ST_i))\) denotes \(\mathscr{H}((f^{-1})^{*}\beta_i,  \allowbreak  f(\ST_i))\) of Lemma \ref{Lemma:well-pose}, written in terms of the field \(B_i\) rather than the \(2\)-form representing it.

Equation \eqref{eq:relhel-vector} makes the gauge behaviour transparent. Two vector potentials for the same \(B_i\) on \(f(\ST_i)\) differ by a curl-free field \(S\) on \(f(\ST_i)\), and the computation of remark \ref{remark:helicity} states that
\begin{equation}\label{eq:S-dot-B}
  \int_{f(\ST_i)}S\cdot B_i\, \varpi 
  =\sum_{j=1}^{\ell_i}\psi_{i,j}\oint_{f(\sigma_{i,j})}S\cdot d\ell .
\end{equation}
Both terms on the right of \eqref{eq:relhel-vector} therefore change by the same amount under \(A_i\mapsto A_i+S\), and \(\mathscr{H}\) is unchanged (Lemma \ref{Lemma:well-pose} in vector notation).



\subsection{Amperian Gauge}

In vector notation the Amperian gauge \(\mathscr{A}(\ST_i)\) of Subsection \ref{subsection:Relative Helicity} consists of those vector potentials on \(f(\ST_i)\) whose circulation around every loop \(f(\sigma_{i,j})\) vanishes,
\[
  \oint_{f(\sigma_{i,j})}A_i\cdot d\ell=0,\qquad j=1,\dots,\ell_i .
\]
The Amperian gauge fixes these circulations to zero while retaining the freedom to add a gradient. Lemma \ref{Lemma:gauge_exist} asserts that such a potential always exists, and imposing these conditions in \eqref{eq:relhel-vector} shows that the loop-corrected helicity reduces to helicity in this gauge.

\subsection{The Equilibrium and Gauge Conditions}

In vector notation the conclusion of Theorem \ref{prop:vanishing} reads as follows. A configuration at fixed pressures \(p_i\), fixed fluxes \(\psi_{i,j}\) and fixed loop-corrected helicities is a critical point of the magnetic energy \eqref{eq:energy-vector} if and only if
\[
  \nabla\times B_i=\mu_i B_i \, ,\qquad
  \Big\llbracket\,\ifrac12\|B\|^{2}+p\,\Big\rrbracket\Big|_{f(I_q)}=0 \, ,
\]
the multipliers \(\mu_i\) being those of the loop-corrected helicity constraints.

Condition \eqref{eq:gauge_condition} of Theorem \ref{thr:one_gauge}, which appears in the first variation of \(\|\beta_{i,t}\|^{2}_{L^{2}}\) up to a factor \(\mu_i\), is
\[
  \mu_i\sum_{j=1}^{\ell_i}\psi_{i,j}\,
  \frac{d}{dt}\bigg|_{t=0}\oint_{f_t(\sigma_{i,j})}A_{i,t}\cdot d\ell \, ,
  \qquad
  A_{i,t}\coloneqq\big((f_t^{-1})^{*}a_{i,t}\big)^{\sharp} ,
\]
where \(A_{i,t}\) is a vector field on \(f_t(\ST_i)\) and \(d\ell\) is the Euclidean length element there.

\end{document}